\documentclass[11pt]{article}
\usepackage{bbm}
\usepackage[ruled, linesnumbered]{algorithm2e}
\usepackage[margin=1in]{geometry}
\usepackage{amsmath,amsfonts,amssymb,amsthm}
\usepackage{mathtools}
\usepackage{graphicx}
\usepackage{float}
\usepackage{enumerate}
\usepackage{verbatim,nameref}
\usepackage{hyperref}
\usepackage[capitalize,nameinlink]{cleveref}
\usepackage[dvipsnames]{xcolor}
\hypersetup{
	colorlinks=true,
	pdfpagemode=UseNone,
    citecolor=OliveGreen,
    linkcolor=NavyBlue,
    urlcolor=Magenta,
	pdfstartview=FitW
}
\usepackage{tikz}
\usepackage{appendix}

\crefname{appsec}{Appendix}{Appendices}
\crefname{properties}{Parameter Properties}{Parameter Properties}
\crefname{property}{Property}{Properties}
\crefname{probabilities}{Parameter Setting}{Parameter Settings}
\crefname{event}{Event}{Events}

\usepackage{enumitem}
\usepackage{titlesec}
\usepackage{upgreek}

\theoremstyle{plain}
\newtheorem{theorem}{Theorem}[section]

\newtheorem{lemma}[theorem]{Lemma}
\newtheorem{corollary}[theorem]{Corollary}

\newtheorem{claim}[theorem]{Claim}

\newtheorem{flipprob}[theorem]{Parameter Setting}
\newtheorem{probprops}[theorem]{Parameter Properties}

\theoremstyle{definition}
\newtheorem{definition}[theorem]{Definition}

\newtheorem*{assumption*}{Assumption}

\theoremstyle{remark}
\newtheorem{remark}[theorem]{Remark}

\crefname{lemma}{Lemma}{Lemmas}
\crefname{theorem}{Theorem}{Theorems}
\crefname{corollary}{Corollary}{Corollaries}
\crefname{definition}{Definition}{Definitions}
\crefname{remark}{Remark}{Remarks}
\crefname{fact}{Fact}{Facts}
\crefname{claim}{Claim}{Claims}
\crefname{proposition}{Proposition}{Propositions}

\newcommand{\nmet}{\widetilde{H}}

\newcommand{\poly}{\mathrm{poly}}

\newcommand{\fpras}{\mathsf{FPRAS}}

\newcommand{\fptas}{\mathsf{FPTAS}}

\newcommand{\eps}{\varepsilon}
\newcommand{\indicator}{\mathbf{1}}

\newcommand{\EE}{\mathcal{E}}

\newcommand{\Tmix}{T_{\mathrm{mix}}}

\newcommand{\dtv}{d_{\mathsf{TV}}}

\def\Prob#1{{\mathbf{Pr}\left({#1}\right)}}
\def\Exp#1{{\mathbf{E}\left[{#1}\right]}}

\def\ExpCond#1#2{{\mathbf{E}\left({#1} \mid {#2} \right)}}

\def\unblocked#1#2#3{F^0(#1,#2,#3)}
\def\unblockedShort#1#2{F^0_{#1}{(#2)}}

\def\singleBlockedShort#1#2{F^1_{#1}(#2)}

\def\multiBlockedShort#1#2{F^{2}_{#1}(#2)}

\newcommand{\degree}[3]{d^{#1}_{#2}(#3)}

\def\lambbound{1.809}
\def\deltabound{125}

\def\colorList#1{{L}(#1)}

\newboolean{doubleblind}
\setboolean{doubleblind}{false}

\title{Flip Dynamics for Sampling Colorings: Improving $(11/6-\varepsilon)$ Using a Simple Metric}

\ifdoubleblind
\author{Author(s)}
\else
\author{Charlie Carlson\thanks{Department of Computer Science, University at Buffalo, Buffalo. Email: cc387@buffalo.edu.}  \and Eric Vigoda\thanks{Department of Computer Science, University of California, Santa Barbara. Email: vigoda@ucsb.edu.  Research supported in part by NSF grant CCF-2147094.}}
\fi
\date{}

\begin{document}
\maketitle

\begin{abstract} \small\baselineskip=9pt
    We present improved bounds for randomly sampling $k$-colorings of graphs with maximum degree $\Delta$; our results hold without any further structural assumptions on the graph.  The Glauber dynamics is a simple single-site update Markov chain.  Jerrum (1995) proved an optimal $O(n\log{n})$ mixing-time bound for Glauber dynamics whenever $k>2\Delta$ where $\Delta$ is the maximum degree of the input graph.  This bound was improved by Vigoda (1999) to $k>(11/6)\Delta$ using a ``flip'' dynamics which recolors (small) maximal two-colored components in each step.   Vigoda's result was the best known for general graphs for 20 years until Chen et al. (2019) established optimal mixing of the flip dynamics for $k>(11/6-\eps)\Delta$ where $\eps\approx 10^{-5}$.  We present the first substantial improvement over these results.  We prove an optimal mixing-time bound of $O(n\log{n})$ for the flip dynamics when $\Delta\geq\deltabound$ and $k\geq\lambbound\Delta$.
    This yields, through recent spectral independence results, an optimal $O(n\log{n})$ mixing time for the Glauber dynamics for every fixed $\Delta\geq\deltabound$ in the same range of $k/\Delta$.
    Our proof utilizes path coupling with a simple weighted Hamming distance for ``unblocked'' neighbors.
\end{abstract}

\section{Introduction}
A problem of great interest and considerable study at the intersection of theoretical computer science, discrete mathematics, and statistical physics is the random sampling of $k$-colorings of a given input graph $G$.
Given a graph $G=(V,E)$ of maximum degree $\Delta$ and an integer $k\geq 2$, let~$\Omega=[k]^V$ denote the collection of  $k$-labelings of $G$ and let~$\Omega^* \subseteq \Omega$ denote the collection of proper vertex $k$-colorings. That is~$\Omega$ is the collection of assignments $\sigma:V\rightarrow \{1,\dots,k\}$ and~$\Omega^*$ is the subcollection of those assignments $\sigma$ where for all $(v,w)\in E$, $\sigma(v)\neq \sigma(w)$.
Let $\pi$ denote the uniform distribution over~$\Omega^*$.  The coloring problem is a natural example of a non-binary graphical model~\cite{Murphy, KFbook} and, in statistical physics, it is the zero-temperature limit of the antiferromagnetic Potts model~\cite{SalasSokal}.

We study algorithms for the approximate counting problem of estimating $\lvert \Omega^* \rvert$, the number of proper $k$-colorings, and the approximate sampling problem of generating random proper $k$-colorings from a nearly uniform distribution.
In particular, given a graph $G=(V, E)$ and an $\eps>0$, for the sampling problem, our goal is to sample from a distribution $\mu$ which is within total variation distance $\leq\eps$ of the uniform distribution $\pi$ in time polynomial in $n=|V|$ and $\log(1/\eps)$.  In the approximate counting problem, a graph $G=(V, E)$, an $\eps>0$, and a $\delta>0$ are given, and the goal is to obtain a $\fpras$ to estimate $\lvert\Omega^*\rvert$, which is an algorithm to estimate $\lvert\Omega^*\rvert$ within a $(1\pm\eps)$ multiplicative factor with probability $\geq 1-\delta$ in time $\poly(n,1/\eps,\log(1/\delta))$.  These approximate sampling/counting
problems are polynomial-time inter-reducible to each other.  Relevant for our work, an $O(n\log(n/\eps))$ sampling algorithm yields an $O^*(n^2)$ time approximate counting algorithm~\cite{SVV:annealing,Huber,Kolmogorov}.

The Markov chain Monte Carlo (MCMC) method is a natural algorithmic approach to approximate sampling. The Glauber dynamics (also known as the Gibbs sampler) is the simplest example of the MCMC method.  The Glauber dynamics is a Markov chain on the collection of proper $k$-colorings in which the transitions update the color at a randomly chosen vertex in each step as follows.  From $X_t\in\Omega^*$, choose a vertex $v\in V$ and a color $c\in\{1,\dots,k\}$ uniformly at random.  If no neighbor of~$v$ has color $c$ in $X_t$ then we recolor $v$ to color $c$ in $X_{t+1}$ and all other vertices maintain the same color $X_{t+1}(w)=X_t(w)$ for all $w\neq v$; and if color $c$ is not available for $v$ then we set $X_{t+1}=X_t$.
If $k\geq\Delta+2$, then the Glauber dynamics is ergodic and the unique stationary distribution is the uniform distribution~$\pi$ over $\Omega^*$.

The mixing time $\Tmix$ is the number of steps from the worst initial state $X_0$ to guarantee that the total variation distance from the stationary distribution $\pi$ is $\leq 1/4$.  A mixing time of $O(n\log{n})$ is referred to as an optimal mixing time as this matches the lower bound established by Hayes and Sinclair~\cite{HayesSinclair} for any graph of constant maximum degree~$\Delta$.

Jerrum~\cite{Jerrum} (see also Salas and Sokal~\cite{SalasSokal}) established an optimal mixing time of $O(n\log{n})$ for the Glauber dynamics whenever $k>2\Delta$.  This was a seminal result in the development of coupling techniques, including the path coupling method of Bubley and Dyer~\cite{BubleyDyer}.  Vigoda~\cite{Vigoda} improved Jerrum's result to $k>(11/6)\Delta$ by proving $O(n\log{n})$ mixing time of the following {\em flip dynamics}, which implied $O(n^2)$ mixing time of the Glauber dynamics.

The flip dynamics is a generalization of the Glauber dynamics, which recolors maximal two-colored components in each step.  For a proper $k$-coloring $X_t$, vertex $v\in V$, and color $c\in\{1,\dots,k\}$, let
$S_{X_t}(v,c)$ denote the set of vertices $w$ which have an alternating path on colors $(X_t(v),c)$ between~$v$ and~$w$; we refer to the set $S_{X_t}(v,c)$ as a cluster. When $c=X_t(v)$, we set $S_{X_t}(v,c)=\emptyset$. The flip dynamics is defined by a set of parameters $(P_i)_{i\geq 1}$ for the flip probabilities.  The dynamics operates by choosing a random vertex $v$ and color $c$, and then, when $\ell\geq 1$, flipping the cluster $S_{X_t}(v,c)$ by interchanging the colors $X_t(v)$ and $c$ on the chosen cluster with probability $P_\ell/\ell$
where $\ell=|S_{X_t}(v,c)|$.

More formally, for $X_t\in\Omega^*$, the transitions $X_t\rightarrow X_{t+1}$ for the flip dynamics, with flip probabilities $(P_i)_{i\geq 1}$, are defined as follows:
\begin{itemize}
    \item Choose $v\in V$ uniformly at random.
    \item Choose $c \in [k]=\{1,\dots,k\}$ uniformly at random.
    \item Let $S=S_{X_t}(v,c)$ be the cluster in $X_t$ defined by colors $\{X_t(v),c\}$.  Let $\ell=|S|$ denote the size of the cluster.
    \item If $\ell\geq 1$ then with probability $P_\ell/\ell$,
    let $X_{t+1}$ denote the proper coloring obtained by interchanging colors $X_t(v)$ and~$c$ in $S$, and for all $w\notin S$ set $X_{t+1}(w)=X_t(w)$.
    \item If the cluster is not flipped, set $X_{t+1}=X_t$.
\end{itemize}
Observe that if $P_1=1$ and $P_i=0$ for all $i>1$ then the flip dynamics only recolors the selected vertex if no neighbor has the chosen color,  and is equivalent to the Glauber dynamics.
In Vigoda's original work, the parameters satisfy the basic properties: $P_1=1$, $P_i\geq P_{i+1}$ for all $i\geq 1$, and $P_{j}=0$ for $j\geq 7$.   All subsequent works (including this paper) follow these broad settings but differ in the detailed setting.

Since Vigoda's result, there were myriad improved results for various restricted classes of graphs, including optimal mixing on triangle-free graphs when $k>\alpha^*\Delta$ for $\Delta=O(1)$ where $\alpha^*\approx 1.763$~\cite{CGSV21,FGYZ21,CLV21} (see also~\cite{JPV,LSS,DFHV} for related results), optimal mixing on large girth graphs when $k\geq\Delta+3$~\cite{CLMM}, and further improvements for trees~\cite{MSW04}, planar graphs~\cite{HVV}, and sparse random graphs~\cite{EHSV}.

The first improvement on Vigoda's result for general graphs was 20 years later by Chen, Delcourt, Moitra, Perarnau, and Postle~\cite{CDMPP19}, who proved an $O(n\log{n})$ mixing-time bound for the flip dynamics when $k>(11/6-\eps_0)\Delta$, where $\eps_0\approx 10^{-5}$ is a fixed positive constant.
Their result (as well as Vigoda's result~\cite{Vigoda}) was obtained for a specific setting of the flip parameters.  We present the first substantial improvement over Vigoda's result by obtaining optimal mixing of the flip dynamics when $\Delta\geq\deltabound$ and $k\geq \lambbound\Delta$.

\begin{theorem}
\label[theorem]{thrm:flip-mixing}
Let $\Delta\geq\deltabound$ and $k\geq\lambbound\Delta$.
There exists a setting of the parameters for the flip dynamics with $P_j=0$ for all $j\geq7$ such that, for every graph $G$ on $n$ vertices with maximum degree at most $\Delta$, the flip dynamics has mixing time $O(n\log n)$.
\end{theorem}

We specify the flip probabilities after defining the metric.
Note that we use a universal setting of the flip probabilities for all $\lambbound\leq k/\Delta$ that differs from those used in \cite{CDMPP19,Vigoda}.

As in \cite{Vigoda,CDMPP19} this implies polynomial mixing of the Glauber dynamics for the same range of $k/\Delta$.  In particular, by comparison of the spectral gaps of the transition matrices, it implies~$O(n^2)$ mixing time of the Glauber dynamics where the hidden constant is polynomial in~$\Delta$ and~$k$.  Moreover, recent work of~\cite{BCCPSV22,Liu21} utilizing spectral independence~\cite{CLV21,ALO}, implies $O(n\log{n})$ mixing time of the Glauber dynamics when~$\Delta$ is constant; the same result held for the previous work of~\cite{CDMPP19} for the corresponding range of parameters $k$ and $\Delta$.

\begin{corollary}
\label[corollary]{cor:glauber-mixing}
Let $\Delta\geq\deltabound$ and $k\geq\lambbound\Delta$.
There exists a constant $C=C(\Delta,k)$ such that, for every graph $G$ on $n$ vertices with maximum degree at most $\Delta$, the mixing time of the Glauber dynamics is at most $Cn\log n$.
\end{corollary}

A recent algorithm of Chen, Feng, Guo, Zhang, and Zou~\cite{CFGZZ} uses our coupling proof to obtain a deterministic fully polynomial-time approximation scheme ($\fptas$) for counting proper $k$-colorings when~$k$ and~$\Delta$ are constant.

\subsection{Proof Overview}
\label{subsec:proof-overview}

Our proof of \cref{thrm:flip-mixing} utilizes a novel distance metric described here at a high level.

Jerrum's bound of $k>2\Delta$ can be proved using path coupling in which we consider a pair of configurations $X_t,Y_t$ that differ at exactly one vertex, say $v^*$.  We then analyze the expected Hamming distance of $X_{t+1},Y_{t+1}$ after one step of a coupled transition $(X_t,Y_t)\rightarrow (X_{t+1},Y_{t+1})$.  The coupling in this setting is fairly simple as it is the identity coupling (i.e., both chains attempt the same (vertex, color) pair $(v,c)$) for the Glauber update except when the updated vertex $v$ is a neighbor~$w$ of $v^*$, in which case we couple trying to recolor $w$ to color $X_t(v^*)$ in one chain with color $Y_t(v^*)$ in the other chain.  Under this coupling, there is at most one coupled transition per neighbor which can increase the Hamming distance by at most one, and this yields the $k>2\Delta$ bound.

Vigoda's result for $k>(11/6)\Delta$ uses the same path coupling framework to analyze the expected Hamming distance for a pair $X_t, Y_t$ that differ at a single vertex $v^*$.  In flip dynamics, the coupling is more complicated than in Jerrum's analysis because of additional moves.

Recall that the transitions of the flip dynamics correspond to flipping maximal alternating two-colored components, where flipping refers to interchanging the two colors.  An alternative and equivalent view of the transitions of the flip dynamics is as follows.  For $X_t\in\Omega^*$, consider the collection of all nonempty clusters; note that there are at most $nk$ clusters.  Choose a cluster $S$ with probability $P_{|S|}/(nk)$ and then flip $S$ to obtain $X_{t+1}$ (with the remaining probability set $X_{t+1}=X_t$).

We give a brief high-level overview of Vigoda's coupling $(X_t,Y_t)\rightarrow (X_{t+1},Y_{t+1})$;
the same coupling is used in \cite{CDMPP19} where it is referred to as the greedy coupling.
Consider a pair of proper $k$-colorings $X_t,Y_t\in\Omega^*$ that differ at a single vertex $v^*$, thus $X_t(w)=Y_t(w)$ for all $w\neq v^*$.
The coupling uses the identity coupling for clusters $S$ that have the same underlying vertex set and unordered color pair in chains $X_t$ and $Y_t$, except for the cluster flips described below. Under an identity pairing, the cluster $S$ is flipped in both chains or in neither chain.

The exceptional cluster flips arise from attempts to recolor $v^*$, or to recolor a neighbor $w$ of $v^*$ to one of the colors $X_t(v^*)$ and $Y_t(v^*)$. For each $c\in[k]$, the class $\mathcal D_c$ contains the flips in the individual chains represented by recoloring $v^*$ to $c$, or recoloring a $c$-colored neighbor to $Y_t(v^*)$ in $X_t$ or to $X_t(v^*)$ in $Y_t$; see \cref{subsec:greedy-coupling}. Flips are coupled within each class, allowing a flip to be paired with a self-loop.

Vigoda demonstrated a choice of the flip probabilities and a coupling so that the expected increase in the Hamming distance is $\leq 5/6$ in an amortized cost per neighbor $w$ of $v^*$, which yields the bound $k>(11/6)\Delta$.  Vigoda~\cite{Vigoda} identified six extremal configurations (these are the configurations that maximize the expected increase in the Hamming distance). Chen et al.~\cite{CDMPP19} presented a slightly different setting of the flip probabilities with only two extremal configurations.
They considered a weighted Hamming distance where for every neighbor $w\in N(v^*)$, if the local configuration around $w$ is different from the two extremal configurations, then the definition of the distance (between $X_t$ and $Y_t$) is decreased by $\eta/\Delta$ for a fixed constant $\eta$ where $1/2>\eta>0$.
Using Vigoda's greedy coupling with this new metric, they established that the expected distance decreases when $k>(11/6-\eps_0)\Delta$ where $\eps_0\approx 10^{-5}$.

We take a complementary approach.  Whereas Chen et al.~\cite{CDMPP19} reweight the worst configurations (or equivalently all non-worst case configurations), we instead consider a particularly ``good'' configuration. Namely, we consider the local configuration where the neighbor $w$ is {\em unblocked}, which means that the colors $X_t(v^*)$ and $Y_t(v^*)$ do not appear in $N(w)\setminus\{v^*\}$.
When an unblocked neighbor $w$ is the only neighbor of its color, the potentially problematic recolorings of $w$ with colors $X_t(v^*)$ or $Y_t(v^*)$ can be coupled with flips of clusters of size 2 (containing~$v^*$ and $w$). Consequently, its contribution to the expected increase in Hamming distance is at most $(1-P_2)/(nk)$, significantly less than $(5/6)/(nk)$ for the flip probabilities considered here.

We choose the flip probabilities to optimize for this new metric, which results in a setting for the flip probabilities that is suboptimal (with respect to Hamming distance) for the extremal configurations considered by Chen et al.~\cite{CDMPP19}.
However, for our choice of distance metric, these extremal configurations improve as they have a reasonable probability of moving to the unblocked configurations, yielding a decrease in the distance.

An essential aspect of our proof is that if a neighbor $w$ is unblocked, then it may have a considerable probability of becoming blocked (which increases the distance with respect to our metric); however, a new disagreement at $w$ can have unblocked neighbors and hence a smaller weight in our metric. This is the crucial trade-off in our argument: we bound the expected loss of the $\eta/\Delta$ discount contributed by $w$ at $v^*$ minus the expected discount attached to a new disagreement at $w$.
Our overall proof is of a similar technical level of difficulty as in~\cite{CDMPP19} but yields a substantially improved bound of $k \geq \lambbound \Delta$ when $\Delta\geq\deltabound$.

We first present the basic definitions, the path coupling theorem, and Vigoda's greedy coupling. We then define our metric, describe the extremal configurations, and analyze the metric to prove~\cref{thrm:flip-mixing}.

\section{Preliminaries}
\label{sec:preliminaries}

In this section we detail the basic definitions and concepts that are required background for our proofs.

For a graph $G = (V,E)$ and a vertex $v\in V$, let $N(v)$ denote the neighbors of~$v$, and let $d(v)=|N(v)|$ denote the degree of $v$ in $G$.
Let $\Delta\geq\max_{v\in V}d(v)$ be an upper bound on the maximum degree.

\subsection{List Colorings}
\label{subsec:list-colorings}

We prove our results in the more general context of  list labelings and proper list colorings.
For any positive integer $k$, let $[k]=\{1,\dots,k\}$.
Fix a graph $G = (V,E)$,
an integer $k\geq 2$, and
a list of colors $\colorList{v}\subseteq [k]$ for each $v \in V$.
A {\em list labeling} of $G$ is a function $\sigma$ that maps each vertex $v \in V$ to a color $\sigma(v)\in\colorList{v}$.
A list labeling is called a {\em proper list coloring} if for all $(u,v) \in E$, $\sigma(u) \neq \sigma(v)$.
Observe that if $\colorList{v} = [k]$ for all $v \in V$, then a list labeling is a $k$-labeling and a proper list coloring is a proper $k$-coloring.

Let $\Omega=\{\sigma\in [k]^V:\forall v\in V, \sigma(v)\in \colorList{v}\}$ be the collection of all list labelings of $G$, and
let $\Omega^*\subseteq\Omega$ be the collection of proper list colorings of $G$, i.e.,
\begin{equation*} \Omega^*=\{\sigma\in [k]^V:\forall v\in V, \sigma(v)\in \colorList{v}; \forall \{v,w\}\in E, \sigma(v)\neq\sigma(w)\}.
\end{equation*}

For a pair $X_t,Y_t\in\Omega$, let $H(X_t,Y_t)=|\{v\in V: X_t(v)\neq Y_t(v)\}|$ denote the Hamming distance between $X_t$ and $Y_t$.
For $v\in V$, let $\Omega^2_{v}\subseteq\Omega^2$ denote the pairs $(X_t,Y_t)\in\Omega^2$ where $X_t(v) \neq Y_t(v)$ and $X_t(u) = Y_t(u)$ for all $u \neq v$; thus, $\Omega^2_v$ is the set of pairs of list labelings that only differ at vertex~$v$.
Let $\Sigma := \bigcup_{v \in V} \Omega^2_{v}$ be the set of all pairs of list labelings that differ at exactly one vertex.
We will refer to pairs $(\sigma,\tau)\in\Sigma$ as neighboring list labelings.

\subsection{Mixing Time and Path Coupling}
\label{subsec:path-coupling}

Our interest is the mixing time $\Tmix$, which measures the speed of convergence to the unique stationary distribution $\pi$ from the worst initial state $X_0\in\Omega$.  For a pair of distributions $\mu,\pi$ on $\Omega$, the total variation distance is defined as $\dtv(\mu,\pi)=\frac{1}{2}\sum_{X\in\Omega} |\mu(X)-\pi(X)|$.  For $\eps>0$, define the mixing time by \begin{equation*} \Tmix(\eps)=\max_{X_0\in\Omega}\min\{t: \dtv(P^t(X_0,\cdot),\pi)\leq\eps\},
\end{equation*}
where $P$ is the transition matrix and $\pi$ is its stationary distribution.
We will often refer to $\Tmix=\Tmix(1/4)$ as the mixing time since $\Tmix(\eps)\leq \Tmix\times \lceil\log_2(1/\eps)\rceil$ for $0<\eps<1$.

We will utilize the coupling method to upper-bound the mixing time. For a pair of states $X_t, Y_t\in\Omega$, a coupling for the flip dynamics is a joint evolution $(X_t,Y_t)\rightarrow (X_{t+1},Y_{t+1})$ such that each individual transition $(X_t\rightarrow X_{t+1})$ and $(Y_t\rightarrow Y_{t+1})$ follows the flip dynamics. See~\cite{Jerrum:notes} for a more detailed introduction.

We will bound the mixing time using the path coupling framework of Bubley and Dyer~\cite{BubleyDyer,DyerGreenhill}; we use the following simplified form which is sufficient for our purposes.

\begin{theorem}
\label[theorem]{thrm:path-coupling}
Consider a Markov chain with transition matrix~$P$, state space~$\Omega$, and unique stationary distribution~$\pi$.
Let $\Sigma\subseteq\Omega^2$ be a set of pairs such that every two states are joined by a path of at most $n$ edges in $(\Omega,\Sigma)$, and let $w(\sigma,\tau)$ be defined for every $(\sigma,\tau)\in\Sigma$.
Assume that there exists a constant $1/2\leq C\leq1$ such that $w(\sigma,\tau)\in[C,1]$ for every $(\sigma,\tau)\in\Sigma$.
For arbitrary $\sigma,\tau\in\Omega$, extend $w$ by the length of the shortest path in $(\Omega,\Sigma)$ when each edge $(\rho,\rho')\in\Sigma$ has weight $w(\rho,\rho')$.

Suppose that there exists $\delta>0$ such that, for every $(\sigma,\tau)\in\Sigma$, there is a coupling $(\sigma,\tau)\rightarrow(\sigma',\tau')$ satisfying
\begin{equation*}
    \Exp{w(\sigma',\tau')}\leq(1-\delta)w(\sigma,\tau).
\end{equation*}
Then
\begin{equation*}
    \Tmix(\eps)=O\left(\frac{\log(n/\eps)}{\delta}\right).
\end{equation*}
\end{theorem}

The restriction $C\geq 1/2$ is sufficient so that the shortest path metric $w(\cdot,\cdot)$ defined using $(\Omega,\Sigma)$ has the property that for all $(\sigma,\tau)\in\Sigma$ a shortest path from $\sigma$ to $\tau$ is the single edge $(\sigma,\tau)$. Thus, the assigned weight of a neighboring pair equals its induced shortest-path distance; this is the premetric condition used in \cite{CDMPP19}.

\subsection{Flip Dynamics}
\label{subsec:flip-dynamics}

For a list labeling $\sigma\in\Omega$, vertex $v\in V$, and color $c\in [k]$,
if $c\neq\sigma(v)$, let
\begin{equation*} \widehat{S}_{\sigma}(v,c)=\{w\in V: \mbox{there is a $(\sigma(v),c)$-alternating path between $v$ and $w$ in $\sigma$}\}.
\end{equation*}
We allow alternating paths of length zero, so $v\in\widehat S_\sigma(v,c)$ when $c\neq\sigma(v)$.
If $c=\sigma(v)$, set $\widehat S_\sigma(v,c)=S_\sigma(v,c)=\emptyset$.
When $c\neq\sigma(v)$, the labeling $\sigma$ need not be a proper list coloring, but we still require that the colors alternate along the path; in particular, a monochromatic edge does not extend such a path.

If $c\neq\sigma(v)$ and for all $w\in \widehat{S}_{\sigma}(v,c)$ it holds that $\{c,\sigma(v)\}\subseteq \colorList{w}$, then
we set $S_{\sigma}(v,c)=\widehat{S}_{\sigma}(v,c)$, and otherwise, we set $S_{\sigma}(v,c)=\emptyset$.
We refer to $S = S_{\sigma}(v,c)$ as a {\em cluster} in $\sigma$
and the operation of
interchanging colors $\sigma(v)$ and $c$ on the set~$S$ as {\em flipping} $S$ in $\sigma$.
In particular, $S_{\sigma}(v,\sigma(v))=\emptyset$, since interchanging a color with itself leaves the labeling unchanged.
Note that, when $c\neq\sigma(v)$, if the flip of $\widehat{S}_{\sigma}(v,c)$ is invalid (namely, for some $w\in \widehat{S}_{\sigma}(v,c)$ the new color is not in its list $\colorList{w}$) then we have set the cluster ${S}_{\sigma}(v,c) = \emptyset$. Hence, all nonempty clusters $S_{\sigma}(v,c)$ can be flipped in $\sigma$. For convenience, empty clusters will still be considered in the upcoming coupling definition but they will be flipped with probability zero.

It is important to stress that a cluster $S_{\sigma}(v,c)$ is identified
by its underlying vertex set together with the unordered color pair
$\{c,\sigma(v)\}$. Thus, two clusters with the same underlying
vertex set are distinct if their color pairs differ. Conversely,
two representations with the same underlying vertex set and the
same color pair represent the same cluster.
Whenever we use set notation such as $|S|$, $w\in S$, a union, or an intersection, it refers to the underlying vertex set of the cluster. In particular, an expression of the form $S=\emptyset$ or $S=U$ for a displayed vertex set $U$ is shorthand for equality of underlying vertex sets. When two clusters are compared as clusters, equality requires both the same underlying vertex set and the same unordered color pair.

\begin{figure}[ht]
    \centering
    \includegraphics[width=.78\textwidth]{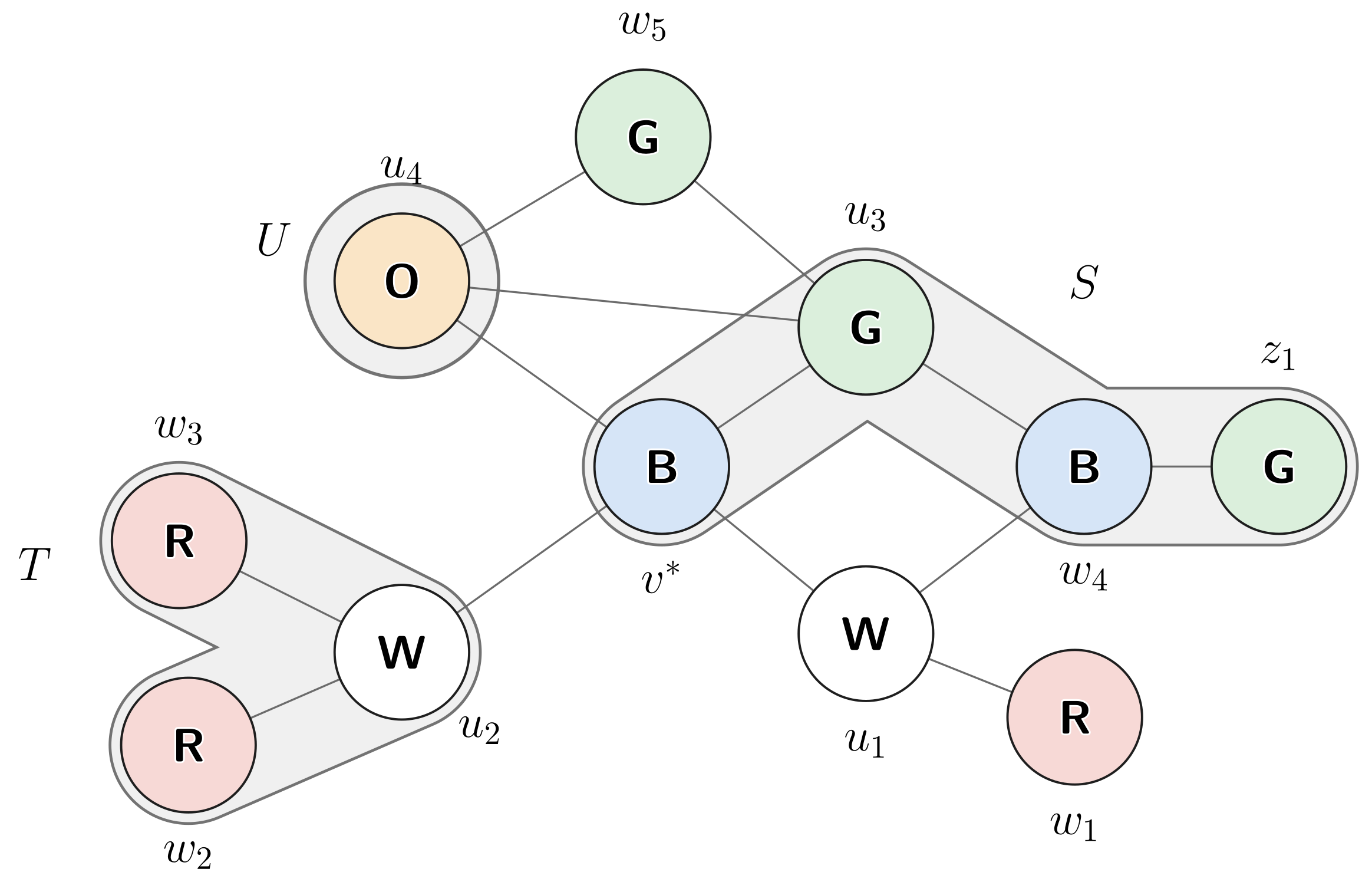}
    \caption{Examples of clusters in a list labeling. The shaded region $S$ is the $(B,G)$-cluster containing $v^*$; it contains $v^*,u_3,w_4,$ and $z_1$, but does not contain $w_5$ because the edge from $u_3$ to $w_5$ is not alternating. The shaded region $T$ is the $(W,R)$-cluster containing $u_2$; it does not contain $u_1$ or $w_1$ because there is no edge joining the two white-red components. Finally, the shaded region $U$ is the $(O,R)$-cluster containing $u_4$.}
    \label{fig:cluster-examples}
\end{figure}

For $X_t\in\Omega$ and a cluster $S=S_{X_t}(v,c)$, let $X_{t+1}$ be the list labeling obtained by flipping $S$ in $X_t$.
If $X_t\in\Omega^*$ (i.e., $X_t$ is a proper list coloring), then $X_{t+1}\in \Omega^*$; moreover in this case, if we flip $S'=S_{X_{t+1}}(v,X_t(v))$ in $X_{t+1}$ then we obtain $X_t$, and hence the flip operation is symmetric on $\Omega^*$.

The following generalizes the flip dynamics to list labelings. Consider probabilities $(P_i)_{i\geq 1}$ where $0\leq P_i\leq 1$ for all $i\geq 1$.
For $X_t\in\Omega$, the transition $X_t\rightarrow X_{t+1}$ of the flip dynamics is defined as follows:
\begin{itemize}
    \item Choose $v\in V$ and $c\in [k]$ uniformly at random.
    \item Let $S=S_{X_t}(v,c)$ and let $\ell=|S|$ denote the size of the cluster.
    \item If $\ell \geq 1$ then with probability $P_\ell/\ell$,
    let $X_{t+1}$ denote the list labeling obtained by interchanging colors $X_t(v)$ and~$c$ in $S$ (and for $w\notin S$ set $X_{t+1}(w)=X_t(w)$).
    \item If the cluster is not flipped, let $X_{t+1}=X_t$.
\end{itemize}
Observe that every nonempty cluster $S$ of size $\ell$ has exactly $\ell$ representations of the form $S_{X_t}(v,c)$. Consequently, the flip dynamics flips $S$ with probability $P_\ell/(nk)$.

Moreover, suppose that $P_1>0$ and $|\colorList{v}|\geq d(v)+2$ for every $v\in V$.
Then the unique stationary distribution of the flip dynamics is the uniform distribution $\pi$ over $\Omega^*$, which is the set of proper list colorings; note that the states in $\Omega\setminus\Omega^*$ have zero probability in the stationary distribution, and every list labeling has a path to a proper list coloring.

\subsection{Greedy Coupling}
\label{subsec:greedy-coupling}

The greedy coupling below follows Vigoda~\cite{Vigoda} and Chen et al.~\cite{CDMPP19}; we give the details needed for list labelings.

In the following coupling definition we assume the flip probabilities satisfy the following basic properties:
\begin{equation*} P_0=0, \qquad P_1=1, \qquad \mbox{ and } \qquad P_i\geq P_{i+1} \ \mbox{ for all } i\geq 1.
\end{equation*}
Later in the coupling analysis, we will assume further properties of the flip parameters; however, the above basic properties suffice to ensure that the following coupling is well-defined.

For the purposes of the coupling definition, consider the following alternative formulation of the flip dynamics.
For a state $X_t\in\Omega$, every cluster $S$ in $X_t$ has an associated flip probability $P_{|S|}$.  From a state $X_t\in\Omega$, choose a cluster $S$ with probability $P_{|S|}/(nk)$
and flip it to obtain the new state $X_{t+1}$, and with the remaining probability stay in the same state $X_{t+1}=X_t$.

We now define the greedy coupling.
If a cluster $S$ exists in both chains and is not among the potentially disagreeing clusters described below, the greedy coupling will flip $S$ in both chains or in neither chain; this is referred to as an identity pairing.  All non-identity pairings are confined to the cluster flips described below, which either differ between the two chains or may change the Hamming distance.  Such a cluster must include $v^*$ or a neighbor of~$v^*$; these clusters can be grouped according to a color $c$ in the following manner.

We partition the neighbors of a vertex $z \in V$ on a per-color basis.
For any list labeling $X_t\in\Omega$, color~$c \in [k]$,
and vertex $z\in V$, let
\begin{equation*}
N_c(z) = N_c(z,X_t) := \{ w \in N(z):  X_t(w) = c \}
\end{equation*}
be the neighbors of $z$ that are color $c$ in $X_t$.
Note that if $(X_t, Y_t) \in \Omega^2_{v^*}$ then $N_{c}(v^*,X_t)=N_{c}(v^*,Y_t)$.
In this case, let
\begin{equation*}
d_c(v^*) := \lvert N_{c}(v^*)\rvert
\end{equation*}
be the number of neighbors of $v^*$ that are colored~$c$ in $X_t$, or equivalently in $Y_t$; we call $d_c(v^*)$ the \emph{color degree} of $c$ at $v^*$.

We now begin the formal definition of the greedy coupling for $(X_t,Y_t)\in\Sigma$. The construction below defers a color at the disagreement vertex only when that color also appears on one of its neighbors; these cases are handled in \cref{sec:degenerate-colors}.
Fix $v^*\in V$ and $(X_t,Y_t)\in\Omega^2_{v^*}$.
Fix once and for all a total ordering $\prec$ of $[k]$.
If necessary, interchange the names of the two chains so that $X_t(v^*)\prec Y_t(v^*)$, and let $c_X=X_t(v^*)$ and $c_Y=Y_t(v^*)$.
This is only a notational orientation of the neighboring pair; after constructing the coupled transition, the original names of the two chains can be restored. In particular, it does not change either marginal chain or the Hamming distance.
For the remainder of the coupling definition, let $d_c=d_c(v^*)=|N_c(v^*)|$ for every $c\in[k]$.

By a \emph{marginal cluster flip} we mean a possible cluster flip in one of the two individual chains. For every color $c\in[k]$, let $\mathcal D_c$ denote the marginal cluster flips resulting from either: (i) attempting to recolor $v^*$ to color $c$, or (ii) recoloring a $c$-colored neighbor of $v^*$ to one of the current colors of $v^*$.
We write each marginal cluster flip as a pair consisting of the chain and the displayed cluster. More precisely, let
\begin{equation}
\label{eqn:DC_def}
\mathcal{D}_{c} :=  \left\{ (X_t,S_{X_t}(v^*,c)),(Y_t,S_{Y_t}(v^*,c)) \right\} \cup \bigcup_{u \in N_{c}(v^*)} \{(X_t,S_{X_t}(u,c_Y)),(Y_t,S_{Y_t}(u,c_X))\}.
\end{equation}
Thus, occurrences of the same cluster in the two chains are distinct elements.
To reiterate, the class $\mathcal D_c$ includes the marginal cluster flips whose clusters either: (i) contain $v^*$ and color $c$, or (ii) contain a $c$-colored neighbor $u$ of $v^*$ and one of the colors $c_X$ or $c_Y$.

 {Note, $\mathcal{D}_{c}$ and $\mathcal{D}_{c'}$ are disjoint for colors $c\neq c'$, except possibly for colors $c=c_X,c'=c_Y$.
These two color classes $\mathcal D_{c_X}$ and $\mathcal D_{c_Y}$ are the only sets that may contain a common marginal cluster flip. Indeed, a cluster using the color pair $\{c_X,c_Y\}$ may be represented through $v^*$ in one class and through a neighbor of $v^*$ in the other. For bookkeeping, after making the above definition for every color, replace
$\mathcal D_{c_Y}$ by $\mathcal D_{c_Y}\setminus\mathcal D_{c_X}$.
This assigns every shared flip to $\mathcal D_{c_X}$, so the classes $(\mathcal D_c)_{c\in[k]}$ are pairwise disjoint. This reassignment is only for bookkeeping: it leaves the union of the classes and all marginal flip probabilities unchanged.
We call $c_X$ and $c_Y$ the degenerate colors, and all other colors nondegenerate.}

Let
\begin{equation*}
\mathcal{D} = \bigcup_{c\in[k]}\mathcal{D}_c.
\end{equation*}
The set $\mathcal{D}$ contains every marginal cluster flip that can differ in the two chains or potentially change the Hamming distance.
If a cluster $S$ exists in both chains and neither $(X_t,S)$ nor $(Y_t,S)$ lies in $\mathcal D$, then, with probability $P_{|S|}/(nk)$, the coupling pairs the flip of $S$ in $X_t$ with the flip of $S$ in $Y_t$; this is the identity coupling for clusters outside $\mathcal{D}$.

For marginal cluster flips in $\mathcal{D}$, the coupling considers each $\mathcal{D}_{c}$ separately and defines a pairing of flips assigned to $\mathcal{D}_{c}$.  A flip in $X_t$ assigned to $\mathcal{D}_c$ is paired with a flip in $Y_t$ assigned to the same class, or one of the transitions is a self-loop.  Consequently, we can define the effect of the coupled moves with respect to a single class $\mathcal{D}_c$.

\begin{definition}
\label[definition]{def:hamming-change}
For color $c\in[k]$, let $\mathcal{C}_c$ denote the event that the selected coupled move is assigned to the class $\mathcal D_c$.
A \emph{one-sided move} is a coupled transition in which exactly one chain flips a cluster. We regard the selected move as part of the coupling, so a representation assigned probability zero does not trigger $\mathcal C_c$, while a one-sided move assigned to $\mathcal D_c$ does.
Then, let
\begin{equation*}
W_c := \ExpCond{\left(H(X_{t+1},Y_{t+1})-1\right)\indicator(\mathcal{C}_c)}{X_t,Y_t}.
\end{equation*}
Moreover, let
\begin{equation*}
W := \ExpCond{H(X_{t+1},Y_{t+1})-1}{X_t,Y_t}.
\end{equation*}
\end{definition}

Because the classes $(\mathcal D_c)_{c\in[k]}$ are pairwise disjoint and every coupled move below is allocated within a single class, the events $\left(\mathcal{C}_c\right)_{c\in [k]}$ are pairwise disjoint. Moreover, every coupled move that changes the Hamming distance is assigned to exactly one of these events.
All coupled moves not assigned to one of these events, namely the identity pairings outside $\mathcal D$ and the double self-loop in which neither chain changes, do not change the Hamming distance.  Consequently, we have the following identity:
\begin{equation}
    \label{eqn:hamming-color-decomposition}
    W = \sum_{c\in [k]} W_c.
\end{equation}

In the remainder of this subsection, we handle the case $d_c=0$ for every color and the case of nondegenerate colors $c\notin\{c_X,c_Y\}$ with $d_c\geq1$. The remaining cases, $c\in\{c_X,c_Y\}$ with $d_c\geq1$, require separate bookkeeping and are handled in \cref{sec:degenerate-colors}.
We omit the common factor $1/(nk)$ from the probabilities in the subsequent description to simplify the expressions.

Throughout the coupling definition, a hatted variable denotes the \emph{unadjusted size} of the relevant cluster $S$.  An unhatted variable denotes its \emph{adjusted coupling size} after repeated representations have been assigned to a single color class and a single neighbor.  Thus, an adjusted coupling size is either its corresponding unadjusted size or zero.  Setting an adjusted coupling size to zero is only a bookkeeping convention. The corresponding cluster is still present in the chain and its flip is retained under another representation.  In particular, the hatted variables are sizes of clusters $S$, rather than sizes of the components $\widehat S$.

Suppose $d_c=0$, which corresponds to $c\notin X_{t}(N(v^*))$. In this case, $\mathcal{D}_{c}$ consists of exactly one cluster in each chain: $S_{X_{t}}(v^*,c)$ in $X_t$ and $S_{Y_{t}}(v^*,c)$ in $Y_t$.

If $c\notin \colorList{v^*}$, both clusters have empty underlying vertex sets and no flip occurs. Suppose instead that $c\in \colorList{v^*}$. If $c\notin\{c_X,c_Y\}$, both clusters have underlying vertex set $\{v^*\}$, although they have different unordered color pairs, and the coupling flips both clusters with probability $P_1$. If $c=c_X$, then $S_{X_t}(v^*,c)=\emptyset$ and $S_{Y_t}(v^*,c)$ has underlying vertex set $\{v^*\}$; the coupling pairs the flip in $Y_t$ with a self-loop in $X_t$. The case $c=c_Y$ is symmetric. Thus, whenever $c\in \colorList{v^*}$, with probability $P_1=1$ the coupled transition makes the two chains agree at $v^*$.
Consequently, for $c\in [k]$ where $d_c=0$ we have:
\begin{equation}
\label{eqn:zero-neighbor-change}
    (nk)W_c \leq \begin{cases} -1 & \mbox{if } c\in \colorList{v^*} \\
0 & \mbox{if } c\notin \colorList{v^*}.
\end{cases}
\end{equation}

We now consider the cases where $c \notin \{c_X,c_Y\}$ and $d_c \geq 1$. We do this by first considering the simplest case, $d_c = 1$, and then giving a more general explanation of the coupling for all $d_c \geq 1$.

Suppose $d_c=1$ and $c \notin \{c_X,c_Y\}$.
Let $N_{c}(v^*)=\{u\}$.
In this case, $\mathcal{D}_{c}$ consists of exactly two clusters in each chain. In ${X_{t}}$ we have $S_{{X_{t}}}(v^*,c)$ and $S_{{X_{t}}}(u,c_Y)$, and in ${Y_{t}}$ we have $S_{{Y_{t}}}(v^*,c)$ and $S_{{Y_{t}}}(u,c_X)$.
In this case, let $\widehat a=|S_{{X_{t}}}(u,c_Y)|$ and $\widehat A=|S_{{Y_{t}}}(v^*,c)|$, and similarly let $\widehat b=|S_{{Y_{t}}}(u,c_X)|$ and $\widehat B=|S_{{X_{t}}}(v^*,c)|$.
There are no repeated representations in this case, so set $a=\widehat a$, $A=\widehat A$, $b=\widehat b$, and $B=\widehat B$.
That is, the adjusted and unadjusted sizes agree.
We make a few important observations about the relationship of these clusters:
\begin{enumerate}
    \item If $S_{{X_{t}}}(u,c_Y) = \emptyset$ (i.e., the cluster is not flippable) then $S_{{Y_{t}}}(v^*,c)= \emptyset$ (i.e., the cluster is not flippable) and $A = a = 0$.
    \item If $S_{{X_{t}}}(u,c_Y) \neq \emptyset$ (i.e., the cluster is flippable) but $S_{{Y_{t}}}(v^*,c)= \emptyset$ (i.e., the cluster is not flippable), then $c \notin \colorList{v^*}$. In this case, $A = 0 < a$.
    \item If $S_{{X_{t}}}(u,c_Y) \neq \emptyset$ and $S_{{Y_{t}}}(v^*,c) \neq \emptyset$ (i.e., both clusters are flippable), then $S_{Y_t}(v^*,c) = S_{X_t}(u,c_Y)\cup\{v^*\}$. Moreover, $v^* \notin S_{X_t}(u,c_Y)$ since $c \neq c_X$ and thus, $A=1+a$.
\end{enumerate}
Note that in all of the above cases, $P_A \leq P_a$.
Indeed, even in the second case where $A < a$, it follows that $P_A = 0 \leq P_a$.
Symmetric observations can be made about the relationship of $B$ and $b$.

The greedy coupling pairs the flip of $S_{X_t}(u,c_Y)$ in $X_t$ with the flip of $S_{Y_t}(v^*,c)$ in $Y_t$ as much as possible, which is with probability $P_A$; afterwards, $v^*$ is still the only possible disagreement.  Similarly, it pairs $S_{{X_{t}}}(v^*,c)$ with $S_{{Y_{t}}}(u,c_X)$ with probability $P_B$.  Finally, we pair the remaining flips as much as possible.  Here is the more detailed description.
\begin{enumerate}
    \item With probability $P_{A}$, couple the flip of $S_{X_t}(u,c_Y)$ (of size $a$) in ${X_{t}}$ with $S_{Y_{t}}(v^*,c)$ (of size $A$) in $Y_{t}$. Note that if $A>0$ then $A=1+a$ and hence $P_{A} \leq P_{a}$.
    \item With probability $P_{B}$, couple the flip of $S_{X_t}(v^*,c)$ (of size $B$) in ${X_{t}}$ with $S_{Y_t}(u,c_X)$ (of size ${b}$) in ${Y_{t}}$.  Again, note that $P_B\leq P_b$.
    \item With probability $\min\{P_{a}-P_{A},P_{b}-P_{B}\}$, flip $S_{{X_{t}}}(u,c_Y)$ in ${X_{t}}$ and $S_{{Y_{t}}}(u,c_X)$ in ${Y_{t}}$.
    \item With probability $P_a - P_A - \min\{P_{a}-P_{A},P_{b}-P_{B}\}$, flip $S_{{X_{t}}}(u,c_Y)$ in ${X_{t}}$ and no flip in $Y_t$.
    \item With probability $P_b - P_B - \min\{P_{a}-P_{A},P_{b}-P_{B}\}$, flip $S_{{Y_{t}}}(u,c_X)$ in ${Y_{t}}$ and no flip in $X_t$.
\end{enumerate}

In cases 1 and 2 of the coupling, after the coupled move, the vertex $v^*$ is the only possible disagreement, and hence the Hamming distance does not increase in either case.  In case 3, the Hamming distance increases by at most $a+b-m$ where $m$ is the size of the overlap of these two clusters.  Finally, in cases 4 and 5 the Hamming distance increases by at most $a$ and $b$, respectively. Summarizing, when $d_c=1$ and $c \notin \{c_X,c_Y\}$,
the expected change in Hamming distance contributed by the coupled moves assigned to $\mathcal{D}_c$ satisfies
\begin{equation}
\label{eqn:single-neighbor-change}
    (nk)W_c \leq a(P_{a}-P_A) + b(P_{b}-P_B) - m\min\{P_a-P_A,P_b-P_B\},
\end{equation}
where $m=|S_{X_t}(u,c_Y)\cap S_{Y_t}(u,c_X)|$.    Moreover, if both of these clusters $S_{X_t}(u,c_Y)$ and $S_{Y_t}(u,c_X)$ are flippable, and hence nonempty, then $m\geq 1$ since both clusters contain $u$.  The bound in \cref{eqn:single-neighbor-change} is still valid if one of the clusters in $\mathcal{D}_c$ is not flippable since $P_0=0$.

Now consider the general case $d_c\geq 1$ and $c \notin \{c_X,c_Y\}$.
Let $N_c(v^*)=\{u_1, \ldots, u_{d_c}\}$.
In this case, $\mathcal{D}_{c}$ consists of the following clusters in ${X_{t}}$:
\begin{equation*} S_{{X_{t}}}(v^*,c), \mbox{ and } S_{X_{t}}(u_i,c_Y) \mbox{ for all $1\leq i\leq d_c$};
\end{equation*}
and in ${Y_{t}}$ we have:
\begin{equation*} S_{{Y_{t}}}(v^*,c), \mbox{ and } S_{Y_{t}}(u_i,c_X) \mbox{ for all $1\leq i\leq d_c$}.
\end{equation*}
We call $S_{X_t}(v^*,c)$ and $S_{Y_t}(v^*,c)$ the \emph{large clusters}, and we call $S_{X_t}(u_i,c_Y)$ and $S_{Y_t}(u_i,c_X)$ the \emph{small clusters}. In this nondegenerate case, every nonempty large cluster contains $v^*$, whereas each nonempty small cluster contains $u_i$ but not $v^*$.
For $1\leq i\leq d_c$, let $\widehat a_i=|S_{X_{t}}(u_i,c_Y)|$ and $\widehat b_i=|S_{Y_{t}}(u_i,c_X)|$.
Initially set $a_i=\widehat a_i$ and $b_i=\widehat b_i$.
If $u_i \in S_{X_{t}}(u_j,c_Y)$ for some $j<i$, set $a_i=0$ so that this cluster flip is counted only under its first representation.  Similarly, if $u_i \in S_{Y_{t}}(u_j,c_X)$ for some $j<i$, set $b_i=0$.

Let $a_{\max}=\max_i a_i$ and $b_{\max}=\max_i b_i$.
Let $i_{\max}$ be an index $i$ where $a_i=a_{\max}$ and similarly let $i'_{\max}$ be an index $i'$ where $b_{i'}=b_{\max}$.
Moreover, if possible set $i_{\max} = i'_{\max}$.
Finally, let $\widehat A=|S_{{Y_{t}}}(v^*,c)|$ and $\widehat B=|S_{{X_{t}}}(v^*,c)|$, and set $A=\widehat A$ and $B=\widehat B$.
If $A>0$, the corresponding union identity and the list validity of the large cluster give $A=1+\sum_{i=1}^{d_c} a_i$. Thus, either $A=1+\sum_{i=1}^{d_c} a_i$ or $A=0$.  Similarly, either $B=1+\sum_{i=1}^{d_c} b_i$ or $B=0$.

The following figure illustrates the large and small clusters in the two chains.

\begin{figure}[ht]
    \centering
    \includegraphics[width=.78\textwidth]{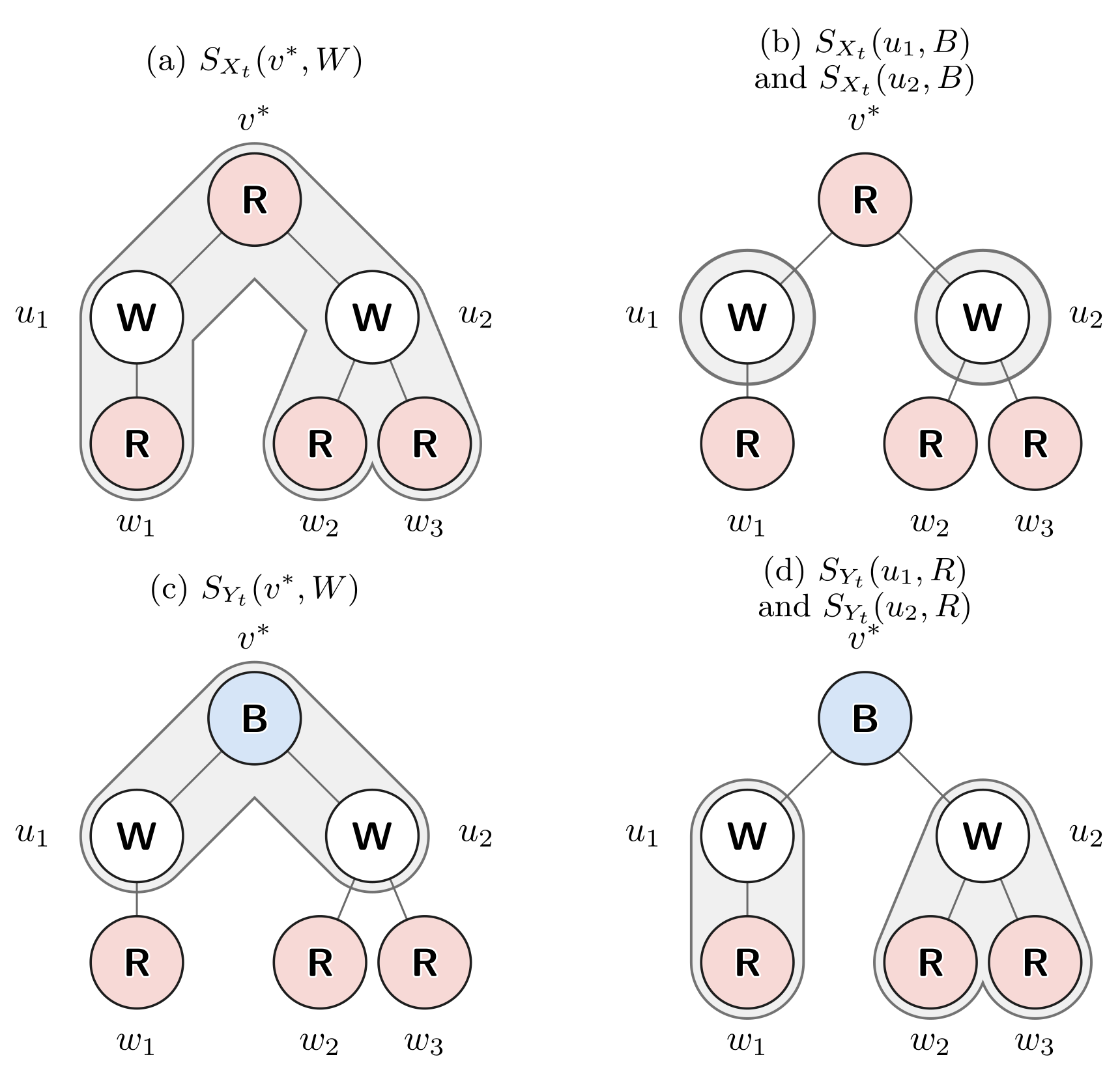}
    \caption{The clusters in the greedy coupling for $c=W$, with $c_X=R$ and $c_Y=B$. The top row shows $X_t$ and the bottom row shows $Y_t$; the left-hand panels show the large clusters and the right-hand panels show the small clusters. In $X_t$, the cluster $S_{X_t}(v^*,W)$ contains the entire graph, while $S_{X_t}(u_1,B)$ and $S_{X_t}(u_2,B)$ are singletons. In $Y_t$, the cluster $S_{Y_t}(v^*,W)$ contains $v^*,u_1,$ and $u_2$, while $S_{Y_t}(u_1,R)$ contains $u_1,w_1$ and $S_{Y_t}(u_2,R)$ contains $u_2,w_2,w_3$. Thus, the corresponding sizes are $B=|S_{X_t}(v^*,W)|$, $a_i=|S_{X_t}(u_i,B)|$, $A=|S_{Y_t}(v^*,W)|$, and $b_i=|S_{Y_t}(u_i,R)|$.}
    \label{fig:standard-coupling-clusters}
\end{figure}

For $1\leq i \leq d_c$, let
\begin{equation*}  Q_i = P_{a_i} - P_{A} \indicator({i = i_{\max}}), \mbox{ and }  Q'_i = P_{b_i} - P_{B} \indicator({i = i'_{\max}});
\end{equation*}
Here $Q_i$ and $Q'_i$ are the remaining flip probabilities for the small clusters after each large-cluster flip is paired with a selected small-cluster flip in the other chain.

The following coupling for the case $d_c\geq 1$ generalizes the earlier $d_c=1$ case, and greedily pairs the large clusters with the largest of the small clusters.
\begin{enumerate}
\item With probability $P_A$ couple the flip of
$S_{X_t}(u_{i_{\max}},c_Y)$ in $X_t$ with the flip of
$S_{Y_t}(v^*,c)$ in ${Y_t}$. Note that $P_A\leq P_{a_{i_{\max}}}$
since $A\geq 1+a_{i_{\max}}$ or $A=0$.
\item With probability $P_B$ couple the flip of
$S_{X_t}(v^*,c)$ in $X_t$ with
$S_{Y_t}(u_{i'_{\max}},c_X)$ in $Y_t$.
As in the previous case, $P_B\leq P_{b_{i'_{\max}}}$.
    \item For each $1\leq i \leq d_c$:
    \begin{enumerate}
    \item With probability $\min\{Q_i,Q'_i\}$ flip $S_{{X_{t}}}(u_i,c_Y)$ in ${X_{t}}$ and $S_{{Y_{t}}}(u_i,c_X)$ in ${Y_{t}}$.
    \item With the probability $Q_i - \min\{Q_i,Q'_i\}$, flip $S_{{X_{t}}}(u_i,c_Y)$ in ${X_{t}}$ and no flip in $Y_t$.
    \item With the probability $Q'_i - \min\{Q_i,Q'_i\}$, flip $S_{{Y_{t}}}(u_i,c_X)$ in ${Y_{t}}$ and no flip in $X_t$.
    \end{enumerate}
\end{enumerate}

Observe that for any $c \notin \{c_X,c_Y\}$ where $d_c\geq1$, we have
\begin{equation}
\label{eqn:hamming-change-bound}
    (nk)W_c \leq (A - a_{\max} - 1)P_A + (B- b_{\max} -1)P_B + \sum_{i=1}^{d_c} \left( a_i Q_i + b_i Q'_i - m_i \min \{Q_i,Q'_i\}\right),
\end{equation}
where
\begin{equation*}
m_i=
\begin{cases}
|S_{{X_{t}}}(u_i,c_Y)\cap S_{{Y_{t}}}(u_i,c_X)| & \mbox{if } a_i,b_i>0,\\
0 & \mbox{otherwise}.
\end{cases}
\end{equation*}
When $d_c=1$ then \cref{eqn:hamming-change-bound} is the same as \cref{eqn:single-neighbor-change}.

For any color $c \notin \{c_X,c_Y\}$ such that $d_c\geq 1$, define
\begin{equation}
\label{def:z-c-upper-bound}
Z_c :=
\frac{1}{nk}
\left[
(A-a_{\max}-1)P_A+(B-b_{\max}-1)P_B
+\sum_{i=1}^{d_c}
\left(
a_iQ_i+b_iQ'_i-\min\{Q_i,Q'_i\}
\right)
\right].
\end{equation}

It follows from \cref{eqn:hamming-change-bound} that, for every $c\in[k]\setminus\{c_X,c_Y\}$ such that $d_c\geq1$,
\begin{equation}
\label{eqn:nondegenerate-w-c-z-c-bound}
W_c\leq Z_c.
\end{equation}
Indeed, if $\min\{Q_i,Q'_i\}=0$, the corresponding terms are equal.
Otherwise, both retained clusters in the definition of $m_i$ are nonempty and contain $u_i$, and hence $m_i\geq1$ and the bound follows.

This completes the coupling construction and the Hamming-distance bound for nondegenerate colors, as well as all cases with $d_c=0$. The degenerate-color coupling for $d_c\geq1$ is handled in \cref{sec:degenerate-colors}; the contribution estimate needed for the metric analysis is stated before the contraction argument.

\section{New Metric Definition}
\label{sec:new-metric}
This section introduces our metric, which is the heart of the proof of~\cref{thrm:flip-mixing}.
First, we introduce the notation for our metric. Then, we describe our flip probabilities and identify some of their key properties that will make our subsequent analysis simpler.

\subsection{New Metric}
\label{subsec:new-metric}

Before we formally introduce our metric, we need several new definitions.
For an integer $s\geq 0$, consider a pair of list labelings $X_s,Y_s \in \Omega$; note this pair $X_s,Y_s$ may differ at an arbitrary number of vertices.

 {
For a vertex $z\in V$ such that $X_s(z) \neq Y_s(z)$, we partition the neighbors of $z$ based on how obstructed they are from being recolored with one of the disagreement colors, $X_s(z)$ and $Y_s(z)$.
For $z\in V$, $y\in N(z)$, and $(\sigma,\tau)\in\Omega^2$ such that $\sigma(z) \neq \tau(z)$,
let
\begin{equation*}
    B^z(y,\sigma,\tau) := \{w \in N(y)\setminus\{z\} : \{\sigma(w),\tau(w)\} \cap \{\sigma(z),\tau(z)\} \neq \emptyset \}
\end{equation*}
be the \emph{blocking neighbors} of $y$ with respect to $z$ in $\sigma$ and $\tau$.
For ease of notation, when considering the pair $X_s,Y_s$ with $X_s(z)\neq Y_s(z)$, we write
\begin{equation*}
    B^z_s(y) := B^z(y,X_s,Y_s).
\end{equation*}
We now introduce the notion of a \emph{blocking level} $b^z(y,\sigma,\tau)$ of $y$ with respect to $z$ in $\sigma$ and $\tau$.
For $z \in V$, $y \in N(z)$, and $(\sigma,\tau)\in \Omega^2$ such that $\sigma(z) \neq \tau(z)$, let
\begin{equation*}
    b^z(y,\sigma,\tau) :=
    \begin{cases}
        2 & \mbox{if $\{\sigma(z),\tau(z)\} \nsubseteq \colorList{y}$}\\
        2 & \mbox{if $\{\sigma(y),\tau(y)\} \cap \{\sigma(z),\tau(z)\} \neq \emptyset$} \\
        \min\{\lvert B^z(y,\sigma,\tau) \rvert,2\} & \mbox{otherwise}
    \end{cases}
\end{equation*}
When $X_s(z)\neq Y_s(z)$, let $b^z_s(y) = b^z(y,X_s,Y_s)$.}

When $X_s(z)\neq Y_s(z)$, we say that a neighbor $y \in N(z)$ is \emph{unblocked} with respect to $z$ if $b^z_s(y) = 0$,
\emph{singly blocked} with respect to $z$ if $b^z_s(y) = 1$,
and \emph{multiblocked} with respect to $z$ if $b^z_s(y) = 2$.
Blocking level $2$ captures the cases in which $y$ has at least two blocking neighbors, one of the disagreement colors at $z$ is absent from $\colorList{y}$, or $y$ itself receives one of the disagreement colors in either labeling.
We group these scenarios because each obstructs recoloring $y$ with one of the disagreement colors at $z$.

Accordingly, we partition the neighbors of $z$ by blocking level.
For an integer $0\leq i\leq2$, $z\in V$, and $(\sigma,\tau)\in\Omega^2$ such that $\sigma(z) \neq \tau(z)$, let
\begin{equation*}
F^i(z,\sigma,\tau) :=
    \{y\in N(z): b^z(y,\sigma,\tau)
    = i \}
\end{equation*}
be the set of neighbors of $z$ with blocking level $i$.
If $\sigma(z) = \tau(z)$, let $F^i(z,\sigma,\tau)=\emptyset$.
Let
\begin{equation*}
    \degree{i}{}{z,\sigma,\tau} := \lvert F^i(z,\sigma,\tau) \rvert
\end{equation*}
denote the number of neighbors of $z$ with blocking level $i$.
For ease of notation, when considering the pair $X_s,Y_s$, we write
\begin{equation*}
F^i_s(z) := F^i(z,X_s,Y_s)
\mbox{ and }
    \degree{i}{s}{z} := \degree{i}{}{z,X_s,Y_s}.
\end{equation*}
Note that $\unblockedShort{s}{z}$ is the set of unblocked neighbors of $z$ in $(X_s,Y_s)$, $\singleBlockedShort{s}{z}$ is the set of singly blocked neighbors of $z$, and $\multiBlockedShort{s}{z}$ is the set of multiblocked neighbors of $z$.
When $X_s(z)\neq Y_s(z)$, the sets $\unblockedShort{s}{z}$, $\singleBlockedShort{s}{z}$, and $\multiBlockedShort{s}{z}$ form a partition of $N(z)$,
\begin{equation*}
    \unblockedShort{s}{z} \sqcup  \singleBlockedShort{s}{z} \sqcup \multiBlockedShort{s}{z} = N(z)
\end{equation*}
and likewise, $d(z) = \degree{0}{s}{z} + \degree{1}{s}{z} + \degree{2}{s}{z}$. When $X_s(z)=Y_s(z)$, all three sets are empty by convention.
When $X_s(z)\neq Y_s(z)$, we call the collection of blocking levels of a set of neighbors of $z$ their \emph{blocking profile}. The profile is uniform if all these levels are equal and mixed otherwise.

\begin{figure}[ht]
\centering
\includegraphics[width=.78\textwidth]{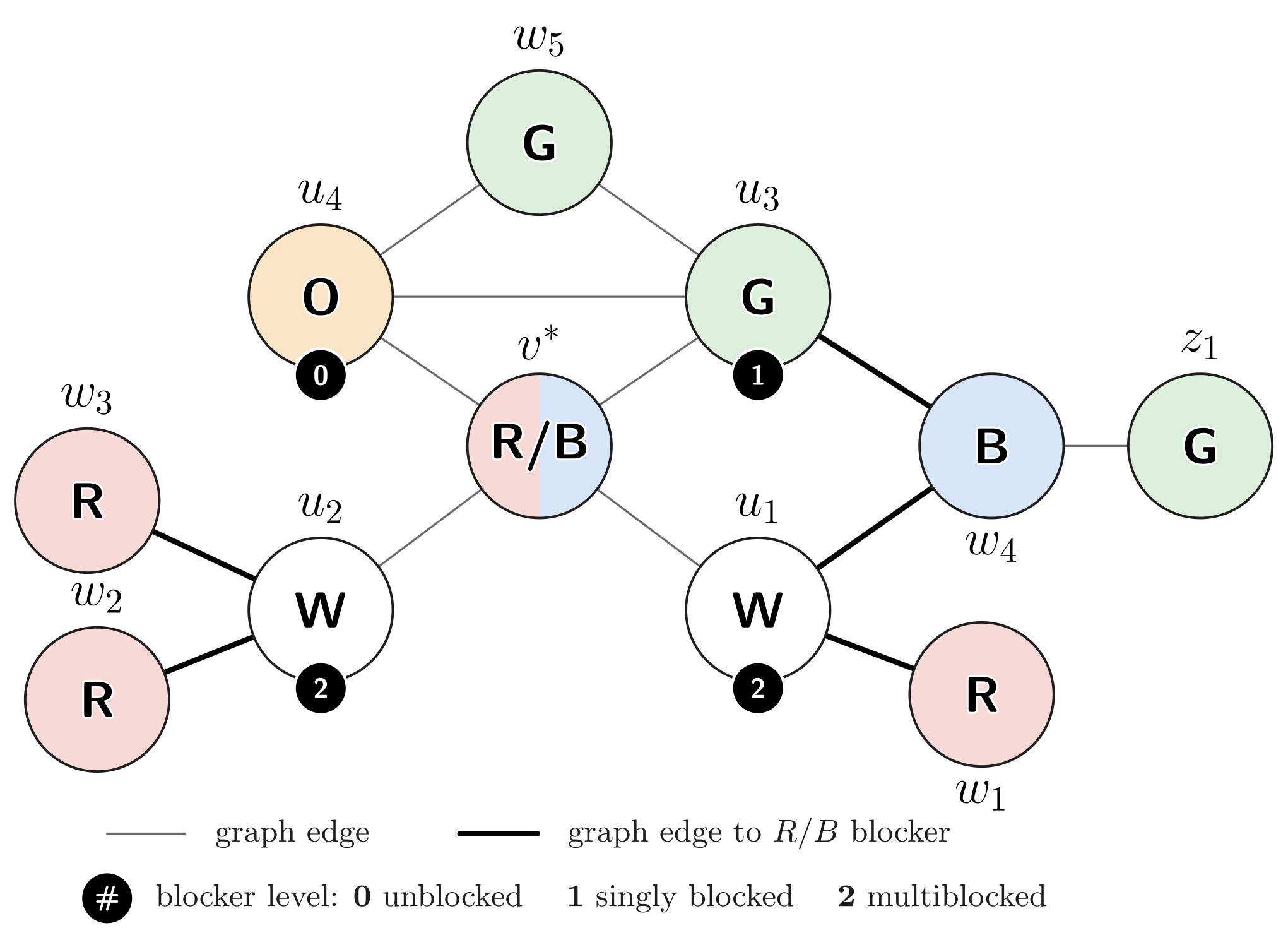}

\caption{A small graph illustrating unblocked, singly blocked, and
multiblocked neighbors of $v^*$. The two list labelings differ only at
$v^*$, where their colors are $R$ and $B$; the letter inside every other
vertex indicates its common color in the two labelings.
Each vertex's list is assumed to contain every color shown.
The badge on each neighbor of $v^*$ indicates its blocking level.
Bold edges connect a neighbor of $v^*$ to a blocking neighbor, namely
a vertex other than $v^*$ colored $R$ or $B$.
The vertex $u_1$ is multiblocked because it has two blocking neighbors,
$w_1$ and $w_4$, colored $R$ and $B$, respectively.
Likewise, $u_2$ is multiblocked because it has two blocking neighbors,
$w_2$ and $w_3$, both colored $R$.
The vertex $u_3$ is singly blocked because its only blocking neighbor
is $w_4$, colored $B$.
Finally, $u_4$ is unblocked because it has no neighbors other than $v^*$
colored $R$ or $B$.
The shared green neighbor $w_5$ of $u_3$ and $u_4$, and the additional
orange neighbor $u_4$ of $u_3$, do not affect their blocking levels.}
\label{fig:blocking-levels}
\end{figure}

We can now formally define our metric.
Consider a pair of list labelings $(X_t, Y_t) \in\Omega^2$.
We first define for a vertex $z\in V$,
\begin{equation*}
    H_z(X_t,Y_t) := \indicator\bigl(X_t(z) \neq Y_t(z)\bigr),
\end{equation*}
and
\begin{equation}
\label{def:vertex-metric}
    \nmet_z(X_t,Y_t) := H_z(X_t,Y_t)
    \left(1 - \frac{\eta}{\Delta} \degree{0}{t}{z}\right).
\end{equation}
We call $\eta\degree{0}{t}{z}/\Delta$ the \emph{discount} attached to a disagreement at $z$.
We assume that $\Delta>0$ and $0\leq\eta<1/2$ and specify the value of $\eta$ when we define flip parameters in the subsequent subsection.
For $z\in V$ and $(\sigma,\tau)\in\Omega^2_z$, define the edge weight between these neighboring list labelings by
\begin{equation*}
    \omega(\sigma,\tau) := 1 - \frac{\eta}{\Delta}\degree{0}{}{z,\sigma,\tau}.
\end{equation*}

We now form a metric by extending $\omega$ to all pairs of list labelings.
Consider the graph $\Gamma = (\Omega,\Sigma)$ where the weight of each edge $(\sigma,\tau) \in \Sigma$ is $\omega(\sigma,\tau)$.
The graph $\Gamma$ is connected since any pair of list labelings can be joined by changing the labels of the disagreement vertices one at a time.
For $\sigma,\tau \in \Omega$, let $\Phi_\Gamma(\sigma,\tau)$ be the set of all paths from $\sigma$ to $\tau$ in $\Gamma$.
For $\sigma,\tau \in \Omega$, let $\nmet(\sigma,\tau)$ be the length of the shortest weighted path between $\sigma$ and $\tau$ in $\Gamma$ with respect to $\omega$,
\begin{equation*}
    \nmet(\sigma,\tau) := \min_{\phi\in\Phi_\Gamma(\sigma,\tau)} \sum_{(\phi_i,\phi_{i+1}) \in \phi} \omega(\phi_i,\phi_{i+1}).
\end{equation*}
Since $0\leq\eta<1/2$, it follows that the shortest path between any two neighboring list labelings $(X_t,Y_t) \in \Sigma$ with unique disagreement vertex $v^*$ must be the edge between them with weight $\omega(X_t,Y_t)$,
\begin{equation}
\label{eqn:neighboring-metric}
    \nmet(X_t,Y_t) = \omega(X_t,Y_t) = 1 - \frac{\eta}{\Delta}\degree{0}{t}{v^*}.
\end{equation}
Indeed, $\degree{0}{t}{v^*}\leq \Delta$, so every edge weight is in $[1-\eta,1]$. Thus, the edge between two neighboring list labelings has weight at most $1$, while any path containing two or more edges has weight strictly greater than $1$.

The path metric is uniformly equivalent to Hamming distance:
\begin{equation}
\label{eqn:metric-hamming-equivalence}
    (1-\eta)H(\sigma,\tau)
    \leq \nmet(\sigma,\tau)
    \leq H(\sigma,\tau)
    \qquad\mbox{for all $\sigma,\tau\in\Omega$}.
\end{equation}
Indeed, every path from $\sigma$ to $\tau$ contains at least $H(\sigma,\tau)$ edges, each of weight at least $1-\eta$, while changing the disagreement vertices one at a time gives a path with $H(\sigma,\tau)$ edges, each of weight at most one.

The following deterministic bound relates the path metric $\nmet$ to the Hamming distance and the number of unblocked neighbors. It is tight for neighboring list labelings.
Recall that if $X_t(u) = Y_t(u)$ then $\unblockedShort{t}{u} = \emptyset$ and hence
$\degree{0}{t}{u}=0$. Thus, by \cref{def:vertex-metric},
\begin{equation*}
\sum_{z\in V}\nmet_z(X_t,Y_t)
=H(X_t,Y_t)-\frac{\eta}{\Delta}\sum_{z\in V}\degree{0}{t}{z}.
\end{equation*}

\begin{lemma}
\label[lemma]{lem:metric-decomposition}
If $0\leq\eta<1/2$, then for any $X_t,Y_t\in\Omega$,
\begin{equation*}
    \nmet(X_t,Y_t)\leq\sum_{z\in V}\nmet_z(X_t,Y_t).
\end{equation*}
\end{lemma}

\begin{proof}
Let $U = \{u \in V  :  {Y_{t}}(u) \neq {X_{t}}(u)\}$ be the set of vertices that ${X_{t}}$ and ${Y_{t}}$ disagree on, and
let $U=\{u_1, \ldots, u_{\lvert U \rvert}\}$ be an arbitrary ordering of $U$.
For $0\leq i\leq \lvert U\rvert$, let $\rho_i(u_j)=Y_t(u_j)$ if $j\leq i$, and let $\rho_i(v)=X_t(v)$ for every other vertex $v$.
Thus, $\rho_0=X_t$, $\rho_{\lvert U\rvert}=Y_t$, and $\rho_{i-1}$ and $\rho_i$ differ only at $u_i$.
It follows that
\begin{equation}
    \label{eqn:path-upper-bound}
    \nmet({X_{t}},{Y_{t}}) \leq \sum_{i=1}^{\lvert U\rvert} \nmet(\rho_{i-1},\rho_i)
\end{equation}
since $\nmet({X_{t}},{Y_{t}})$ is defined to be the length of the shortest path between ${X_{t}}$ and ${Y_{t}}$ and the path $\rho_0,\ldots, \rho_{\lvert U \rvert}$ is a particular path.
We will prove that for all $u_i \in U$,
\begin{equation}
    \label{eqn:unblocked-path-bound}
    \degree{0}{}{u_i,\rho_{i-1},\rho_{i}}  \geq \degree{0}{}{u_i,X_t,Y_t}.
\end{equation}
Assuming \cref{eqn:unblocked-path-bound} and making use of \cref{eqn:path-upper-bound} we can conclude the lemma as follows:
\begin{align*}
    \nmet({X_{t}},{Y_{t}}) & \leq \sum_{i=1}^{\lvert U\rvert} \nmet(\rho_{i-1}, \rho_{i})  &&\mbox{(by \cref{eqn:path-upper-bound})}\\
    &= \sum_{u_i\in U} \left(1- \frac{\eta} {\Delta} \degree{0}{}{u_i,\rho_{i-1},\rho_{i}} \right) &&\mbox{(by \cref{eqn:neighboring-metric} since $\eta <1/2$)}\\
    &= H({X_{t}},{Y_{t}}) - \frac{\eta} {\Delta}\sum_{u_i\in U} \degree{0}{}{u_i,\rho_{i-1},\rho_{i}}  && \mbox{($\lvert U \rvert = H(X_t,Y_t)$)}\\
    &\leq H({X_{t}},{Y_{t}}) - \frac{\eta}{\Delta}\sum_{u_i\in U} \degree{0}{}{u_i,X_t,Y_t} &&\mbox{(by \cref{eqn:unblocked-path-bound} since $\eta \geq 0$)}\\
    &= \sum_{z \in V} \nmet_{z}(X_t,Y_t).
\end{align*}

It remains to prove \cref{eqn:unblocked-path-bound}, for which it suffices to show that, for all $u_i \in U$,
\begin{equation*}
\unblocked{u_i}{\rho_{i-1}}{\rho_{i}} \supseteq  \unblocked{u_i}{X_t}{Y_t}.
\end{equation*}
Fix $i$ and suppose $w \in \unblocked{u_i}{X_t}{Y_t}$.
We will show that $w \in \unblocked{u_i}{\rho_{i-1}}{\rho_{i}}$.
By the definition of $\unblocked{u_i}{X_t}{Y_t}$, we have $B^{u_i}(w,X_t,Y_t) = \emptyset$,
\begin{equation*}
\{X_t(u_i),Y_t(u_i)\}\subseteq \colorList{w},
\end{equation*}
and $\{X_t(w),Y_t(w)\}\cap\{X_t(u_i),Y_t(u_i)\}=\emptyset$.
Consider $z \in N(w) \setminus \{u_i\}$.
Since $B^{u_i}(w,X_t,Y_t)=\emptyset$, we have
\begin{equation*}
\{{X_{t}}(z),{Y_{t}}(z)\} \cap \{{X_{t}}(u_i),{Y_{t}}(u_i)\}~=~\emptyset.
\end{equation*}
By construction,
\begin{equation*}
\rho_{i-1}(z)\in \{X_t(z),Y_{t}(z)\},\rho_{i}(z) \in \{X_t(z),Y_{t}(z)\},\rho_{i-1}(u_i) = {X_{t}}(u_i), \mbox{  and } \rho_{i}(u_i) = {Y_{t}}(u_i).
\end{equation*}
Thus, for all $z \in N(w)\setminus \{u_i\}$,
\begin{equation*} \{\rho_{i-1}(z),\rho_{i}(z)\} \cap \{\rho_{i-1}(u_i),\rho_{i}(u_i)\} \subseteq \{X_t(z),Y_t(z)\} \cap \{X_t(u_i),Y_t(u_i)\} = \emptyset.
\end{equation*}
The disagreement colors at $u_i$ remain $X_t(u_i)$ and $Y_t(u_i)$, so the list-containment condition at $w$ continues to hold. The color-disjointness condition also continues to hold because $\rho_{i-1}(w)=\rho_i(w)\in\{X_t(w),Y_t(w)\}$. Since this holds for all $z\in N(w)\setminus \{u_i\}$, we have shown that $w \in \unblocked{u_i}{\rho_{i-1}}{\rho_i}$ as desired, which completes the proof of the lemma.
\end{proof}

\subsection{Flip Probabilities}
\label{subsec:flip-probabilities}
Now that we have defined our metric, we specify our new flip parameters and the precise setting of $\eta$.
These settings differ significantly from those in older results~\cite{Vigoda,CDMPP19}.
\begin{flipprob}
    \label[probabilities]{def:flip-probabilities}
    We use the following flip probabilities:
    \begin{align*}
P_1 := 1, \quad P_2 := 0.324, \quad & P_3 := 0.154,
\quad P_4 := 0.088, \quad P_5 := 0.044, \quad P_6 := 0.011,  \\
& \mbox{and } P_i := 0 \mbox{ for all } i \geq 7.
\end{align*}
We set the variable used in the definition of our metric as follows:
\begin{equation}
\label{def:eta}
    \eta := (P_2 - P_3)\frac{\Delta}{2k}.
\end{equation}
\end{flipprob}

Note, for the desired bound in \cref{thrm:flip-mixing}, $k\geq\lambbound\Delta$, our choice of $\eta$ satisfies
\begin{equation*}
    0\leq\eta=\frac{0.085\Delta}{k}
    \leq\frac{0.085}{\lambbound}
    <\frac12.
\end{equation*}

Rather than work directly with the values in \cref{def:flip-probabilities}, we record several useful consequences of this parameter setting. These properties are observations about our setting, chosen so that the case analysis in later sections can use a small collection of inequalities instead of repeatedly substituting the numerical values.
The flip probabilities in \cref{def:flip-probabilities} satisfy the following properties.

\begin{probprops}
\label[properties]{def:flip-properties}
\begin{align}
    &P_1=1,\qquad P_2\leq \frac13,\qquad P_7=0.
    \tag{FP0}\label[property]{eqn:line:basic-flip-bounds}
    \\
    &P_j\leq \frac23 P_{j-1}
    &&\mbox{for all $j\geq2$}.
    \tag{FP1}\label[property]{eqn:line:geometric-decay}
    \\
    &(j-1)P_j\leq P_2
    &&\mbox{for all $j\geq2$}.
    \tag{FP2}\label[property]{eqn:line:p2-bound}
    \\
    &(j-1)(P_j-P_{j+1})
    \leq
    (i-1)(P_i-P_{i+1})
    &&\mbox{for all $j\geq i\geq2$}.
    \tag{FP3}\label[property]{eqn:line:weighted-gap-monotonicity}
    \\
    &2P_2+4P_4\leq1.
    \tag{FP4}\label[property]{eqn:line:p2-p4-bound}
    \\
    &2P_3+P_5\leq4P_4.
    \tag{FP5}\label[property]{eqn:line:p3-p4-p5-bound}
    \\
    &3(P_2-P_3)\leq1-P_2+2(P_3-P_4).
    \tag{FP6}\label[property]{eqn:line:multiblocked-normal-form}
    \\
    &2P_3+P_6\leq P_2.
    \tag{FP7}\label[property]{eqn:line:p3-p6-bound}
\end{align}
\end{probprops}

Throughout the rest of the paper, we assume that the flip properties in \cref{def:flip-properties} hold.

Note that \cref{eqn:line:basic-flip-bounds,eqn:line:p2-bound}
imply
\begin{equation}
    jP_j\leq 1
    \qquad\mbox{for all $j\geq1$}.
    \label{eqn:j-pj-bound}
\end{equation}
Moreover, \cref{eqn:line:geometric-decay,eqn:line:p3-p4-p5-bound}
imply
\begin{equation}
    (j-2)P_j\leq 2P_4
    \qquad\mbox{for all $j\geq3$}.
    \label{eqn:shifted-pj-bound}
\end{equation}
Indeed, the case $j=3$ follows from \cref{eqn:line:p3-p4-p5-bound} and $P_5\geq0$. For $j\geq4$, the result follows by iterating \cref{eqn:line:geometric-decay}, since $(r-2)P_r\leq(r-3)P_{r-1}$ for every $r\geq5$.
Finally, we observe that
\begin{equation}
    j(P_j-P_{j+1})
    \leq
    i(P_i-P_{i+1})
    \qquad
    \mbox{for all $j\geq i\geq1$}.
    \label{eqn:weighted-gap-bound}
\end{equation}
For $j\geq i\geq2$, this follows from
\cref{eqn:line:geometric-decay,eqn:line:weighted-gap-monotonicity} together with
\begin{equation*}
    \frac{j}{j-1}\leq\frac{i}{i-1}.
\end{equation*}
If $i=j=1$, the inequality holds with equality. If $i=1$ and $j\geq2$, the case $i=2$ above and \cref{eqn:line:basic-flip-bounds} give
\begin{equation*}
    j(P_j-P_{j+1})
    \leq 2(P_2-P_3)
    \leq 2P_2
    \leq 1-P_2
    =
    P_1-P_2.
\end{equation*}
Recall that $P_0=0$.

Direct substitution verifies that \cref{def:flip-probabilities} satisfies all of \cref{def:flip-properties}. Since $P_j=0$ for every $j\geq7$, only finitely many indices need to be checked in \cref{eqn:line:geometric-decay}, \cref{eqn:line:p2-bound}, and \cref{eqn:line:weighted-gap-monotonicity}. In particular, \cref{eqn:line:p2-p4-bound,eqn:line:p3-p4-p5-bound} hold with equality.
Vigoda's original parameter setting satisfies every property except \cref{eqn:line:p2-bound,eqn:line:p3-p6-bound}, while the exact alternative-metric parameters and the unrounded linear programming (LP) parameters for the variable-length coupling of Chen et al. fail exactly \cref{eqn:line:p2-bound}, \cref{eqn:line:p2-p4-bound}, and \cref{eqn:line:p3-p6-bound}~\cite{Vigoda,CDMPP19}.

These properties are convenient observations tailored to our setting, rather than necessary conditions for the greedy coupling. The main difference used in our later case analysis is the $j=3$ case of \cref{eqn:line:p2-bound}, namely $2P_3\leq P_2$. Previous works take $P_3\approx1/6$, whereas our setting has $P_3<1/6$~\cite{Vigoda,CDMPP19}. This increases the expected change in Hamming distance for some of the extremal configurations in Vigoda's analysis, but our metric offsets this effect.

\subsection{ {Extremal and Canonical Configurations}}
\label{subsec:extremal-configurations}
In this section, we consider those configurations that determine the worst-case bounds in previous Hamming-distance analyses of the greedy coupling. We then describe the canonical configurations used in our analysis and explain how our metric and flip probabilities change the trade-off between them.
Here a \emph{configuration} for a fixed color $c$ records its local coupling data: the color degree $d_c$, the large-cluster sizes $A,B$, the coupling-size pairs $(a_i,b_i)$, and the blocking levels of the corresponding neighbors. We call configurations that attain the relevant worst-case bounds in a coupling analysis \emph{extremal}. These are bounds on the expected change, and need not equal the actual expected change.
In Vigoda's original analysis and the follow-up work of Chen et al., extremality is considered with respect to the Hamming distance~\cite{Vigoda,CDMPP19}.
The constraints associated with these extremal configurations are sufficient to determine the optimal flip parameters for the Hamming metric.

In Vigoda's original work, he identifies six extremal configurations up to symmetry:
\begin{itemize}
    \item $d_c = 1$, $2\leq a_1\leq5$, and $b_1 = 1$; and
    \item $d_c = 2$, $2\leq a_1=a_2\leq3$, and $b_1 = b_2 = 1$.
\end{itemize}
Optimizing over the constraints associated with these configurations suffices to derive Vigoda's flip parameters.

For example, consider the configuration $d_c=1$, $a_1=2$, $b_1=1$, $A=3$, and $B=2$.
This configuration represents a singly blocked neighbor, and its Hamming-distance upper bound is
\begin{equation*}
    (nk)Z_c = 1-P_3 = -1+(2-P_3).
\end{equation*}
Next, consider a configuration with $d_c=2$, $(a_1,a_2)=(3,3)$, $(b_1,b_2)=(1,1)$, $A=7$, and $B=3$ in which both neighbors are multiblocked and their blocking neighbors use the same disagreement color.
Its Hamming-distance upper bound is
\begin{equation*}
    (nk)Z_c
    =2+4P_3+P_7
    =-1+2\left(\frac{3}{2}+2P_3+\frac{P_7}{2}\right).
\end{equation*}
Consequently, if we seek a uniform bound of the form $(nk)Z_c\leq -1+\lambda d_c$, these two configurations require
\begin{equation*}
    \lambda\geq 2-P_3
    \qquad\mbox{and}\qquad
    \lambda\geq \frac{3}{2}+2P_3+\frac{P_7}{2}.
\end{equation*}
Since $P_7\geq0$, the best possible choice suggested by these two constraints has $P_7=0$.
Balancing the two bounds then gives $P_3=1/6$ and $\lambda=11/6$.

Chen et al.~\cite{CDMPP19} find a slightly different setting of the flip parameters for which, up to symmetry, only the two configurations above remain extremal.
Thus, their parameter setting still has $P_7=0$ and $P_3=1/6$.
However, they introduce a new metric that transfers some of the contribution from these two extremal configurations to the remaining configurations, whose per-neighbor contributions are bounded strictly below $11/6$.
This transfer allows their analysis to improve slightly on the $11/6$ threshold even though the two extremal Hamming configurations remain unchanged.

For our metric and flip parameters, six \emph{canonical configurations} suffice to bound the contributions,
up to symmetry: for $1\leq d_c\leq2$, every neighbor has
$(a_i,b_i)=(1,1)$, every neighbor has $(a_i,b_i)=(1,2)$, or every
neighbor has $(a_i,b_i)=(1,3)$, with the same orientation for both
neighbors when $d_c=2$. The first two choices describe unblocked
and singly blocked neighbors, respectively. The third can describe
multiblocked neighbors, and we use it for that case.
In these canonical configurations, $A=1+\sum_i a_i$ and $B=1+\sum_i b_i$.
The reduction concerns upper bounds on the expected changes; it does not assert that all six configurations attain equality in the final bounds.
Thus, four of these six canonical configurations coincide, up to symmetry,
with configurations in Vigoda's analysis. The two additional
configurations are the unblocked cases $a_i=b_i=1$ for all $i$, with
$d_c=1$ or $d_c=2$.

In Vigoda's original analysis,
the unblocked cases have relatively small expected contributions to Hamming distance.
They must also be considered here because an unblocked neighbor may become blocked, removing its discount in the metric and increasing its expected contribution.
Conversely, a singly blocked neighbor benefits from the possibility that it becomes unblocked.
Decreasing $P_3$ worsens the singly blocked Hamming case but improves the case of two multiblocked neighbors; both were originally tight in Vigoda's analysis.
Our new metric compensates for this change by transferring expected contribution from singly blocked configurations to unblocked configurations, improving the singly blocked case at the cost of worsening the unblocked case.
The amount of this transfer is governed by $\eta$ and is tuned together with our choice of flip probabilities.
Although a similar tuning can be applied to the metric of Chen et al., that metric does not directly isolate this local trade-off and seems to offer less room for improvement.
\subsection{ {Characterizing Blocking Level Configurations}}
\label{subsec:blocking-level-configurations}

In this subsection, we explore what unblocked, singly blocked, and multiblocked neighbors look like in terms of the clusters that contain them.
This will be helpful when we reduce the upper bounds for our new metric to the canonical configurations.

Throughout this subsection, fix $v^*\in V$ and $(X_t,Y_t)\in\Omega^2_{v^*}$.
For $c\in[k]\setminus\{X_t(v^*),Y_t(v^*)\}$, let $N_c(v^*)=\{u_1,\ldots,u_{d_c}\}$ and use the coupling sizes $a_i,b_i$ from \cref{subsec:greedy-coupling}.
A zero size can mean that the corresponding cluster $S$ is empty, or that its flip is counted under an earlier neighbor.

\begin{lemma}
\label[lemma]{lem:blocking-cluster-sizes}
For every $1\leq i\leq d_c$, the following hold.
\begin{enumerate}
\item The vertex $u_i$ is unblocked with respect to $v^*$ if and only if $a_i=b_i=1$.
\item If $u_i$ is singly blocked with respect to $v^*$, then, up to symmetry, $a_i=1$ and either $b_i=0$ or $b_i\geq2$.
Moreover, if $a_i=1$ and $b_i=2$, then $u_i$ is singly blocked with respect to $v^*$.
\item If $u_i$ is multiblocked with respect to $v^*$ and one of $a_i$ and $b_i$ equals $1$, then the other is either $0$ or at least $3$.
Moreover, if $a_i\neq1$ and $b_i\neq1$, then $u_i$ is multiblocked with respect to $v^*$.
\end{enumerate}
\end{lemma}

\begin{proof}
First note that a nonempty cluster whose flip is counted under an earlier neighbor contains two distinct vertices $u_i,u_j$ of the same color.
An alternating path between them has length at least two, so the cluster has at least three vertices.
Thus, assigning zero to a repeated representation leaves every cluster size equal to $1$ or $2$ unchanged, and every positive coupling size equals the size of its cluster $S$.

If $u_i$ is unblocked, both disagreement colors belong to $\colorList{u_i}$ and $u_i$ has no blocking neighbor, so both clusters are the singleton $\{u_i\}$ and $a_i=b_i=1$.
Conversely, $a_i=b_i=1$ implies that both clusters are flippable singletons, so both disagreement colors belong to $\colorList{u_i}$ and $u_i$ has no blocking neighbor.
Since $X_t(u_i)=Y_t(u_i)=c\notin\{X_t(v^*),Y_t(v^*)\}$, it follows that $u_i$ is unblocked.

Suppose that $u_i$ is singly blocked and, without loss of generality, its unique blocking neighbor $w$ has color $X_t(v^*)$.
Both disagreement colors belong to $\colorList{u_i}$, so $S_{X_t}(u_i,Y_t(v^*))=\{u_i\}$ and $a_i=1$.
Since $u_i$ is not unblocked, the first assertion gives $b_i\neq1$, so either $b_i=0$ or $b_i\geq2$.
Conversely, if $a_i=1$ and $b_i=2$, both clusters are flippable and have underlying vertex sets $\{u_i\}$ and $\{u_i,w\}$ for some $w\neq v^*$.
Thus, both disagreement colors belong to $\colorList{u_i}$ and $w$ is the unique blocking neighbor of $u_i$, so $u_i$ is singly blocked.

If $u_i$ is multiblocked and one size is $1$, the first two assertions rule out the other size being $1$ or $2$.
If neither size is $1$, the vertex cannot be unblocked or singly blocked.
\end{proof}

\section{Coupling Analysis of Our Metric}
\label{sec:coupling-analysis}
To prove \cref{thrm:flip-mixing} using our metric and the path coupling theorem (\cref{thrm:path-coupling}), it remains to bound the expected change in the metric for any two neighboring list labelings.
We do this by decomposing the expected change by color and then bounding each contribution.

Fix $v^*\in V$ and $(X_t, Y_t) \in \Omega^2_{v^*}$.
Write $c_X=X_t(v^*)$, $c_Y=Y_t(v^*)$, and $d_c=d_c(v^*)=|N_c(v^*)|$ for every $c\in[k]$.
All expectations in this section are with respect to the one-step greedy coupling conditional on the fixed pair $(X_t,Y_t)$; to simplify notation we suppress this conditioning.

Recall from \cref{def:hamming-change} that, under this convention,
\begin{equation*}
    W = \Exp{H(X_{t+1},Y_{t+1})-H(X_t,Y_t)}.
\end{equation*}
Define the analogous quantity for our metric by
\begin{equation}
\label{def:weighted-change}
    \widetilde{W} := \Exp{\nmet(X_{t+1},Y_{t+1})-\nmet(X_t,Y_t)}.
\end{equation}
Thus, $W$ and $\widetilde{W}$ are the expected changes in the Hamming distance and our metric, respectively, over the coupled transition $(X_t,Y_t)\rightarrow(X_{t+1},Y_{t+1})$.

\subsection{Per-Color Analysis}
\label{subsec:color-analysis}
We decompose the expected change in the Hamming distance $W$ and our metric~$\widetilde{W}$ with respect to each color $c\in [k]$.

For an integer $s\geq0$, a pair $X_s,Y_s\in\Omega$, a color $c\in[k]$, and a vertex $z\in V$, let
\begin{equation*}
    N_{s,c}(z) := N_c(z,X_s) \cup N_c(z,Y_s).
\end{equation*}
Thus, $N_{s,c}(z)$ is the set of neighbors of $z$ that have color $c$ in at least one of the two chains.
For an integer $0\leq i\leq2$, define
\begin{equation*}
    d^i_{s,c}(v^*) := \lvert N_{s,c}(v^*) \cap F^i_s(v^*) \rvert.
\end{equation*}
This is the number of neighbors of $v^*$ that have color $c$ in at least one of the two chains and have blocking level $i$ with respect to $v^*$.

Since $(X_t,Y_t)\in\Omega^2_{v^*}$, the two chains agree on every neighbor of $v^*$.
Consequently, $N_{t,c}(v^*)=N_c(v^*,X_t)=N_c(v^*,Y_t)$, and the sets $(N_{t,c}(v^*))_{c\in[k]}$ form a partition of $N(v^*)$.

Recall from \cref{def:hamming-change} that, under the suppressed conditioning,
\begin{equation*}
    W_c = \Exp{\left(H(X_{t+1},Y_{t+1})-H(X_t,Y_t)\right)\indicator(\mathcal{C}_c)}.
\end{equation*}
Here, $\mathcal{C}_c$ is the event that the coupling selects a positive-probability coupled move allocated to $\mathcal{D}_c$.
For every color $c\in [k]$ such that $d_c\geq 1$, define
\begin{equation}
\label{def:per-color-weighted-change}
    \widetilde{W}_c
    :=
    W_c
    -\frac{\eta}{\Delta}
    \sum_{u\in N_{t,c}(v^*)}
    \Exp{\degree{0}{t+1}{u}}
    -\frac{\eta}{\Delta}
    \left(
    \Exp{
    \left|
    \unblockedShort{t+1}{v^*}\cap N_{t,c}(v^*)
    \right|
    }
    -\degree{0}{t,c}{v^*}
    \right).
\end{equation}

In the definition of $\widetilde{W}_c$, the first term is the expected change in Hamming distance allocated to $\mathcal C_c$.
The second term accounts for the expected number of unblocked neighbors of each $c$-colored neighbor $u$ of $v^*$ if $u$ becomes a disagreement; recall that $\degree{0}{t+1}{u}=0$ if $X_{t+1}(u)=Y_{t+1}(u)$.
The final term is the expected change in the number of neighbors in the time-$t$ color class $N_{t,c}(v^*)$ that are unblocked with respect to $v^*$.
Note that, unlike $W_c$, $\widetilde{W}_c$ is not exclusively the expected change in $\nmet$ caused by the event $\mathcal C_c$.
Namely, it also includes contributions from events that change the blocking levels of neighbors that have color $c$ at time $t$.

The per-color bounds on $\widetilde{W}_c$ combine to bound the total expected change as follows.
We apply \cref{lem:metric-decomposition}, retain the discounts at $v^*$ and its neighbors, and then group these terms by the neighbors' colors at time $t$.
\begin{lemma}
\label[lemma]{lem:per-color-decomposition}
For $(X_t,Y_t)\in\Omega^2_{v^*}$,
\begin{equation*}
    \widetilde{W}\leq\sum_{\substack{c\in[k] \\ d_c=0}}W_c
    +\sum_{\substack{c\in[k] \\ d_c\geq1}}\widetilde{W}_c.
\end{equation*}
\end{lemma}

\begin{proof}
By \cref{def:weighted-change,lem:metric-decomposition,def:vertex-metric},
\begin{align}
    \nmet(X_{t+1},Y_{t+1}) &\leq \sum_{z \in V} \nmet_z(X_{t+1}, Y_{t+1}) \nonumber \\
    &= \sum_{z \in V} \left( 1 -\frac{\eta}{\Delta} \degree{0}{t+1}{z}\right)\indicator(X_{t+1}(z) \neq Y_{t+1}(z)) \nonumber \\
    &= H(X_{t+1}, Y_{t+1}) - \frac{\eta}{\Delta} \sum_{z \in V} \degree{0}{t+1}{z} \nonumber\\
    &\leq H(X_{t+1}, Y_{t+1}) - \frac{\eta}{\Delta} \left( \degree{0}{t+1}{v^*} + \sum_{z \in N(v^*)} \degree{0}{t+1}{z}\right) \label{eqn:next-step-metric-bound}
\end{align}
where the second equality follows since $\degree{0}{t+1}{z} = 0$ if $X_{t+1}(z) = Y_{t+1}(z)$ and the last inequality follows since $\eta \geq 0$.

Since $(X_t,Y_t) \in \Omega^2_{v^*}$, by \cref{eqn:neighboring-metric},
\begin{equation}
    \nmet(X_t,Y_t) = H(X_t,Y_t) - \frac{\eta}{\Delta}\degree{0}{t}{v^*}.
    \label{eqn:current-metric-identity}
\end{equation}

We can now use \cref{eqn:next-step-metric-bound,eqn:current-metric-identity} to write
\begin{align}
    \widetilde{W} &=  \Exp{\nmet(X_{t+1},Y_{t+1})-\nmet(X_t,Y_t)}
    \nonumber \\
    &\leq \Exp{H(X_{t+1}, Y_{t+1}) - \frac{\eta}{\Delta} \left( \degree{0}{t+1}{v^*} + \sum_{z \in N(v^*)} \degree{0}{t+1}{z}\right)} - \nmet(X_t,Y_t)
    \tag{\mbox{by \cref{eqn:next-step-metric-bound}}} \nonumber\\
    &= \Exp{H(X_{t+1}, Y_{t+1}) - \frac{\eta}{\Delta} \left( \degree{0}{t+1}{v^*} + \sum_{z \in N(v^*)} \degree{0}{t+1}{z}\right)} - H(X_t,Y_t) + \frac{\eta}{\Delta}\degree{0}{t}{v^*} \tag{\mbox{by \cref{eqn:current-metric-identity}}} \nonumber \\
    &= \Exp{H(X_{t+1}, Y_{t+1}) - H(X_t,Y_t)} + \frac{\eta}{\Delta} \Exp{\degree{0}{t}{v^*} - \degree{0}{t+1}{v^*} - \sum_{z \in N(v^*)} \degree{0}{t+1}{z} } \nonumber \\
    &= \sum_{c \in [k]} W_c + \frac{\eta}{\Delta} \Exp{\degree{0}{t}{v^*} - \degree{0}{t+1}{v^*} - \sum_{z \in N(v^*)} \degree{0}{t+1}{z} }
    \label{eqn:weighted-change-decomposition}
\end{align}
where the last equality follows from \cref{eqn:hamming-color-decomposition}.

Since $(X_t,Y_t)\in\Omega^2_{v^*}$, the sets $N_{t,c}(v^*)$ form a partition of $N(v^*)$.
In particular,
\begin{equation*}
    N(v^*) = \bigsqcup_{c\in[k]}N_{t,c}(v^*).
\end{equation*}
Hence, we can write
\begin{equation}
    \degree{0}{t}{v^*} = \sum_{c \in [k]} \degree{0}{t,c}{v^*},
    \label{eqn:current-unblocked-decomposition}
\end{equation}
and
\begin{equation}
    \degree{0}{t+1}{v^*} = \left| \unblockedShort{t+1}{v^*}\right|  = \sum_{\substack{c \in [k] \\ d_c \geq 1}} \left| \unblockedShort{t+1}{v^*} \cap N_{t,c}(v^*) \right|.
    \label{eqn:next-unblocked-decomposition}
\end{equation}
We can use \cref{eqn:current-unblocked-decomposition,eqn:next-unblocked-decomposition} to rewrite \cref{eqn:weighted-change-decomposition} as
\begin{align*}
    \widetilde{W}
    &\leq \sum_{c\in[k]} W_c
    - \frac{\eta}{\Delta}
    \sum_{\substack{c\in[k]\\d_c\geq1}}
    \Exp{
    \sum_{u\in N_{t,c}(v^*)}\degree{0}{t+1}{u}
    +\left|\unblockedShort{t+1}{v^*}\cap N_{t,c}(v^*)\right|
    -\degree{0}{t,c}{v^*}
    }
    \\
    &= \sum_{\substack{c\in[k]\\d_c=0}} W_c
    + \sum_{\substack{c\in[k]\\d_c\geq1}}
    \left(
    W_c
    - \frac{\eta}{\Delta}
    \Exp{\sum_{u\in N_{t,c}(v^*)}\degree{0}{t+1}{u}}
    - \frac{\eta}{\Delta}
    \left(
    \Exp{\left|\unblockedShort{t+1}{v^*}\cap N_{t,c}(v^*)\right|}
    -\degree{0}{t,c}{v^*}
    \right)
    \right)
    \\
    &= \sum_{\substack{c\in[k]\\d_c=0}} W_c
    + \sum_{\substack{c\in[k]\\d_c\geq1}} \widetilde{W}_c
\end{align*}
as desired.
\end{proof}

\subsection{Vertices Changing Between Blocked and Unblocked}
\label{subsec:blocking-transitions}
Our goal now is to bound $\widetilde{W}_c$ for all colors $c\in[k]$ with $d_c\geq1$.
Throughout the remainder of this section, we assume that
$\lvert\colorList{z}\rvert\geq d(z)+k-\Delta$ for every $z\in V$.

We first introduce some helpful notation.
For every integer $s\geq0$, vertex $z\in V$, and color $c\in[k]$, let
\begin{equation*}
    d_{s,c}(z) := \lvert N_{s,c}(z)\rvert.
\end{equation*}
For $i \geq 1$, let
\begin{equation*}
    \Delta_i
    := \left\lvert
    \left\{
    u \in N(v^*):
    \degree{}{t,X_t(u)}{v^*}=i
    \mbox{ and }
    X_t(u)\in\colorList{v^*}
    \right\}
    \right\rvert
\end{equation*}
denote the number of neighbors whose color is in $\colorList{v^*}$ and appears exactly $i$ times in the neighborhood of $v^*$.
Finally, for every $z\in V$, let
\begin{equation*}
    \mathcal{A}_t(z) := \{c\in\colorList{z}:\degree{}{t,c}{z}=0\}
\end{equation*}
be the set of available colors for $z$ in $X_t$ and $Y_t$.

We next define the following function, which will serve as an upper bound on $\widetilde{W}_c$, analogous to how $Z_c$ bounds $W_c$.
\begin{definition}
\label[definition]{def:weighted-upper-bound}
        Let
        \begin{equation}
            \label{def:alpha}
            \alpha := (k -\Delta - 2)
        \end{equation}
        and
        \begin{equation}
        \label{def:beta}
            \beta := k - P_2 \left(\Delta_{1} + \Delta_{2} \right) +(1+P_2) \Delta.
        \end{equation}
        Then for $(X_t,Y_t) \in \Omega^2_{v^*}$ and a nondegenerate color $c\in[k]\setminus\{c_X,c_Y\}$ such that $d_c\geq1$, let
        \begin{equation}
            \label{def:tilde-z-c-upper-bound}
            \widetilde{Z}_c := Z_c + \frac{1}{nk}
            \left(
            \degree{0}{t,c}{v^*}\frac{\eta\beta}{\Delta}
            -\degree{1}{t,c}{v^*}\frac{\eta\alpha}{\Delta}
            \right).
        \end{equation}
\end{definition}

For the remainder of this section, we also assume that
\begin{equation}
    \alpha\geq0,
    \qquad
    \frac{\eta\beta}{\Delta}\leq P_2.
    \label{eqn:auxiliary-parameter-bounds}
\end{equation}
These bounds hold for the parameters in \cref{def:flip-probabilities}
when $\Delta\geq\deltabound$ and $k\geq\lambbound\Delta$.
Indeed,
$\alpha=k-\Delta-2\geq(\lambbound-1)\Delta-2>0$,
and $\Delta_1+\Delta_2\leq\Delta$ gives
$0<k+\Delta\leq\beta\leq k+(1+P_2)\Delta$. Consequently,
\begin{equation*}
    \frac{\eta\beta}{\Delta}
    \leq\frac{P_2-P_3}{2}
    \left(1+\frac{1+P_2}{k/\Delta}\right)
    \leq\frac{0.17}{2}
    \left(1+\frac{1.324}{\lambbound}\right)
    <P_2.
\end{equation*}

We refer to the second term on the right-hand side of \cref{def:tilde-z-c-upper-bound} as the \emph{correction}.
The positive term allows for initially unblocked neighbors losing their discount, while the negative term accounts for singly blocked neighbors becoming unblocked and gaining a discount.

The following lemma shows that, for a nondegenerate color, it suffices to bound $\widetilde{Z}_c$ instead of $\widetilde{W}_c$ directly.

\begin{lemma}
\label[lemma]{lem:weighted-upper-bound}
For $(X_t,Y_t)\in\Omega^2_{v^*}$ and $c\in[k]\setminus\{c_X,c_Y\}$ such that $d_c\geq1$,
\begin{equation*}
    \widetilde{W}_c\leq\widetilde{Z}_c.
\end{equation*}
\end{lemma}

To prove \cref{lem:weighted-upper-bound}, we use the following two lemmas that bound the expected number of vertices colored $c$ in $X_t$ and $Y_t$ that become blocked or unblocked after a single step of the greedy coupling.

We first give a lower bound on the number of newly unblocked neighbors of $v^*$.
In particular, for a color $c$, we lower bound the probability that a $c$-colored neighbor $u$ of $v^*$ becomes unblocked after its unique blocking neighbor is recolored.

\begin{lemma}
\label[lemma]{lem:newly-unblocked-bound}
For $(X_t,Y_t)\in\Omega^2_{v^*}$ and $c\in[k]$,
    \begin{equation*}
    \Exp{ \lvert (\unblockedShort{t+1}{v^*} \setminus \unblockedShort{t}{v^*}) \cap N_{t,c}(v^*) \rvert}  \geq \degree{1}{t,c}{v^*} \frac{\alpha}{nk}.
    \end{equation*}
\end{lemma}

\begin{proof}
If $\alpha\leq0$, then the claim is immediate, so assume that $\alpha>0$.
Let $u \in \singleBlockedShort{t}{v^*}\cap N_{t,c}(v^*)$.
Since $u$ is a singly blocked neighbor of $v^*$, there is a unique neighbor $w \in N(u)\setminus \{v^*\}$ such that
\begin{equation*}
    X_t(w)=Y_t(w)\in\{X_t(v^*),Y_t(v^*)\}.
\end{equation*}
Since $X_t$ and $Y_t$ differ only at $v^*$, at most $d(w)$ colors outside $\{X_t(v^*),Y_t(v^*)\}$ appear in the neighborhood of $w$ across the two chains.
There are at least
\begin{equation*}
    \left\lvert \mathcal{A}_t(w) \setminus \{X_t(v^*),Y_t(v^*)\} \right\rvert
    \geq \lvert \colorList{w} \rvert-d(w)-2
    \geq k-\Delta-2
    =\alpha
\end{equation*}
colors other than $X_t(v^*)$ and $Y_t(v^*)$ that do not appear in the neighborhood of $w$ in either $X_t$ or $Y_t$.

Fix a color $c' \in \mathcal{A}_t(w) \setminus \{X_t(v^*),Y_t(v^*)\}$.
Since $c'$ does not appear in the neighborhood of $w$ in either chain, and $c'\in\colorList{w}$, we have
\begin{equation*}
    S_{X_t}(w,c')=S_{Y_t}(w,c')=\{w\}.
\end{equation*}
These are the same cluster in both chains, including the same unordered color pair, so the greedy coupling uses an identity pairing and flips this common singleton cluster with probability $P_1/(nk)=1/(nk)$.
When this cluster flips, $X_{t+1}(w)=Y_{t+1}(w)=c'$, and all other vertices remain unchanged.
Thus, the unique blocking neighbor of $u$ is recolored while $v^*$ remains unchanged.
The disagreement colors at $v^*$ are still contained in $\colorList{u}$, the common color of $u$ remains outside the two disagreement colors, and $u$ has no remaining blocking neighbor.
Hence, all three conditions in the definition of an unblocked neighbor hold and $u\in\unblockedShort{t+1}{v^*}$.

The coupled moves corresponding to distinct choices of $c'$ are disjoint.
It follows that, for every $u\in\singleBlockedShort{t}{v^*}\cap N_{t,c}(v^*)$,
\begin{equation*}
    \Prob{u\in\unblockedShort{t+1}{v^*}}
    \geq\frac{\alpha}{nk}.
\end{equation*}
Finally, by linearity of expectation,
\begin{align*}
    \Exp{ \lvert (\unblockedShort{t+1}{v^*} \setminus \unblockedShort{t}{v^*}) \cap N_{t,c}(v^*) \rvert}
    &\geq
    \sum_{u\in\singleBlockedShort{t}{v^*}\cap N_{t,c}(v^*)}
    \Prob{u\in\unblockedShort{t+1}{v^*}}
    \\
    &\geq
    \degree{1}{t,c}{v^*}\frac{\alpha}{nk}.
\end{align*}
\end{proof}

We now want to bound the probability that an unblocked neighbor $u$ of $v^*$ is no longer unblocked after a single step of the greedy coupling.
Note that it does not suffice to only consider the probability $u$ becomes blocked because there may no longer be a disagreement at $v^*$ after a step of the greedy coupling, which results in $u$ being neither blocked nor unblocked (by definition).

An unblocked neighbor $u$ may be relatively likely to lose its discount at $v^*$.
In these cases, a new disagreement at $u$ can have many unblocked neighbors and thus receive a smaller weight in our metric.
We therefore bound the loss of the original discount and the gain from these new discounts together, using the following two terms.
\begin{equation}
\label{eqn:blocking-tradeoff}
\Exp{ \lvert (\unblockedShort{t}{v^*} \setminus \unblockedShort{t+1}{v^*}) \cap N_{t,c}(v^*) \rvert }
\mbox{ and } \sum_{u \in N_{t,c}(v^*)}\Exp{ \degree{0}{t+1}{u}}.
\end{equation}
Note that $(\unblockedShort{t}{v^*} \setminus \unblockedShort{t+1}{v^*}) \cap N_{t,c}(v^*)$ includes those neighbors $w \in \unblockedShort{t}{v^*} \cap N_{t,c}(v^*)$ such that either: (i) $X_{t+1}(v^*)\neq Y_{t+1}(v^*)$ and $w$ is singly blocked or multiblocked with respect to $v^*$ at time $t+1$; or (ii) $X_{t+1}(v^*)=Y_{t+1}(v^*)$, in which case $w$ has no blocking level with respect to $v^*$ under our convention.
Hence, the first term of \cref{eqn:blocking-tradeoff} is the expected number of $c$-colored vertices in the neighborhood of $v^*$ that are initially unblocked but are no longer unblocked after one step.
The summand $\Exp{ \degree{0}{t+1}{u}}$ in the second term of \cref{eqn:blocking-tradeoff} is the expected number of unblocked neighbors of a vertex $u$ in $(X_{t+1}, Y_{t+1})$ counting zero whenever $u$ is not a disagreement at time $t+1$, as required by the definition.
This trade-off is central to our improved bound.

\begin{lemma}
\label[lemma]{lem:blocked-unblocked-tradeoff}
For $(X_t,Y_t)\in\Omega^2_{v^*}$ and $c\in[k]$,
\begin{equation*}
    \Exp{ \lvert (\unblockedShort{t}{v^*} \setminus \unblockedShort{t+1}{v^*} ) \cap N_{t,c}(v^*) \rvert }
    - \sum_{u \in N_{t,c}(v^*)}\Exp{ \degree{0}{t+1}{u}}
    \leq \degree{0}{t,c}{v^*}\frac{\beta}{nk}.
\end{equation*}
\end{lemma}

\begin{proof}
Let $c_X=X_t(v^*)$, $c_Y=Y_t(v^*)$, and $C_{v^*}=\{c_X,c_Y\}$.
First suppose that $c\in C_{v^*}$.
Every vertex in $N_{t,c}(v^*)$ is then multiblocked, so $\degree{0}{t,c}{v^*}=0$ and
$\unblockedShort{t}{v^*}\cap N_{t,c}(v^*)=\emptyset$.
The first term on the left-hand side is therefore zero, while the second term is nonpositive.
Thus, the result is immediate in this case.

We may consequently assume that $c\notin C_{v^*}$.
It suffices to show that, for every $u \in \unblockedShort{t}{v^*}\cap N_{t,c}(v^*)$,
\begin{equation}
    \label{eqn:per-vertex-tradeoff}
    \Prob{u\notin\unblockedShort{t+1}{v^*}}
    -\Exp{\degree{0}{t+1}{u}}
    \leq\frac{\beta}{nk}.
\end{equation}
Indeed, summing this inequality over all $u \in \unblockedShort{t}{v^*}\cap N_{t,c}(v^*)$ and subtracting the nonnegative terms $\Exp{\degree{0}{t+1}{u}}$ for $u \in N_{t,c}(v^*) \setminus \unblockedShort{t}{v^*}$ gives the lemma.

Fix $u \in \unblockedShort{t}{v^*}\cap N_{t,c}(v^*)$.
Thus, $u$ is an unblocked neighbor of $v^*$ and $X_t(u)=Y_t(u)=c$.
Let $\EE_{v^*}$ be the event that $v^*$ is recolored in at least one chain (i.e., $X_{t+1}(v^*) \neq X_t(v^*)$ or $Y_{t+1}(v^*) \neq Y_t(v^*)$) and let $\overline{\EE_{v^*}}$ be the event that $v^*$ is not recolored in either chain (i.e., $X_{t+1}(v^*) = X_t(v^*)$ and $Y_{t+1}(v^*) = Y_t(v^*)$).
Let $\EE_u$ be the event that $\overline{\EE_{v^*}}$ occurs and $u$ receives one of the colors in $C_{v^*}$ in at least one chain, i.e.,
\begin{equation*}
    \EE_u :=
    \overline{\EE_{v^*}}
    \cap
    \left\{
    \{X_{t+1}(u),Y_{t+1}(u)\}
    \cap
    C_{v^*}\neq\emptyset
    \right\}.
\end{equation*}
For every $w\in N(u)\setminus\{v^*\}$, let $\EE_w$ be the event that $\overline{\EE_{v^*}}$ occurs and $w$ receives one of the colors in $C_{v^*}$ in at least one chain, i.e.,
\begin{equation*}
    \EE_w :=
    \overline{\EE_{v^*}}
    \cap
    \left\{
    \{X_{t+1}(w),Y_{t+1}(w)\}
    \cap
    C_{v^*}\neq\emptyset
    \right\}.
\end{equation*}
These events record the possible ways that $u$ can cease to be unblocked.
Either the pair of colors at $v^*$ changes, $u$ itself receives a color in $C_{v^*}$, or some other neighbor of $u$ receives a color in $C_{v^*}$.
Hence,
\begin{equation*}
    \{u\notin\unblockedShort{t+1}{v^*}\}
    \subseteq
    \EE_{v^*}\cup\EE_u
    \cup\bigcup_{w\in N(u)\setminus\{v^*\}}\EE_w.
\end{equation*}
Moreover, $\Exp{\degree{0}{t+1}{u}}=\sum_{w\in N(u)}\Prob{w\in\unblockedShort{t+1}{u}}$.
Consequently, a union bound and the omission of the nonpositive term corresponding to $w=v^*$ give
\begin{align*}
    &\Prob{u\notin\unblockedShort{t+1}{v^*}}
    -\Exp{\degree{0}{t+1}{u}}
    \\
    &\qquad\leq
    \Prob{\EE_{v^*}}+\Prob{\EE_u}
    +\sum_{w\in N(u)\setminus\{v^*\}}
    \left(
    \Prob{\EE_w}
    -\Prob{w\in\unblockedShort{t+1}{u}}
    \right).
\end{align*}

Since $u$ is unblocked with respect to $v^*$, its color $c$ is not in $C_{v^*}$ and the clusters $S_{X_t}(u,c_Y)$ and $S_{Y_t}(u,c_X)$ are both singletons.
While $\overline{\EE_{v^*}}$ occurs, the event $\EE_u$ can only result from residual flips of these two small clusters.
Let $Q/(nk)$ and $Q'/(nk)$ be the residual marginal probabilities of flipping $S_{X_t}(u,c_Y)$ and $S_{Y_t}(u,c_X)$, respectively, after each is paired with its corresponding large cluster.
The coupling flips the two smaller clusters together with probability $\min\{Q,Q'\}/(nk)$ and assigns any remaining mass to a one-sided flip.
Therefore, the probability that at least one of these smaller clusters is flipped is $\max\{Q,Q'\}/(nk)$, and hence
\begin{equation*}
    \Prob{\EE_u}=\frac{\max\{Q,Q'\}}{nk}\leq\frac{1}{nk}.
\end{equation*}

We now bound $\Prob{\EE_{v^*}}$, which is the probability that $v^*$ is recolored in at least one of the chains.
If $X_{t+1}(v^*) \neq X_t(v^*)$ or $Y_{t+1}(v^*) \neq Y_t(v^*)$ then a cluster containing~$v^*$ must have flipped in $X_t$ or $Y_t$.
For every color $c'\in\mathcal{A}_t(v^*)$, the $d_{c'}=0$ coupling has one event of probability $1/(nk)$ that recolors $v^*$ in at least one chain.
For every $c'\in\colorList{v^*}$ that appears in the neighborhood of $v^*$, the two candidate clusters are $S_{X_t}(v^*,c')$ and $S_{Y_t}(v^*,c')$.
Each is either empty or contains $v^*$ together with every vertex in $N_{t,c'}(v^*)$, and therefore has size $0$ or at least $\degree{}{t,c'}{v^*}+1$.
Using $P_0=0$ and the monotonicity of the flip probabilities, we obtain
\begin{align}
    \Prob{\EE_{v^*}}
    &\leq \frac{1}{nk}
    \left(
    \lvert \mathcal{A}_t(v^*) \rvert
    +
    \sum_{\substack{c'\in\colorList{v^*}\\\degree{}{t,c'}{v^*}\geq1}}
    \left(
    P_{\lvert S_{X_t}(v^*,c')\rvert}
    +
    P_{\lvert S_{Y_t}(v^*,c')\rvert}
    \right)
    \right)
    \nonumber\\
    &\leq \frac{1}{nk}
    \left(
    \lvert\colorList{v^*}\rvert
    -(1-2P_2)\Delta_1
    -\frac{1-2P_3}{2}\Delta_2
    \right)
    \nonumber\\
    &\leq \frac{1}{nk}
    \left(
    \lvert\colorList{v^*}\rvert
    -P_2(\Delta_1+\Delta_2)
    \right).
    \label{eqn:line:vertex-change-bound}
\end{align}
The last inequality holds because $P_2\leq 1/3$ by
\cref{eqn:line:basic-flip-bounds} and $P_2+P_3\leq1/2$, which follows by combining
\cref{eqn:line:p2-p4-bound,eqn:line:p3-p4-p5-bound} and using $P_5\geq0$.

We will repeatedly use the following observation.
For a fixed vertex $w\neq v^*$ and a fixed target color $c'$, the components $\widehat S_{X_t}(w,c')$ and $\widehat S_{Y_t}(w,c')$ are identical unless at least one of them contains $v^*$, since the two labelings agree away from $v^*$.
If neither component contains $v^*$, then $S_{X_t}(w,c')$ and $S_{Y_t}(w,c')$ are also identical, including the possibility that both are empty because of the same list obstruction, and the coupling pairs their flips together.
If one of the clusters contains $v^*$, then it cannot flip while $\overline{\EE_{v^*}}$ occurs.

We now show that, for every $w \in N(u)\setminus\{v^*\}$,
\begin{equation}
    \label{eqn:neighbor-tradeoff}
    \Prob{\EE_w}
    - \Prob{w \in \unblockedShort{t+1}{u}}
    \leq \frac{1+P_2}{nk}.
\end{equation}
Assuming \cref{eqn:neighbor-tradeoff}, the right-hand side of the preceding union bound is at most
\begin{equation*}
    \frac{k-P_2(\Delta_1+\Delta_2)+1+(d(u)-1)(1+P_2)}{nk}
    \leq
    \frac{k-P_2(\Delta_1+\Delta_2)+(1+P_2)\Delta}{nk}
    =\frac{\beta}{nk},
\end{equation*}
where we used \cref{eqn:line:vertex-change-bound}, $\lvert\colorList{v^*}\rvert\leq k$, $d(u)\leq\Delta$, and $1+(d(u)-1)(1+P_2)\leq(1+P_2)\Delta$.
This proves \cref{eqn:per-vertex-tradeoff}, subject to \cref{eqn:neighbor-tradeoff}.

It remains to prove the local trade-off.
Let $w \in N(u)\setminus\{v^*\}$.
We consider four cases:

\begin{enumerate}
    \item[(i)] $c_X \notin \colorList{w}$ or $c_Y \notin \colorList{w}$;
    \item[(ii)] $C_{v^*} \subseteq \colorList{w}$ and $w \in N(v^*)$;
    \item[(iii)] $C_{v^*} \subseteq \colorList{w}$, $w \notin N(v^*)$, and $w$ has at least one neighbor $z \neq u$ whose color in $X_t$ and $Y_t$ belongs to $C_{v^*}$; and
    \item[(iv)] $C_{v^*} \subseteq \colorList{w}$ and $w$ has no neighbors in $X_t$ or $Y_t$ whose color belongs to $C_{v^*}$.
\end{enumerate}

Suppose case (i) occurs; hence, $c_X \notin \colorList{w}$ or $c_Y \notin \colorList{w}$.
Since $u$ is unblocked with respect to $v^*$, we have $X_t(w)=Y_t(w)\notin C_{v^*}$.
For any color $c'\in[k]\setminus\{X_t(w)\}$, the probability that $w$ is recolored to $c'$ in at least one chain while $\overline{\EE_{v^*}}$ occurs is at most $1/(nk)$.
Indeed, this follows from the preceding observation about corresponding clusters.
Moreover, this probability is zero if $c'\notin\colorList{w}$.
Thus, $\Prob{\EE_w}\leq1/(nk)$ because at most one of the two disagreement colors is in $\colorList{w}$.
Hence, \cref{eqn:neighbor-tradeoff} holds in this case since $\Prob{w \in \unblockedShort{t+1}{u}} \geq 0$.

Suppose case (ii) occurs; hence, $C_{v^*} \subseteq \colorList{w}$ and $w \in N(v^*)$.
Since $w \in N(v^*)$, both $\widehat S_{X_t}(w,c_X)$ and $\widehat S_{Y_t}(w,c_Y)$ contain $v^*$. Thus, each of $S_{X_t}(w,c_X)$ and $S_{Y_t}(w,c_Y)$ is either empty or contains $v^*$.
Moreover, if the greedy coupling flips either of these clusters, it prevents $\overline{\EE_{v^*}}$ from occurring.
Consequently, while $\overline{\EE_{v^*}}$ occurs, $\EE_w$ can only result from residual flips of $S_{X_t}(w,c_Y)$ and $S_{Y_t}(w,c_X)$.
If the residual flips of $S_{X_t}(w,c_Y)$ and $S_{Y_t}(w,c_X)$ are retained under the same neighbor index, then the greedy coupling pairs their flips as much as possible, and their union has probability at most $1/(nk)$.
Otherwise, at least one is retained under the index of another neighbor of $v^*$.
Such a repeated cluster contains $w$ and another vertex of the same color as $w$, so it has size at least three by the argument in \cref{lem:blocking-cluster-sizes}.
Its flip probability is therefore at most $P_3/(nk)\leq P_2/(nk)$, while the other residual cluster has flip probability at most $1/(nk)$.
Hence, $\Prob{\EE_w}\leq(1+P_2)/(nk)$, which proves \cref{eqn:neighbor-tradeoff} in this case.

Suppose case (iii) occurs; hence, $C_{v^*} \subseteq \colorList{w}$, $w \notin N(v^*)$, and $w$ has at least one neighbor $z \neq u$ whose color belongs to $C_{v^*}$.
Since $z \neq v^*$, $X_t(z) = Y_t(z)$.
Moreover, since $u \in \unblockedShort{t}{v^*}$, $X_t(w) = Y_t(w) \notin C_{v^*}$.
Let $c' = X_t(z)$ and assume without loss of generality that $c' = c_X$.
The clusters $S_{X_t}(w,c')$ and $S_{Y_t}(w,c')$ are each empty or contain both $w$ and $z$, and hence each flips with probability at most $P_2/(nk)$.
If these clusters agree, the identity coupling recolors $w$ to $c'$ in both chains with probability at most $P_2/(nk)$.
If exactly one is empty, only the other can flip, again with probability at most $P_2/(nk)$.
If they are distinct and both nonempty, one must contain $v^*$ because $X_t$ and $Y_t$ differ only at $v^*$.
In fact, $v^*\notin S_{Y_t}(w,c')$ because $c'=c_X$ and $c_Y\notin\{X_t(w),c'\}$, so $v^*\in S_{X_t}(w,c')$.
Thus, under $\overline{\EE_{v^*}}$, the probability that $w$ receives $c'$ in at least one chain is at most $P_2/(nk)$.
For the other disagreement color, the two candidate clusters likewise agree, or one is empty, or one contains $v^*$.
Since each candidate flips with probability at most $1/(nk)$, the probability that $w$ receives this other color while $\overline{\EE_{v^*}}$ occurs is at most $1/(nk)$.
Therefore, $\Prob{\EE_w}\leq(1+P_2)/(nk)$, and \cref{eqn:neighbor-tradeoff} holds in this case.

Finally, suppose case (iv) occurs; hence, $C_{v^*} \subseteq \colorList{w}$ and $w$ has no neighbors in $X_t$ or $Y_t$ whose color belongs to $C_{v^*}$.
Since $u \in \unblockedShort{t}{v^*}$ and $w \in N(u) \setminus \{v^*\}$, $X_t(w) = Y_t(w) \notin C_{v^*}$.
Hence, in this case, $S_{Y_t}(w,c_X)$ and $S_{X_t}(w,c_Y)$ are both $\{w\}$ and the greedy coupling flips each with probability $1/(nk)$.
Therefore, a union bound gives $\Prob{\EE_w}\leq2/(nk)$.
Since $u \in \unblockedShort{t}{v^*}$, it follows that $S_{X_t}(u,c_Y)$ and $S_{Y_t}(u,c_X)$ are both $\{u\}$ and the greedy coupling flips each with probability $P_1/(nk) = 1/(nk)$.
Observe that both $S_{X_t}(v^*,c)$ and $S_{Y_t}(v^*,c)$ are either empty or contain $u$ and $v^*$.
Therefore, each is empty or of size at least $2$, and the greedy coupling flips them with probability at most $P_2/(nk)$.
Consequently, each singleton flip at $u$ has residual probability at least $(1-P_2)/(nk)$ after its pairing with the corresponding large cluster.
The greedy coupling pairs these residual singleton flips together, so with probability at least $(1-P_2)/(nk)$ it flips both $\{u\}$ clusters.
In this event, $X_{t+1}(u)=Y_t(v^*)$, $Y_{t+1}(u)=X_t(v^*)$, and every other vertex remains unchanged.
In particular, the disagreement colors at $u$ are $c_X$ and $c_Y$.
Moreover, $C_{v^*}\subseteq\colorList{w}$, the common color of $w$ remains outside $C_{v^*}$, and every vertex in $N(w)\setminus\{u\}$ retains a color outside $C_{v^*}$.
Thus, all the conditions in the definition of an unblocked neighbor hold, and $w\in\unblockedShort{t+1}{u}$.
It follows that
\begin{equation*}
    \Prob{\EE_w}
    -\Prob{w\in\unblockedShort{t+1}{u}}
    \leq\frac{2}{nk}-\frac{1-P_2}{nk}
    =\frac{1+P_2}{nk},
\end{equation*}
which proves \cref{eqn:neighbor-tradeoff} in the final case and completes the proof.
\end{proof}

\begin{proof}[Proof of \cref{lem:weighted-upper-bound}]
We start by observing:
\begin{align}
\lefteqn{
\lvert \unblockedShort{t+1}{v^*} \cap N_{t,c}(v^*)\rvert - \degree{0}{t,c}{v^*}
} \hspace{.5in}
\nonumber
\\
& = \lvert \unblockedShort{t+1}{v^*} \cap N_{t,c}(v^*)\rvert - \lvert \unblockedShort{t}{v^*} \cap N_{t,c}(v^*)\rvert
\nonumber
\\
& =\lvert (\unblockedShort{t+1}{v^*} \setminus \unblockedShort{t}{v^*})\cap N_{t,c}(v^*) \rvert - \lvert (\unblockedShort{t}{v^*} \setminus \unblockedShort{t+1}{v^*})\cap N_{t,c}(v^*) \rvert,
\label{eqn:set-difference-identity}
\end{align}
where the second equality uses the fact that, for any sets $S$ and $T$, $|S|-|T|=|S\setminus T|-|T\setminus S|$.
We can now bound $\widetilde{W}_c$
as follows, where we use $\eta\geq0$ when applying \cref{lem:newly-unblocked-bound,lem:blocked-unblocked-tradeoff}:
\begin{align*}
\widetilde{W}_c
&= W_c
-\frac{\eta}{\Delta}
\sum_{u\in N_{t,c}(v^*)}
\Exp{\degree{0}{t+1}{u}}
-\frac{\eta}{\Delta}
\left(
\Exp{
\left|
\unblockedShort{t+1}{v^*}\cap N_{t,c}(v^*)
\right|
}
-\degree{0}{t,c}{v^*}
\right) \\
    &  = W_c - \frac{\eta}{\Delta} \sum_{u \in N_{t,c}(v^*)} \Exp{\degree{0}{t+1}{u}}
    -  \frac{\eta}{\Delta}\Exp{\lvert ( \unblockedShort{t+1}{v^*} \setminus \unblockedShort{t}{v^*})\cap N_{t,c}(v^*) \rvert}
    \\
    & \qquad \qquad + \frac{\eta}{\Delta}\Exp{\lvert (\unblockedShort{t}{v^*} \setminus \unblockedShort{t+1}{v^*})\cap N_{t,c}(v^*) \rvert}
        \tag{\mbox{by \cref{eqn:set-difference-identity}}}
    \\
    & \leq W_c - \degree{1}{t,c}{v^*} \frac{\alpha\eta}{\Delta nk}
    - \frac{\eta}{\Delta} \sum_{u \in N_{t,c}(v^*)} \Exp{\degree{0}{t+1}{u}}
    + \frac{\eta}{\Delta}\Exp{\lvert ( \unblockedShort{t}{v^*} \setminus \unblockedShort{t+1}{v^*})\cap N_{t,c}(v^*) \rvert}
    \tag{\mbox{by \cref{lem:newly-unblocked-bound}}}
    \\
    &  \leq W_c +  \degree{0}{t,c}{v^*}\frac{\beta \eta}{\Delta nk} - \degree{1}{t,c}{v^*} \frac{\alpha \eta}{\Delta nk}
    \tag{\mbox{by \cref{lem:blocked-unblocked-tradeoff}}}
    \\
    &\leq Z_c +  \frac{1}{nk}\left(\degree{0}{t,c}{v^*}\frac{\beta \eta}{\Delta} - \degree{1}{t,c}{v^*} \frac{\alpha \eta}{\Delta}\right)
    \tag{\mbox{by \cref{eqn:nondegenerate-w-c-z-c-bound}}}
    \\
    &  = \widetilde{Z}_c.
    \tag{\mbox{by \cref{def:tilde-z-c-upper-bound}}}
\end{align*}
\end{proof}

\subsection{ {Reduction to Six Canonical Configurations}}
\label{subsec:extremal-cases}

Our goal in this subsection is to bound $\widetilde Z_c$ for
nondegenerate colors $c\in\colorList{v^*} \setminus \{c_X, c_Y\}$ with $d_c\geq1$.
We first give a general bound for $d_c\geq3$, and then reduce the
cases $1\leq d_c\leq2$ to six canonical configurations.

Recall from \cref{def:alpha,def:beta,def:tilde-z-c-upper-bound} that,
for these nondegenerate colors,
\begin{equation*}
(nk)\widetilde Z_c=(nk)Z_c
+\frac{\eta}{\Delta}
\left(\beta\degree{0}{t,c}{v^*}-\alpha\degree{1}{t,c}{v^*}\right),
\end{equation*}
where
\begin{equation*}
\alpha=k-\Delta-2,
\qquad
\beta=k-P_2(\Delta_1+\Delta_2)+(1+P_2)\Delta.
\end{equation*}
The bounds in \cref{eqn:auxiliary-parameter-bounds}
make the singly blocked correction nonpositive and bound the
correction per unblocked neighbor in $(nk)\widetilde Z_c$ by $P_2$, respectively.

\begin{lemma}
\label[lemma]{lem:large-color-degree}
Let $c\in\colorList{v^*}\setminus\{c_X,c_Y\}$ with $d_c\geq3$.
Then
\begin{equation*}
    (nk)\widetilde Z_c\leq-1+d_c\left(\frac{4}{3}+P_2+\frac{4}{3}P_4\right).
\end{equation*}
\end{lemma}

\begin{proof}
Let $N_{t,c}(v^*)=\{u_1,\ldots,u_{d_c}\}$.
If $A>0$, then the large cluster contains $v^*$ and all $d_c$ vertices
in $N_{t,c}(v^*)$, so $A\geq4$; the union identity also gives
$a_{\max}\geq1$. Therefore,
$(A-a_{\max}-1)P_A\leq(A-2)P_A\leq2P_4$ by
\cref{eqn:shifted-pj-bound}. If $A=0$, this term is zero.
The same argument gives $(B-b_{\max}-1)P_B\leq2P_4$.

For each $1\leq i\leq d_c$, we claim that
\begin{equation}
    a_iQ_i+b_iQ'_i-\min\{Q_i,Q'_i\}
    +\indicator(u_i\in\unblockedShort{t}{v^*})P_2
    \leq1+P_2.
    \label{eqn:large-degree-term-bound}
\end{equation}
The blocking-level characterizations apply to the adjusted pair
$(a_i,b_i)$ by \cref{lem:blocking-cluster-sizes}.
If $u_i$ is unblocked, then $a_i=b_i=1$, and the left-hand side is
$\max\{Q_i,Q'_i\}+P_2\leq1+P_2$.
Otherwise, the indicator vanishes and $(a_i,b_i)\neq(1,1)$.
If one size is zero, its residual probability is zero, and the
contribution is at most $jP_j\leq1$, where $j$ is the other size,
using \cref{eqn:j-pj-bound} and $P_0=0$.

We may now assume by symmetry that $1\leq a_i\leq b_i$, with $b_i\geq2$.
Since $a_i,b_i\geq1$, the expression
$a_iQ_i+b_iQ'_i-\min\{Q_i,Q'_i\}$ is nondecreasing in each residual
probability. Consequently, $Q_i\leq P_{a_i}$ and $Q'_i\leq P_{b_i}$ give
\begin{equation*}
a_iQ_i+b_iQ'_i-\min\{Q_i,Q'_i\}
\leq a_iP_{a_i}+(b_i-1)P_{b_i}\leq1+P_2
\end{equation*}
by monotonicity of the flip probabilities and
\cref{eqn:j-pj-bound,eqn:line:p2-bound}.
This proves \cref{eqn:large-degree-term-bound} in every case.

Since $\alpha\geq0$ and $\eta\geq0$, the singly blocked correction
is nonpositive, while $\eta\beta/\Delta\leq P_2$ bounds the
correction per unblocked neighbor in $(nk)\widetilde Z_c$ by $P_2$.
Summing \cref{eqn:large-degree-term-bound} and the large-cluster terms gives
\begin{align*}
(nk)\widetilde Z_c
&\leq4P_4+d_c(1+P_2)\\
&\leq-1+d_c\left(\frac43+P_2+\frac43P_4\right),
\end{align*}
where the last inequality uses $d_c\geq3$.
\end{proof}

We now reduce the remaining cases $1\leq d_c\leq2$ for nondegenerate
colors $c\in\colorList{v^*}$ to six configurations.
The reductions below compare values of our upper bound and need not correspond to changes in the graph.

\noindent\begin{minipage}{\linewidth}
\begin{lemma}
\label[lemma]{lem:canonical-blocking-reduction}
Let $c\in\colorList{v^*}\setminus\{c_X,c_Y\}$ with $1\leq d_c\leq2$.
Then $\widetilde Z_c$ is bounded above by its value at one of the
configurations in \cref{tab:six-extremal-configurations} with the same $d_c$.
\begin{table}[H]
\centering
\normalfont
\small
\renewcommand{\arraystretch}{1.15}
\setlength{\tabcolsep}{6pt}
\begin{tabular}{l|cc|ccc}
\hline
& \multicolumn{2}{c|}{$d_c=1$} & \multicolumn{3}{c}{$d_c=2$}\\
Blocking level $j$ & $(a_1,b_1)$ & $(A,B)$
& $(a_1,b_1)$ & $(a_2,b_2)$ & $(A,B)$\\
\hline
0 (unblocked)      & $(1,1)$ & $(2,2)$ & $(1,1)$ & $(1,1)$ & $(3,3)$\\
1 (singly blocked) & $(1,2)$ & $(2,3)$ & $(1,2)$ & $(1,2)$ & $(3,5)$\\
2 (multiblocked)   & $(1,3)$ & $(2,4)$ & $(1,3)$ & $(1,3)$ & $(3,7)$\\
\hline
\end{tabular}
\caption{The six canonical configurations. For $d_c=2$, both neighbors have the indicated blocking level.}
\label{tab:six-extremal-configurations}
\end{table}
Moreover, if all $c$-colored neighbors already have the same blocking level,
the corresponding configuration gives an upper bound.
\end{lemma}
\end{minipage}

\begin{proof}
Write $N_{t,c}(v^*)=\{u_i:1\leq i\leq d_c\}$.
We keep $\alpha$, $\beta$, and $\eta$ fixed throughout the comparisons.
The first three claims remove zero sizes, reduce the positive sizes to
short lists, and then put each pair $(a_i,b_i)$ in the form $(1,b_i)$.
These claims preserve the blocking level of each $u_i$, so the correction
in \cref{def:tilde-z-c-upper-bound} is unchanged.
It therefore suffices to show that these comparisons do not decrease
$(nk)Z_c$. We then compare configurations with different blocking levels.
We use the necessary cluster-size conditions from
\cref{subsec:blocking-level-configurations} throughout.

Recall that
\begin{equation*}
Q_i=P_{a_i}-P_A\indicator(i=i_{\max}),
\qquad
Q'_i=P_{b_i}-P_B\indicator(i=i'_{\max}),
\end{equation*}
where $i_{\max}$ and $i'_{\max}$ maximize $a_i$ and $b_i$,
respectively, and are chosen equal when possible.
We separate the two sides of the coupling by assigning each minimum
term to a side that attains it. Take $\chi_i=\indicator(Q_i\leq Q'_i)$
for $1\leq i\leq d_c$.
Thus, $\chi_i=1$ means that the minimum term for $u_i$ is subtracted
on the $a$-side. Expanding the residual probabilities gives
\begin{align}
(nk)Z_c
&=\sum_{i=1}^{d_c}(a_i-\chi_i)P_{a_i}
 +(A-2a_{i_{\max}}-1+\chi_{i_{\max}})P_A\nonumber\\
&\quad+\sum_{i=1}^{d_c}(b_i-1+\chi_i)P_{b_i}
 +(B-2b_{i'_{\max}}-\chi_{i'_{\max}})P_B.
\label{eqn:canonical-side-decomposition}
\end{align}
We call the first line the $a$-side contribution and the second line
the $b$-side contribution.
After a replacement, the old assignment may no longer subtract the
smaller residual probability. It then subtracts a larger amount, so
keeping $\chi_i$ fixed gives a lower bound on the new value of $(nk)Z_c$.

We will also need to account for changes in the selected maximum
on an unchanged side. Suppose $d_c=2$ and $b_1=b_2$.
Keeping $\chi_1,\chi_2$ fixed, changing the selected index on the
$b$-side from $1$ to $2$ changes its contribution by
\begin{equation}
(\chi_1-\chi_2)P_B.
\label{eqn:tied-maximum-index-change}
\end{equation}
If $Q_1=Q_2$, this change leaves the actual value of $(nk)Z_c$
unchanged: the sum of its linear terms is unchanged because $b_1=b_2$,
and swapping $Q'_1,Q'_2$ leaves the sum of the two minimum terms
unchanged. There is also no change when $P_B=0$.
The same statements hold with the two sides exchanged.

\begin{claim}
\label[claim]{clm:positive-cluster-sizes}
The cluster sizes may be taken to satisfy
$A,B>0$ and $a_i,b_i>0$ for $1\leq i\leq d_c$.
\end{claim}

\begin{proof}[Proof of the claim]
First suppose that $d_c=1$.
An unblocked neighbor already has $(a_1,b_1)=(1,1)$.
For a singly blocked neighbor, the only zero case, up to symmetry,
is $(1,0)$.
Replacing it by $(1,2)$ changes $(nk)Z_c$ from $1-P_2$ to
$1-P_3$, which is no smaller.
For a multiblocked neighbor, replace $(1,0)$ by $(1,3)$;
the new value is $1-P_2+2(P_3-P_4)\geq1-P_2$.
The remaining zero cases are $(0,j)$ with $j=0$ or $j\geq2$.
Their values are respectively zero and
$j(P_j-P_{j+1})\leq2(P_2-P_3)$ by
\cref{eqn:weighted-gap-bound}.
Replacing either by $(2,2)$ gives $3(P_2-P_3)$.
Each replacement satisfies the necessary size conditions for the
prescribed blocking level.

Now suppose $d_c=2$.
By symmetry, consider a zero on the $a$-side and relabel the neighbors
so that $a_2=0$.
Since $Q_2=0$, we may take $\chi_2=1$.
If $a_1=0$, then $A=0$.
The original $a$-side contribution is at most
$(a_1-\chi_1)P_{a_1}$. This is its value when $A=0$; when $A>0$,
we have $A=a_1+1$ and its value is
$(a_1-\chi_1)(P_{a_1}-P_{a_1+1})$.
We replace $a_2$ by $2$ or $3$, keep $a_1$ fixed, and set
$A=a_1+a_2+1$ after the replacement.
Substitution in the $a$-side contribution in
\cref{eqn:canonical-side-decomposition} gives the following lower bounds
on its increase.
\begin{table}[H]
\centering
\normalfont
\small
\renewcommand{\arraystretch}{1.15}
\setlength{\tabcolsep}{6pt}
\begin{tabular}{l|c|c}
\hline
$a_1$ & Replace $a_2$ by $2$ & Replace $a_2$ by $3$\\
\hline
$0$ & $P_2-P_3$ & $2(P_3-P_4)$\\
$1$ & $P_2$ & $2P_3-P_5$\\
$2$ & $P_2$ & $2P_3$\\
$3$ & $P_2-P_6$ & $2P_3$\\
$a_1\geq4$ & $P_2$ & $2P_3$\\
\hline
\end{tabular}
\caption{Lower bounds on the increase in the $a$-side contribution when $a_2=0$ is replaced by $2$ or $3$, with $a_1$ fixed and $A=a_1+a_2+1$ after the replacement. The bounds allow either selected maximum when the replacement creates a tie. Changes on the $b$-side are accounted for separately.}
\label{tab:zero-size-replacements}
\end{table}

These bounds are all nonnegative by monotonicity and $P_j=0$ for $j\geq7$.

If choosing the same maximum index on both sides moves
$i'_{\max}$ from $1$ to $2$, the unchanged $b$-side is tied, say $b_1=b_2$.
By \cref{eqn:tied-maximum-index-change}, the loss in the fixed-$\chi$
expression for that side is $(1-\chi_1)P_B$, because $\chi_2=1$;
there is no loss when $B=0$.
If $a_1=0$, then $\chi_1=1$ and there is no decrease.
If $a_1\geq3$, or if $a_1=2$ and $a_2$ is replaced by $2$,
we can retain the old common maximum whenever $b_1=b_2$, so again
there is no decrease.

When $a_1=1$, the decrease is at most $P_3$, since $B=2b_1+1\geq3$
when $B>0$.
Replacing $a_2$ by $2$ increases the $a$-side contribution by at least
$P_2\geq P_3$.
Replacing it by $3$ increases that contribution by at least
$2P_3-P_5\geq P_3$.

When $a_1=2$ and $a_2$ is replaced by $3$, a decrease requires
$\chi_1=0$.
This rules out $b_1=1$, because then
$Q_1\leq P_2<1-P_3=Q'_1$.
Thus, if $B>0$, we have $b_1\geq2$ and the decrease is at most
$P_5\leq2P_3$, the lower bound on the increase on the $a$-side.
Therefore, neither replacement decreases $(nk)Z_c$.
The argument for a zero on the $b$-side is symmetric.

For a singly blocked neighbor, replace a zero by $2$.
For a multiblocked neighbor whose other size is $1$, replace a zero
by $3$; otherwise replace it by $2$.
An unblocked neighbor has no zero size.
These choices preserve the necessary size conditions for the
original blocking levels.
Applying the replacements successively removes every zero small size.
On each side where a size was replaced, we set the large size to one
plus the sum of the small sizes.
On an unchanged side, the small sizes are already positive, and
$c\in\colorList{v^*}$ gives the same union identity.
Thus, $A,B>0$ as well.
\end{proof}

With zero sizes removed, we now reduce the positive sizes to short lists.

\noindent\begin{minipage}{\linewidth}
\begin{claim}
\label[claim]{clm:reduced-cluster-sizes}
For $1 \leq i \leq d_c$,
the pair $(a_i,b_i)$ may be taken from the lists in
\cref{tab:retained-cluster-pairs} according to the blocking level of $u_i$.
\begin{table}[H]
\centering
\normalfont
\small
\renewcommand{\arraystretch}{1.15}
\setlength{\tabcolsep}{6pt}
\begin{tabular}{l|c}
\hline
Blocking level $j$ & Configuration sizes $(a_i,b_i)$\\
\hline
0 (unblocked) & $(1,1)$\\
1 (singly blocked) & $(1,2),(2,1)$\\
2 (multiblocked) & $(1,3),(3,1),(2,2)$\\
\hline
\end{tabular}
\caption{Configuration sizes $(a_i,b_i)$ after reducing the positive cluster sizes, listed by the original blocking level. Both orientations are retained at this stage.}
\label{tab:retained-cluster-pairs}
\end{table}
In addition, $A=1+\sum_{i=1}^{d_c}a_i$ and $B=1+\sum_{i=1}^{d_c}b_i$.
\end{claim}
\end{minipage}

\begin{proof}[Proof of the claim]
By \cref{clm:positive-cluster-sizes}, all sizes may be taken positive.

Suppose $d_c=1$.
An unblocked pair is already $(1,1)$.
If one entry is $1$ and the other is at least two, assume by symmetry that $a_1=1$ and $b_1\geq2$.
Substitution gives
\begin{equation*}
(nk)Z_c=1-P_2+(b_1-1)(P_{b_1}-P_{b_1+1}).
\end{equation*}
By \cref{eqn:line:weighted-gap-monotonicity}, this expression does not decrease when $b_1$ is replaced by $2$, or by $3$ if $b_1\geq3$ is required.
Thus a singly blocked pair reduces to $(1,2)$ up to symmetry, and a multiblocked pair containing a $1$ reduces to $(1,3)$ up to symmetry.
If both entries are at least two, exchange the sides if necessary so that the minimum term is assigned to the $a$-side in \cref{eqn:canonical-side-decomposition}.
Then
\begin{align*}
(nk)Z_c
&=(a_1-1)(P_{a_1}-P_{a_1+1})+b_1(P_{b_1}-P_{b_1+1})\\
&\leq(P_2-P_3)+2(P_2-P_3)
=3(P_2-P_3),
\end{align*}
using \cref{eqn:line:weighted-gap-monotonicity,eqn:weighted-gap-bound}.
This is the value at $(2,2)$.

Now suppose $d_c=2$.
We reduce the sizes one at a time, stopping when the pairs have the forms in \cref{tab:retained-cluster-pairs}.
By symmetry, it suffices to show that replacing $a_1\geq3$ by $a_1-1$ does not decrease $(nk)Z_c$.
We hold $a_2,b_1,b_2$ fixed and decrease $A=a_1+a_2+1$ by one.
We first choose the maximum indices so that the $b$-side contribution stays fixed.
We then bound the change in the $a$-side contribution.

The only case in which decreasing $a_1$ can force the selected maximum on the $b$-side to change is when $a_1=a_2$ and $b_1=b_2$.
After the decrease, $a_2$ is the unique maximum on the $a$-side, so the coupling selects index $2$ on both sides.
We may select index $2$ on both sides before the decrease as well.
Indeed, before the decrease $A=2a_1+1\geq7$, so $P_A=0$ and $Q_1=Q_2=P_{a_1}$.
Changing the selected $a$-side maximum therefore does not change either residual probability, and changing the selected $b$-side maximum does not change $(nk)Z_c$ by the observation preceding \cref{clm:positive-cluster-sizes}.
If the $b$-side has a unique maximum, its selected index stays fixed even if the $a$-side maximum changes.
If the decrease creates a tie on the $a$-side, we choose its maximum index to agree with the fixed $b$-side index.
Thus, in every case, the selected $b$-side maximum can stay fixed during the decrease.

With these indices chosen, fix $\chi_1,\chi_2$ as in \cref{eqn:canonical-side-decomposition}.
Expanding the $a$-side contribution gives
\begin{equation*}
(a_1-\chi_1)P_{a_1}+(a_2-\chi_2)P_{a_2}
+(a_1+a_2-2a_{i_{\max}}+\chi_{i_{\max}})P_{a_1+a_2+1}.
\end{equation*}
Let $D$ be the value of this expression after the decrease minus its value before the decrease, using the selected $a$-side maximum in each case.
The $b$-side contribution is unchanged.
Keeping $\chi_1,\chi_2$ fixed gives a lower bound on the new value of $(nk)Z_c$, so it suffices to show that $D\geq0$.

If $a_1+a_2\geq7$, the large-cluster terms vanish both before and after the decrease, and
\begin{equation*}
D=(a_1-1-\chi_1)P_{a_1-1}-(a_1-\chi_1)P_{a_1}\geq0.
\end{equation*}
For $a_1\geq4$, this follows from \cref{eqn:line:geometric-decay}; for $a_1=3$, it follows from $P_2\geq2P_3$.

For $a_1+a_2\leq6$, substitution and \cref{eqn:line:weighted-gap-monotonicity} give the bounds in \cref{tab:size-decrement-bounds}.
For example, when $(a_1,a_2)=(3,3)$, the decrease changes the pair to $(2,3)$, so the new maximum index is $2$.
The large size decreases from $7$ to $6$, giving
\begin{equation*}
D=(2-\chi_1)P_2-(3-\chi_1)P_3+(\chi_2-1)P_6
\geq P_2-2P_3-P_6.
\end{equation*}
The inequality uses $P_2\geq P_3$.

\begin{table}[H]
\centering
\normalfont
\small
\renewcommand{\arraystretch}{1.15}
\setlength{\tabcolsep}{6pt}
\begin{tabular}{l|c}
\hline
$(a_1,a_2)$ before the decrease & $\text{Lower bound on }D$\\
\hline
$(3,1)$ & $P_2-2P_3+P_5$\\
$(3,2)$ & $P_2-2P_3$\\
$(3,3)$ & $P_2-2P_3-P_6$\\
$(4,1)$ & $2P_3-3P_4-P_5+2P_6$\\
$(4,2)$ & $2P_3-3P_4$\\
$(5,1)$ & $3P_4-4P_5-2P_6$\\
\hline
\end{tabular}
\caption{Lower bounds on the change $D$ in the $a$-side contribution when $a_1$ decreases by one, for $a_1\geq3$, $a_2\geq1$, and $a_1+a_2\leq6$. Each bound allows every choice of $\chi_1,\chi_2$ and either selected maximum when a tie is created.}
\label{tab:size-decrement-bounds}
\end{table}

The first, second, third, and fifth bounds are nonnegative by
\cref{eqn:line:p2-bound}, \cref{eqn:line:p3-p6-bound}, and
\cref{eqn:line:geometric-decay}.
For the other two bounds, put
$g_j=(j-1)(P_j-P_{j+1})$ for $j\geq2$.
By \cref{eqn:line:weighted-gap-monotonicity}, the sequence $g_j$
is nonnegative and nonincreasing.
Consequently,
\begin{align*}
2P_3-3P_4-P_5+2P_6
 &=g_3-\frac{g_4}{3}-\frac{g_5}{2}
 \geq\frac{g_3}{6}\geq0,\\
3P_4-4P_5-2P_6
 &=g_4-\frac{g_5}{4}-\frac{3g_6}{5}
 \geq\frac{3g_4}{20}\geq0.
\end{align*}
Thus, $D\geq0$.

We apply these decreases to $a_1,a_2$ and, by symmetry, to $b_1,b_2$.
For a singly blocked pair, reduce the entry larger than $1$ to $2$.
For a multiblocked pair containing a $1$, reduce the other entry to $3$; for a multiblocked pair with both entries at least two, reduce both to $2$.
Unblocked pairs stay at $(1,1)$.
This gives exactly the pairs in \cref{tab:retained-cluster-pairs}.
\end{proof}

Only the listed pairs remain. We next align their orientations and
remove the remaining pairs of the form $(2,2)$.

\begin{claim}
\label[claim]{clm:fixed-profile-canonical}
The cluster sizes may be taken to satisfy
\begin{equation*}
(a_i,b_i)=
\begin{cases}
(1,1), & \text{if $u_i$ is unblocked},\\
(1,2), & \text{if $u_i$ is singly blocked},\\
(1,3), & \text{if $u_i$ is multiblocked},
\end{cases}
\end{equation*}
for $1\leq i\leq d_c$, with $A=1+\sum_{i=1}^{d_c}a_i$ and
$B=1+\sum_{i=1}^{d_c}b_i$.
\end{claim}

\begin{proof}[Proof of the claim]
By \cref{clm:reduced-cluster-sizes}, only the pairs in
\cref{tab:retained-cluster-pairs} remain.
When $d_c=1$, exchanging the sides gives the required orientation.
Replacing the remaining multiblocked pair $(2,2)$ by $(1,3)$
increases $(nk)Z_c$ by $1-4P_2+5P_3-2P_4\geq0$ by
\cref{eqn:line:multiblocked-normal-form}.

Now suppose $d_c=2$.
If at most one neighbor is blocked, exchange the two sides to orient
its pair. The only further replacement is $(2,2)$ by $(1,3)$ next to
an unblocked pair. The increase is
$1-3P_2+2P_3+P_4-P_5\geq0$, using $3P_2\leq1$ and $P_5\leq P_4$.

For two singly blocked neighbors with opposite orientations,
aligning the pairs increases $(nk)Z_c$ by $P_5$.
The same increase holds when a singly blocked pair and a multiblocked
pair $(1,3)$ or $(3,1)$ have opposite orientations.
If the multiblocked pair is instead $(2,2)$, exchange the sides if
necessary so that the singly blocked pair is $(1,2)$.
Replacing $(2,2)$ by $(1,3)$ then increases the value by
$1-3P_2+2P_3\geq0$.

For two multiblocked neighbors, aligning $(1,3)$ and $(3,1)$
increases the value by $2P_5+P_7$.
If exactly one pair is $(2,2)$, exchange the sides if necessary so
that the other pair is $(1,3)$.
Replacing $(2,2)$ by $(1,3)$ increases the value by
\begin{equation*}
1-3P_2+2P_3+P_7.
\end{equation*}
If both pairs are $(2,2)$, replacing them by two copies of $(1,3)$
increases the value by
\begin{equation*}
2-6P_2+4P_3+P_7-P_5.
\end{equation*}
Both increases are nonnegative since $3P_2\leq1$ and $P_5\leq4P_3$.
\end{proof}

If all neighbors have the same blocking level in the original
instance, the first three claims give the configuration in \cref{tab:six-extremal-configurations}
with that blocking level and the same $d_c$. Since none of
these comparisons decreases $\widetilde Z_c$, this proves the
conditional assertion of the lemma.
It remains to compare configurations whose neighbors have different
blocking levels.

\begin{claim}
\label[claim]{clm:uniform-blocking-profile}
After the reductions above, it suffices to consider configurations
in which the neighbors in $N_{t,c}(v^*)$ are all unblocked, all singly
blocked, or all multiblocked.
\end{claim}

\begin{proof}[Proof of the claim]
The result is immediate when $d_c=1$.
Suppose that $d_c=2$, and use the configuration sizes $(a_i,b_i)$ from
\cref{clm:fixed-profile-canonical}.
Compare a configuration with different blocking levels with the two
configurations in which both neighbors have the blocking level of $u_1$
or both have the blocking level of $u_2$.
The correction in \cref{def:tilde-z-c-upper-bound} is linear in the
numbers of neighbors of each blocking level, so its value for the mixed
configuration is the average of its values for these two configurations.
It remains to make the same comparison for $(nk)Z_c$.
Direct substitution in \cref{def:z-c-upper-bound} gives the following
three differences.
For the pairs of blocking levels $(0,1)$, $(0,2)$, and $(1,2)$, respectively,
the average of the two corresponding uniform values of $(nk)Z_c$
exceeds the mixed value by
\begin{equation*}
\frac{P_3+P_5}{2},\qquad
\frac{P_3+2P_5+P_7}{2},\qquad
\frac{P_5+P_7}{2}.
\end{equation*}
All three differences are nonnegative.
The same comparison therefore holds for $\widetilde Z_c$.
Each mixed value is at most an average of two uniform values,
and hence at most their maximum.
\end{proof}

\end{proof}

\subsection{ {Bounding the Canonical Contributions}}
\label{subsec:weighted-color-bounds}

The six configurations in \cref{lem:canonical-blocking-reduction}
suffice to bound $\widetilde Z_c$, and hence $\widetilde W_c$, for
every nondegenerate $c\in\colorList{v^*}$ with $1\leq d_c\leq2$.
We evaluate them in three lemmas, one for each blocking level,
combining the cases $d_c=1$ and $d_c=2$ into a common bound per neighbor.
We also bound the contributions of colors outside $\colorList{v^*}$.

\begin{table}[H]
\centering
\normalfont
\small
\renewcommand{\arraystretch}{1.15}
\setlength{\tabcolsep}{6pt}
\begin{tabular}{l|c|c}
\hline
Blocking level $j$ & $(nk)Z_c$, $d_c=1$ & $(nk)Z_c$, $d_c=2$\\
\hline
0 (unblocked) & $1-P_2$ & $2+P_3$\\
1 (singly blocked) & $1-P_3$ & $2+2P_2+P_5$\\
2 (multiblocked) & $1-P_2+2(P_3-P_4)$ & $2+4P_3+P_7$\\
\hline
\end{tabular}
\caption{Values of $(nk)Z_c$ for the six canonical configurations in
\cref{tab:six-extremal-configurations}, grouped by blocking level and
color degree. The correction $(nk)(\widetilde Z_c-Z_c)$ is
$d_c\eta\beta/\Delta$ for unblocked neighbors,
$-d_c\eta\alpha/\Delta$ for singly blocked neighbors, and zero for
multiblocked neighbors.}
\label{tab:low-degree-summary}
\end{table}

Substitution in \cref{def:z-c-upper-bound} gives the six values in
\cref{tab:low-degree-summary}. Adding
\cref{eqn:line:p2-p4-bound,eqn:line:p3-p4-p5-bound} gives
\begin{equation}
2P_2+2P_3+P_5\leq1,
\label{eqn:two-neighbor-coefficient-bound}
\end{equation}
which we use in the first two comparisons.

We first handle the case when all neighbors are unblocked.

\begin{lemma}
\label[lemma]{lem:all-unblocked-bound}
Let $c\in\colorList{v^*}$ satisfy $1\leq d_c\leq2$ and $\degree{0}{t,c}{v^*}=d_c$.
Then
\begin{equation*}
    (nk)\widetilde Z_c\leq-1+d_c\left(2-P_2+\frac{\eta\beta}{\Delta}\right).
\end{equation*}
\end{lemma}

\begin{proof}
Since every vertex in $N_{t,c}(v^*)$ is unblocked, $c\notin\{c_X,c_Y\}$.
By \cref{lem:blocking-cluster-sizes}, $a_i=b_i=1$ for $1\leq i\leq d_c$,
while $c\in\colorList{v^*}$ gives $A=B=d_c+1$.
When $d_c=1$, the unblocked row of \cref{tab:low-degree-summary}
and the correction give
\begin{equation*}
(nk)\widetilde Z_c=1-P_2+\frac{\eta\beta}{\Delta}
=-1+\left(2-P_2+\frac{\eta\beta}{\Delta}\right).
\end{equation*}
When $d_c=2$, the same row gives the bound because
\begin{equation*}
\frac32+\frac{P_3}{2}\leq2-P_2.
\end{equation*}
This follows from \cref{eqn:two-neighbor-coefficient-bound}.
\end{proof}

The next lemma handles the case in which all neighbors of color $c$ are singly blocked.
Without the correction, the one-neighbor case gives the coefficient
$2-P_3$, which is too large for our desired bound.
The possibility of becoming unblocked supplies the correction
$-\eta\alpha/\Delta$ per neighbor.

\begin{lemma}
\label[lemma]{lem:all-singly-blocked-bound}
Let $c\in\colorList{v^*}$ satisfy $1\leq d_c\leq2$ and $\degree{1}{t,c}{v^*}=d_c$.
Then
\begin{equation*}
    (nk)\widetilde Z_c\leq-1+d_c\left(2-P_3-\frac{\eta\alpha}{\Delta}\right).
\end{equation*}
\end{lemma}

\begin{proof}
Since every vertex in $N_{t,c}(v^*)$ is singly blocked,
$c\notin\{c_X,c_Y\}$.
By \cref{lem:canonical-blocking-reduction}, it suffices to use
$(a_i,b_i)=(1,2)$ for $1\leq i\leq d_c$.
When $d_c=1$, the singly blocked row of
\cref{tab:low-degree-summary} and the correction give
\begin{equation*}
(nk)\widetilde Z_c\leq1-P_3-\frac{\eta\alpha}{\Delta}
=-1+\left(2-P_3-\frac{\eta\alpha}{\Delta}\right).
\end{equation*}
When $d_c=2$, the required comparison is
\begin{equation*}
\frac32+P_2+\frac{P_5}{2}\leq2-P_3.
\end{equation*}
This is \cref{eqn:two-neighbor-coefficient-bound}.
\end{proof}

Finally, the following lemma handles the case in which all neighbors of a nondegenerate color
$c\in\colorList{v^*}$ are multiblocked.
The correction term in $\widetilde Z_c$ vanishes for these blocking profiles, so it suffices to bound $Z_c$.
Neighbors of a degenerate color can also be multiblocked; their combined estimate is stated separately below.

\begin{lemma}
\label[lemma]{lem:all-multiblocked-bound}
Let $c\in\colorList{v^*}\setminus\{c_X,c_Y\}$ satisfy $1\leq d_c\leq2$ and $\degree{2}{t,c}{v^*}=d_c$.
Then
\begin{equation*}
    (nk)\widetilde Z_c\leq-1+d_c(2-P_2+2P_3-2P_4).
\end{equation*}
\end{lemma}

\begin{proof}
The correction in \cref{def:tilde-z-c-upper-bound} vanishes,
so $\widetilde Z_c=Z_c$.
By \cref{lem:canonical-blocking-reduction}, it suffices to use
$(a_i,b_i)=(1,3)$ for $1\leq i\leq d_c$.
When $d_c=1$, the multiblocked row of \cref{tab:low-degree-summary} gives
\begin{equation*}
(nk)\widetilde Z_c=(nk)Z_c\leq1-P_2+2(P_3-P_4)
=-1+(2-P_2+2P_3-2P_4).
\end{equation*}
When $d_c=2$, the same row and $P_7=0$ reduce the claim to
\begin{equation*}
\frac32+2P_3\leq2-P_2+2P_3-2P_4,
\end{equation*}
which is \cref{eqn:line:p2-p4-bound}.
\end{proof}

Combining the three bounds above with
\cref{lem:canonical-blocking-reduction}, we obtain, for every
$c\in\colorList{v^*}\setminus\{c_X,c_Y\}$ with $1\leq d_c\leq2$,
\begin{equation}
(nk)\widetilde W_c\leq(nk)\widetilde Z_c\leq-1+\lambda d_c,
\label{eqn:extremal-color-bound}
\end{equation}
where
\begin{equation}
\lambda:=\max\left\{
2-P_2+\frac{\eta\beta}{\Delta},\;
2-P_3-\frac{\eta\alpha}{\Delta},\;
2-P_2+2P_3-2P_4
\right\}.
\label{def:lambda}
\end{equation}
Here the first inequality in \cref{eqn:extremal-color-bound} is
\cref{lem:weighted-upper-bound}.

We handle colors outside
$\colorList{v^*}$. For these colors $A=B=0$ is forced, so the
positive-size reduction in \cref{subsec:extremal-cases} does not apply.
The following direct estimate works for every $d_c\geq1$ and uses
\cref{eqn:auxiliary-parameter-bounds} to bound the correction.

\begin{lemma}
\label[lemma]{lem:unavailable-color-bound}
Let $c\notin\colorList{v^*}$ with $d_c\geq1$.
Then
\begin{equation*}
(nk)\widetilde Z_c\leq d_c(1+P_2).
\end{equation*}
\end{lemma}

\begin{proof}
Since $c\notin\colorList{v^*}$, we have $A=B=0$ and
$c\notin\{c_X,c_Y\}$.
Consequently, $Q_i=P_{a_i}$ and $Q'_i=P_{b_i}$.
We bound each neighbor's contribution to $(nk)Z_c$, together with
its correction.
The blocking-level characterizations apply to adjusted sizes by
\cref{lem:blocking-cluster-sizes}.

If $u_i$ is unblocked, then $a_i=b_i=1$.
Its contribution to $(nk)Z_c$ is $1$, and its correction is at most
$P_2$ by \cref{eqn:auxiliary-parameter-bounds}.

For a blocked neighbor, the correction is nonpositive:
it is $-\eta\alpha/\Delta$ for a singly blocked neighbor and zero
for a multiblocked neighbor. Moreover, $(a_i,b_i)\neq(1,1)$.
If one size is zero, the contribution to $(nk)Z_c$ is at most
$jP_j\leq1$, where $j$ is the other size, using
\cref{eqn:j-pj-bound} and $P_0=0$.
Otherwise, assume by symmetry that $1\leq a_i\leq b_i$, so $b_i\geq2$.
Its contribution is
\begin{equation*}
a_iP_{a_i}+(b_i-1)P_{b_i}\leq1+P_2
\end{equation*}
by \cref{eqn:j-pj-bound,eqn:line:p2-bound}.
Summing over all $d_c$ neighbors proves the lemma.
\end{proof}

We record the corresponding estimate for the two degenerate colors before combining the per-color bounds.

\begin{lemma}
\label[lemma]{lem:degenerate-color-sum-bound}
Fix $(X_t,Y_t)\in\Omega^2_{v^*}$, and let $c\in\{c_X,c_Y\}$ with $d_c\geq1$.
Then
\begin{equation}
    (nk)\widetilde W_c
    \leq-1+(2-P_2)d_c.
    \label{eqn:degenerate-color-sum-bound}
\end{equation}
\end{lemma}

The proof is given in \cref{sec:degenerate-colors}, where we construct the additional coupling required for these two colors.

\subsection{Contraction and Mixing}
\label{subsec:flip-mixing-proof}

We first establish the following contraction bound for list-coloring instances.
Recall from \cref{subsec:list-colorings}
that $\Sigma\subseteq\Omega^2$ is the set of pairs of list labelings
that differ at exactly one vertex.

\begin{theorem}
\label[theorem]{thrm:list-local-contraction}
Let $\Delta\geq\deltabound$ and $k\geq\lambbound\Delta$.
Let $G$ be a graph on $n$ vertices with maximum degree at most $\Delta$, and let $(\colorList{v})_{v\in V}$ be a list-coloring instance satisfying
$k\geq\lvert\colorList{v}\rvert\geq d(v)+k-\Delta$ for every $v\in V$.
For the flip dynamics with the parameters in \cref{def:flip-probabilities}, the greedy coupling satisfies, for every pair of list labelings $(X_t,Y_t)\in\Sigma$ with unique disagreement vertex $v^*$,
\begin{equation*}
    (nk)\widetilde{W}\leq-10^{-6}\bigl(d(v^*)+k-\Delta\bigr).
\end{equation*}
\end{theorem}

Note that the condition $\lvert\colorList{v}\rvert\geq d(v)+k-\Delta$
reduces to $\lvert\colorList{v}\rvert\geq k$ if $G$ is $\Delta$-regular.
Likewise, it is satisfied in the $k$-coloring case, where
$\colorList{v}=[k]$ for every $v$.
The more general condition accommodates irregular graphs and the residual
lists obtained by fixing the colors of some vertices.

\begin{proof}[Proof of \cref{thrm:list-local-contraction}]
Fix $v^*\in V$ and $(X_t,Y_t)\in\Omega^2_{v^*}$.
We work with the metric $\nmet$ and the parameters in
\cref{def:flip-probabilities}.

We now combine the per-color estimates.
The $-1$ terms sum to $-\lvert\colorList{v^*}\rvert$, so it remains
to show that they outweigh the expected contributions from the neighbors.
Recall that $\Delta_i$ counts
neighbors of $v^*$ whose color lies in $\colorList{v^*}$ and appears
exactly $i$ times in its neighborhood. Set
\begin{equation*}
    \Delta_{\geq3}:=\sum_{i\geq3}\Delta_i,
    \qquad
    \Delta':=\left\lvert
    \{u\in N(v^*):X_t(u)\notin\colorList{v^*}\}
    \right\rvert.
\end{equation*}
These counts partition the neighborhood, so
$\Delta_1+\Delta_2+\Delta_{\geq3}+\Delta'=d(v^*)\leq\Delta$.
Recall from \cref{def:lambda} that
\begin{equation*}
\lambda=\max\left\{
2-P_2+\frac{\eta\beta}{\Delta},\;
2-P_3-\frac{\eta\alpha}{\Delta},\;
2-P_2+2P_3-2P_4
\right\}.
\end{equation*}
The bounds in \cref{eqn:extremal-color-bound,lem:large-color-degree}
give, for every nondegenerate $c\in\colorList{v^*}$ with $d_c\geq1$,
\begin{equation*}
    (nk)\widetilde Z_c\leq
    \begin{cases}
        -1+\lambda d_c, & 1\leq d_c\leq2,\\
        -1+\left(\dfrac43+P_2+\dfrac43P_4\right)d_c, & d_c\geq3.
    \end{cases}
\end{equation*}
The two degenerate colors $c_X,c_Y$ also belong to $\colorList{v^*}$,
so their neighbors are included in $\Delta_1$, $\Delta_2$, or
$\Delta_{\geq3}$ according to their color degree. Since
\begin{equation*}
    2-P_2\leq\lambda,
    \qquad
    2-P_2\leq\frac43+P_2+\frac43P_4,
\end{equation*}
the per-color estimate in \cref{lem:degenerate-color-sum-bound} gives the same bounds for degenerate colors according to their color degree.

Using \cref{lem:per-color-decomposition,eqn:zero-neighbor-change},
then \cref{lem:weighted-upper-bound,lem:unavailable-color-bound} for nondegenerate colors and \cref{lem:degenerate-color-sum-bound} for degenerate colors,
and finally the bounds for colors in $\colorList{v^*}$, we obtain
\begin{align*}
    (nk)\widetilde W
    &\leq-\lvert\mathcal A_t(v^*)\rvert
    +(nk)\sum_{\substack{c\in[k]\\d_c\geq1}}\widetilde W_c\\
    &\leq-\lvert\mathcal A_t(v^*)\rvert
    +(nk)\sum_{\substack{c\in\colorList{v^*}\setminus\{c_X,c_Y\}\\d_c\geq1}}\widetilde Z_c\\
    &\qquad
    +\sum_{\substack{c\in\{c_X,c_Y\}\\d_c\geq1}}
    \left(-1+(2-P_2)d_c\right)
    +(1+P_2)\Delta'\\
    &\leq-\lvert\colorList{v^*}\rvert
    +\lambda(\Delta_1+\Delta_2)
    +\left(\frac43+P_2+\frac43P_4\right)\Delta_{\geq3}
    +(1+P_2)\Delta'\\
    &\leq-(d(v^*)+k-\Delta)
    +\lambda(\Delta_1+\Delta_2)
    +\left(\frac43+P_2+\frac43P_4\right)\Delta_{\geq3}
    +(1+P_2)\Delta'.
\end{align*}
The step producing the term $-\lvert\colorList{v^*}\rvert$ uses
$\lvert\mathcal A_t(v^*)\rvert+
\lvert\{c\in\colorList{v^*}:d_c\geq1\}\rvert
=\lvert\colorList{v^*}\rvert$;
the last line uses the assumed lower bound on the list size.

It remains to prove
\begin{equation}
\label{eqn:coefficient-bound}
    \lambda(\Delta_1+\Delta_2)
    +\left(\frac43+P_2+\frac43P_4\right)\Delta_{\geq3}
    +(1+P_2)\Delta'
    \leq(1-10^{-6})(d(v^*)+k-\Delta).
\end{equation}
We have $1+P_2=1.324$ and
$\frac43+P_2+\frac43P_4=1331/750<1.775$.
The bound for unblocked neighbors decreases as $\Delta_1+\Delta_2$
increases. We keep this dependence in a common bound for the three terms
defining $\lambda$, since $\lambda$ is multiplied by $\Delta_1+\Delta_2$.
Since $\Delta_1+\Delta_2\leq\Delta$, the quantity
$1.324\Delta-0.324(\Delta_1+\Delta_2)$ is positive, so
\begin{align*}
    2-P_2+\frac{\eta\beta}{\Delta}
    &=1.761+\frac{0.085}{k}
    \bigl(1.324\Delta-0.324(\Delta_1+\Delta_2)\bigr)\\
    &\leq1.761+\frac{0.085}{\lambbound}
    \left(1.324-0.324\frac{\Delta_1+\Delta_2}{\Delta}\right)\\
    &\leq1.824098-0.0152\frac{\Delta_1+\Delta_2}{\Delta}.
\end{align*}
The last line rounds the constant term upward and the subtracted
coefficient downward. The other two terms satisfy
\begin{align*}
    2-P_3-\frac{\eta\alpha}{\Delta}
    &=1.761+\frac{0.085\Delta+0.17}{k}\\
    &\leq1.761+\frac{0.085}{\lambbound}
    +\frac{0.17}{\deltabound\cdot\lambbound}<1.808898,\\
    2-P_2+2P_3-2P_4&=1.808.
\end{align*}
For $0\leq\Delta_1+\Delta_2\leq\Delta$, the affine bound above is at
least $1.824098-0.0152=1.808898$. Hence
\begin{equation}
    \lambda\leq1.824098-0.0152\frac{\Delta_1+\Delta_2}{\Delta}.
    \label{eqn:lambda-bound}
\end{equation}
The left-hand side of \cref{eqn:coefficient-bound} is therefore at most
\begin{align*}
    &1.775d(v^*)+0.049098(\Delta_1+\Delta_2)
    -\frac{0.0152}{\Delta}(\Delta_1+\Delta_2)^2\\
    &\hspace{2em}\leq
    1.824098d(v^*)-\frac{0.0152}{\Delta}d(v^*)^2\\
    &\hspace{2em}\leq1.793698d(v^*)+0.0152\Delta.
\end{align*}
The first expression increases with $\Delta_1+\Delta_2$ on
$[0,d(v^*)]$, since its derivative is at least
$0.049098-2(0.0152)=0.018698>0$.
The last inequality uses $d(v^*)^2/\Delta\geq2d(v^*)-\Delta$.
Since $1.793698>1-10^{-6}$, the difference between this last bound and
the right-hand side of \cref{eqn:coefficient-bound} increases with
$d(v^*)\leq\Delta$ and decreases with $k\geq\lambbound\Delta$.
Its maximum is consequently
\begin{equation*}
    \bigl(1.808898-(1-10^{-6})\lambbound\bigr)\Delta<0.
\end{equation*}
This proves \cref{eqn:coefficient-bound} and hence
$(nk)\widetilde W\leq-10^{-6}(d(v^*)+k-\Delta)$, as desired.
\end{proof}

We now prove \cref{thrm:flip-mixing} as a direct corollary by applying
the Path Coupling Theorem (\cref{thrm:path-coupling}).

\begin{proof}[Proof of \cref{thrm:flip-mixing}]
Set $\colorList{v}=[k]$ for every $v\in V$. Then
$\lvert\colorList{v}\rvert=k\geq d(v)+k-\Delta$, so
\cref{thrm:list-local-contraction} applies. Let
\begin{equation*}
    \gamma:=10^{-6}\left(1-\frac1{\lambbound}\right)>0.
\end{equation*}
Since $d(v^*)+k-\Delta\geq k-\Delta\geq(1-1/\lambbound)k$,
the local contraction bound gives $\widetilde W\leq-\gamma/n$
for every neighboring pair. Moreover, $\nmet(X_t,Y_t)\leq1$, so
\begin{align*}
    \Exp{\nmet(X_{t+1},Y_{t+1})}
    &\leq\nmet(X_t,Y_t)-\frac{\gamma}{n}\\
    &\leq\left(1-\frac{\gamma}{n}\right)\nmet(X_t,Y_t).
\end{align*}
The metric has edge weights in $[1-\eta,1]\subseteq[1/2,1]$,
and any two labelings are joined by at most $n$ single-vertex changes.
Applying \cref{thrm:path-coupling} with $\delta=\gamma/n$ therefore gives
$\Tmix(\eps)=O(n\log(n/\eps))$.
\end{proof}

\subsection{\texorpdfstring   {Proof of \cref{cor:glauber-mixing}}   {Proof of Corollary \ref{cor:glauber-mixing}}}

\Cref{thrm:flip-mixing} established a mixing-time bound of $O(n\log{n})$ for the flip dynamics.  By comparison of the associated Dirichlet forms, Vigoda~\cite{Vigoda} showed that $O(n\log{n})$ mixing time for the flip dynamics with constant-size flips implies $O(n^2\log{n})$ mixing time for the Glauber dynamics.  In fact, since the dependence on $\eps$ in the mixing time is of the form $O(n\log(n/\eps))$, one obtains $O(n^2)$ mixing time for the Glauber dynamics; see \cite[Corollary 2]{DJV} or \cite[Chapter 14 notes]{LPW}.

To obtain $O(n\log{n})$ mixing time of the Glauber dynamics for constant $\Delta$, we utilize the work of Chen, Liu, and Vigoda~\cite{CLV21}, who showed optimal mixing via spectral independence.  For an introduction to spectral independence, see~\cite{SV:notes}.
The following proof outlines how \cref{cor:glauber-mixing} follows from \cref{thrm:list-local-contraction} using the method outlined in \cite{BCCPSV22}.

\begin{proof}[Proof of \cref{cor:glauber-mixing}]

To prove \cref{cor:glauber-mixing}, we first observe that the flip dynamics is contractive under any pinning when $\Delta\geq\deltabound$ and $k \geq \lambbound \Delta$, where a pinning is a fixed assignment for an arbitrary subset of vertices.
To this end, assume that $\colorList{v}=[k]$ for every $v$.

Let $U\subseteq V$ and let $\tau:U\to[k]$ be any pinning that extends to a proper $k$-coloring of $G$.
Let
$G^\tau:=G[V\setminus U]$
be the residual graph. We retain the original degree bound $\Delta$.
For every $w\in V\setminus U$, let its residual list be the following:
\begin{equation*}
L^\tau(w)
:=
[k]\setminus
\{\tau(z):z\in N_G(w)\cap U\}.
\end{equation*}

Let
\begin{equation*}
r_w:=|N_G(w)\cap U|.
\end{equation*}
Since at most $r_w$ distinct colors are removed from the list of
$w$, we have
\begin{equation*}
\lvert L^\tau(w)\rvert\geq k-r_w.
\end{equation*}
Moreover,
\begin{equation*}
d_{G^\tau}(w)=d_G(w)-r_w.
\end{equation*}
Consequently,
\begin{equation*}
\lvert L^\tau(w)\rvert
\geq k-r_w
=d_{G^\tau}(w)+k-d_G(w)
\geq d_{G^\tau}(w)+k-\Delta.
\end{equation*}
Also, $\lvert L^\tau(w)\rvert\leq k$, and the maximum degree of $G^\tau$
is at most $\Delta$.  Therefore, the residual list-coloring
instance satisfies all of the hypotheses of
\cref{thrm:list-local-contraction}.

Let $n_\tau:=|V\setminus U|$.
The fully pinned instance is trivial, so assume that $n_\tau\geq1$.
The conclusion of \cref{thrm:list-local-contraction} gives, for every
neighboring pair in the residual instance,
\begin{equation*}
n_\tau k\widetilde W
\leq
-10^{-6}
\left(d_{G^\tau}(v^*)+k-\Delta\right).
\end{equation*}
Since $k\geq\lambbound\Delta$,
\begin{equation*}
d_{G^\tau}(v^*)+k-\Delta
\geq
k-\Delta
\geq
\left(1-\frac{1}{\lambbound}\right)k.
\end{equation*}
Hence, setting
\begin{equation*}
\gamma
:=
10^{-6}\left(1-\frac{1}{\lambbound}\right)>0,
\end{equation*}
we obtain
\begin{equation*}
\widetilde W\leq-\frac{\gamma}{n_\tau}.
\end{equation*}
Since the metric of a neighboring pair is at most one, it follows
that
\begin{equation*}
\Exp{\nmet(X_{t+1},Y_{t+1})}
\leq
\left(1-\frac{\gamma}{n_\tau}\right)
\nmet(X_t,Y_t).
\end{equation*}
We next lift this neighboring-pair contraction to arbitrary pairs.
Let $\sigma,\sigma'$ be proper list colorings of the residual instance, and let
\begin{equation*}
    \sigma=\rho_0,\rho_1,\ldots,\rho_m=\sigma'
\end{equation*}
be a shortest path in the metric $\nmet$.
The intermediate list labelings may be improper, but \cref{thrm:list-local-contraction} applies to every neighboring pair of list labelings.
For each $1\leq i\leq m$, use the preceding one-step coupling for $(\rho_{i-1},\rho_i)$.
Combining these couplings gives a joint coupling $(Z_0,\ldots,Z_m)$ such that $Z_i$ has the one-step distribution from $\rho_i$ and every consecutive pair has the chosen coupling.
Therefore,
\begin{align*}
\Exp{\nmet(Z_0,Z_m)}
&\leq \sum_{i=1}^m\Exp{\nmet(Z_{i-1},Z_i)}\\
&\leq \left(1-\frac{\gamma}{n_\tau}\right)
\sum_{i=1}^m\nmet(\rho_{i-1},\rho_i)\\
&=\left(1-\frac{\gamma}{n_\tau}\right)\nmet(\sigma,\sigma').
\end{align*}
Thus, with $\kappa:=1-\gamma/n_\tau$, every pinned flip dynamics is $\kappa$-contractive for all pairs of proper list colorings.
Moreover, \cref{eqn:metric-hamming-equivalence} shows that $\nmet$ is $(1-\eta)^{-1}$-equivalent to Hamming distance, and hence is $2$-equivalent because $\eta<1/2$.

Treat proposals with clusters of size at least seven as selecting the empty block, since they never flip. The maximum selected-block size is therefore $M\leq6$.
If a fixed vertex belongs to a nonempty selected block, then the selected vertex lies within graph distance at most five of it.
There are at most $\Delta^6$ possible selected vertices. Since the vertex is selected uniformly from the $n_\tau$ unpinned vertices, the maximum probability that a fixed vertex belongs to the selected block is at most $D\leq\Delta^6/n_\tau$.
Since
\begin{equation*}
1-\kappa=\frac{\gamma}{n_\tau},
\end{equation*}
we have
\begin{equation*}
\frac{DM}{1-\kappa}=O_\Delta(1).
\end{equation*}
It therefore follows from \cite[Theorem~1.11]{BCCPSV22} that the
coloring distribution is spectrally independent under every
pinning, with a constant depending only on $k$ and $\Delta$.
Together with the standard lower bound on positive conditional
marginals for colorings, \cite[Theorem 1.7]{BCCPSV22} yields an
$O(n\log{n})$ mixing-time bound for the Glauber dynamics when
$k$ and $\Delta$ are constant.

Since every vertex has at least $k-\Delta$ available colors, the
transition kernel of our Glauber dynamics is a convex combination
of the heat-bath kernel, with weight $(k-\Delta)/k$, and another
$\pi$-stationary kernel.
By convexity of relative entropy, the heat-bath entropy-contraction
bound from \cite[Theorem~1.7]{BCCPSV22} therefore transfers with a
loss of at most $k/(k-\Delta)=O(1)$, giving the same $O(n\log n)$
mixing-time bound.
\end{proof}

\begin{remark}
\label[remark]{rem:approximate-counting}
Our proof of \cref{thrm:list-local-contraction} establishes a uniformly contractive coupling for the
flip dynamics on every residual list-coloring instance, including those obtained from arbitrary boundary assignments when edges entirely inside the pinned set are ignored.  By \cref{eqn:metric-hamming-equivalence}, the metric is $2$-equivalent to Hamming distance. Since the dynamics also has bounded-size local updates, the argument of Chen, Feng, Guo, Zhang, and Zou~\cite{CFGZZ} implies coupling independence with a constant depending only on $k$ and $\Delta$.
Together with the standard marginal lower bound for colorings, their deterministic counting theorem therefore yields an FPTAS for counting proper $k$-colorings when $k$ and $\Delta$ are constant and the assumptions of \cref{thrm:flip-mixing} hold.
See the proof of \cite[Corollary~6]{CFGZZ}.
\end{remark}

\section{ {Handling degenerate colors}}
\label{sec:degenerate-colors}

The preceding analysis uses the degenerate-color estimate stated at the end of \cref{subsec:weighted-color-bounds}. We now construct the corresponding coupling and prove that estimate, keeping its additional bookkeeping separate from the main argument.
Throughout this section, fix $v^*\in V$ and $(X_t,Y_t)\in\Omega^2_{v^*}$, with the orientation and notation of \cref{subsec:greedy-coupling}. All expectations are conditional on this fixed pair.

\subsection{Coupling}
\label{subsec:degenerate-color-coupling}

We now complete the coupling definition for $c\in\{c_X,c_Y\}$ and $d_c\geq1$.
As in \cref{subsec:greedy-coupling}, we omit the common factor $1/(nk)$ from the probabilities in this coupling description.
We retain the terms \emph{large cluster} and \emph{small cluster} for the representations using $v^*$ and $u_i$, respectively. In these degenerate cases, a nonempty small cluster contains $u_i$ but may also contain $v^*$.
These correspond to the additional cases $c\in\{c_X,c_Y\}$ discussed by Chen et al.~\cite[Remark~A.2]{CDMPP19}, where shared component representations also require separate bookkeeping for improper neighboring states.
When the relevant cluster flips have distinct representations, the coupling follows the same greedy pattern as the nondegenerate construction in \cref{subsec:greedy-coupling}. A cluster using the color pair $\{c_X,c_Y\}$, however, may be represented in both of the original classes $\mathcal D_{c_X}$ and $\mathcal D_{c_Y}$ before the reassignment. This complication does not occur for neighboring proper list colorings, but it can occur because the path-coupling state space contains improper list labelings. We must therefore account for the reassignment that assigns every shared marginal flip to $\mathcal D_{c_X}$. After this reassignment, two clusters retained in $\mathcal D_{c_X}$ may overlap without either containing the other. By contrast, the retained clusters in $\mathcal D_{c_Y}$ are pairwise disjoint or nested, as in the nondegenerate construction.

First suppose $c=c_X$.
Let $N_{c_X}(v^*)=\{u_1,\ldots,u_{d_{c_X}}\}$.
The cluster $S_{X_t}(v^*,c_X)$ is empty, as is $S_{Y_t}(u_i,c_X)$ for every $i$. Thus, the possibly nonempty flips in $\mathcal D_{c_X}$ are represented by $S_{Y_t}(v^*,c_X)$ in $Y_t$ and $S_{X_t}(u_i,c_Y)$ in $X_t$.
For $1\leq i\leq d_{c_X}$, let
\begin{equation*}
\widehat A=|S_{Y_t}(v^*,c_X)|
\qquad\mbox{and}\qquad
\widehat a_i=|S_{X_t}(u_i,c_Y)|.
\end{equation*}
Set $A=\widehat A$. In particular, $A=0$ if $S_{Y_t}(v^*,c_X)$ is not flippable. As in \cref{subsec:greedy-coupling}, set
\begin{equation*}
a_i=
\begin{cases}
0 & \mbox{if }u_i\in S_{X_t}(u_j,c_Y)\mbox{ for some }j<i,\\
\widehat a_i & \mbox{otherwise}.
\end{cases}
\end{equation*}
Let $a_{\max}=\max_i a_i$, and let $i_{\max}$ be an index where $a_i=a_{\max}$.
We pair the flips of $S_{X_t}(u_{i_{\max}},c_Y)$ and $S_{Y_t}(v^*,c_X)$ as much as their marginal probabilities permit. Let
\begin{equation*}
R_A=\min\{P_A,P_{a_{\max}}\}
\end{equation*}
and, for $1\leq i\leq d_{c_X}$, let
\begin{equation*}
Q_i=P_{a_i}-R_A\indicator({i=i_{\max}}).
\end{equation*}
Thus, $Q_i$ is the residual probability of flipping the corresponding small cluster after the paired move.
In particular, $Q_{i_{\max}}=P_{a_{\max}}-R_A$, so the subtraction of the paired $R_A$-mass is absorbed into the residual probability of the selected small cluster.

When $R_A>0$, both selected clusters are nonempty. If the selected $X_t$-cluster does not contain $v^*$, then it is contained in the $Y_t$-cluster, and the paired move changes the Hamming distance by $A-a_{\max}-2$. If the selected $X_t$-cluster contains $v^*$, then the other positive adjusted $a_i$-clusters partition the vertices of the $Y_t$-cluster outside the selected $X_t$-cluster. Thus, the selected clusters intersect in $A-\sum_{i\neq i_{\max}}a_i$ vertices, and the paired move changes the Hamming distance by $2\sum_{i=1}^{d_{c_X}}a_i-A-a_{\max}$.

Accordingly, set $h_A=0$ if $R_A=0$, and, when $R_A>0$, let
\begin{equation*}
h_A=
\begin{cases}
A-a_{\max}-2 & \mbox{if }v^*\notin S_{X_t}(u_{i_{\max}},c_Y),\\
2\displaystyle\sum_{i=1}^{d_{c_X}}a_i-A-a_{\max}
& \mbox{if }v^*\in S_{X_t}(u_{i_{\max}},c_Y).
\end{cases}
\end{equation*}

The coupling for $c_X$ is as follows:
\begin{enumerate}
    \item With probability $R_A$, couple the flip of $S_{X_t}(u_{i_{\max}},c_Y)$ in $X_t$ with the flip of $S_{Y_t}(v^*,c_X)$ in $Y_t$.
    \item With probability $P_A-R_A$, flip $S_{Y_t}(v^*,c_X)$ in $Y_t$ and no cluster in $X_t$.
    \item For each $1\leq i\leq d_{c_X}$, with probability $Q_i$, flip $S_{X_t}(u_i,c_Y)$ in $X_t$ and no cluster in $Y_t$.
\end{enumerate}
All probabilities in these three steps are nonnegative by the definition of $R_A$. Moreover, the selected small cluster is flipped with total probability $R_A+Q_{i_{\max}}=P_{a_{\max}}$, every other small cluster is flipped with probability $P_{a_i}$, and the large cluster is flipped with total probability $P_A$. Thus, these moves have the correct marginal probabilities.
Accordingly, define
\begin{equation}
\label{def:degenerate-x-z-upper-bound}
Z_{c_X}:=\frac{1}{nk}\left[
h_AR_A+(A-2)(P_A-R_A)
+\sum_{i=1}^{d_{c_X}}a_iQ_i
\right].
\end{equation}
The three terms in this definition bound respectively the paired move, the residual large-cluster flip, and the residual small-cluster flips. Hence, $W_{c_X}\leq Z_{c_X}$.

Now suppose $c=c_Y$, and let $N_{c_Y}(v^*)=\{u_1,\ldots,u_{d_{c_Y}}\}$.
The clusters $S_{Y_t}(v^*,c_Y)$ and $S_{X_t}(u_i,c_Y)$ for all $1 \leq i \leq d_c$ are empty. The remaining candidate flips are represented by $S_{X_t}(v^*,c_Y)$ in $X_t$ and $S_{Y_t}(u_i,c_X)$ in $Y_t$.
Let
\begin{equation*}
\widehat B=|S_{X_t}(v^*,c_Y)|
\qquad\mbox{and}\qquad
\widehat b_i=|S_{Y_t}(u_i,c_X)|
\end{equation*}
for all $1 \leq i \leq d_c$.
The reassignment is implemented by setting
\begin{equation*}
B=
\begin{cases}
0 & \mbox{if }(X_t,S_{X_t}(v^*,c_Y))\in\mathcal D_{c_X},\\
\widehat B & \mbox{otherwise},
\end{cases}
\end{equation*}
and
\begin{equation*}
b_i=
\begin{cases}
0 & \mbox{if }(Y_t,S_{Y_t}(u_i,c_X))\in\mathcal D_{c_X},\\
0 & \mbox{if }u_i\in S_{Y_t}(u_j,c_X)\mbox{ for some }j<i,\\
\widehat b_i & \mbox{otherwise}.
\end{cases}
\end{equation*}

If the corresponding unadjusted size is zero, then the candidate cluster is not flippable and has zero marginal flip probability. Otherwise, an adjusted size of zero means that the representation has either been removed by reassignment or already counted through an earlier neighbor; the marginal flip is retained under the other representation.
In all cases, the cluster is flipped with the desired marginal probability.

Let $b_{\max}=\max_i b_i$, and let $i'_{\max}$ be an index where $b_i=b_{\max}$.
No retained cluster $S_{Y_t}(u_i,c_X)$ contains $v^*$, since such a flip is also represented by $S_{Y_t}(v^*,c_X)$ in $\mathcal D_{c_X}$ and is removed from $\mathcal D_{c_Y}$ by the reassignment.
If an original candidate cluster $S_{Y_t}(u_i,c_X)$ contains $v^*$, an alternating path from $u_i$ to $v^*$ enters $v^*$ through a $c_X$-colored neighbor. In $X_t$, this shows that $S_{X_t}(v^*,c_Y)$ is also represented through that neighbor in $\mathcal D_{c_X}$, so the reassignment sets $B=0$.
Consequently, if $B>0$, no nonempty candidate small-cluster flip is removed by the reassignment, and every retained small cluster is contained in $S_{X_t}(v^*,c_Y)$. After repeated representations are removed, the positive $b_i$-clusters therefore partition $S_{X_t}(v^*,c_Y)\setminus\{v^*\}$, and
\begin{equation*}
B=1+\sum_{i=1}^{d_{c_Y}}b_i.
\end{equation*}
Thus, $P_B\leq P_{b_{\max}}$. This also holds when $B=0$ since $P_0=0$.
For $1\leq i\leq d_{c_Y}$, let
\begin{equation*}
Q'_i=P_{b_i}-P_B\indicator({i=i'_{\max}}).
\end{equation*}
Thus, $Q'_i$ is the residual probability of flipping the corresponding small cluster after the paired move.
The coupling for $c_Y$ is therefore as follows:
\begin{enumerate}
    \item With probability $P_B$, couple the flip of $S_{X_t}(v^*,c_Y)$ in $X_t$ with the flip of $S_{Y_t}(u_{i'_{\max}},c_X)$ in $Y_t$.
    \item For each $1\leq i\leq d_{c_Y}$, with probability $Q'_i$, flip $S_{Y_t}(u_i,c_X)$ in $Y_t$ and no cluster in $X_t$.
\end{enumerate}
Define
\begin{equation}
\label{def:degenerate-y-z-upper-bound}
Z_{c_Y}:=\frac{1}{nk}\left[
(B-b_{\max}-2)P_B
+\sum_{i=1}^{d_{c_Y}}b_iQ'_i
\right].
\end{equation}
The selected small cluster is flipped with total probability $P_B+Q'_{i'_{\max}}=P_{b_{\max}}$, and all other retained flips also have their original marginal probabilities. It follows directly from the coupling that $W_{c_Y}\leq Z_{c_Y}$.

There is sufficient self-loop probability for the one-sided moves, and assigning the remaining probability to the double self-loop completes the coupling.

Combining the two cases, for every $c\in[k]$ such that $d_c\geq1$ we have
\begin{equation}
\label{eqn:w-c-z-c-bound}
W_c\leq Z_c.
\end{equation}

\subsection{Bounding the contributions}
\label{subsec:degenerate-color-bound}

We now prove \cref{lem:degenerate-color-sum-bound} using the adjusted
cluster sizes from \cref{subsec:degenerate-color-coupling}.

\begin{proof}[Proof of \cref{lem:degenerate-color-sum-bound}]
For $c\in\{c_X,c_Y\}$, every vertex in $N_{t,c}(v^*)$ initially receives one of the disagreement colors and is therefore multiblocked with respect to $v^*$. In particular, $\degree{0}{t,c}{v^*}=0$. Since $\eta\geq0$, both terms subtracted from $W_c$ in \cref{def:per-color-weighted-change} are nonnegative. Therefore,
\begin{equation*}
    \widetilde W_c\leq W_c\leq Z_c,
\end{equation*}
where the second inequality is \cref{eqn:w-c-z-c-bound}.
It remains to bound $Z_c$ for each of the two degenerate colors.
We use the cluster notation from \cref{subsec:degenerate-color-coupling}.
For each degenerate color, the estimate $(nk)Z_c\leq d_c+P_2$ suffices
when $d_c\geq2$. We then handle $d_c=1$ directly.

We first bound $Z_{c_X}$ when $d_{c_X}\geq1$.
Expanding \cref{def:degenerate-x-z-upper-bound} and the definition of $Q_i$ gives
\begin{equation*}
    (nk)Z_{c_X}
    =\sum_i a_iP_{a_i}+(A-2)P_A
    +R_A\bigl(h_A-A+2-a_{\max}\bigr).
\end{equation*}
If $R_A=0$, the last term vanishes. Otherwise, if the selected small
cluster does not contain $v^*$, the coefficient of $R_A$ is $-2a_{\max}$.
If it contains $v^*$, then this coefficient is
$2(\sum_{i\neq i_{\max}}a_i-A+1)\leq0$, since the intersection described in \cref{subsec:degenerate-color-coupling} contains $v^*$.
Thus, $(nk)Z_{c_X}\leq\sum_i a_iP_{a_i}+(A-2)P_A\leq d_{c_X}+P_2$, where we use \cref{eqn:j-pj-bound,eqn:shifted-pj-bound,eqn:line:p2-bound}.
When $d_{c_X}\geq2$, \cref{eqn:line:basic-flip-bounds} therefore gives
\begin{equation*}
    (nk)Z_{c_X}\leq-1+(2-P_2)d_{c_X}.
\end{equation*}
When $d_{c_X}=1$, the same bound follows directly from the following cases.
If $a_1=0$, then $R_A=0$ and
\begin{equation*}
    (nk)Z_{c_X}=(A-2)P_A.
\end{equation*}
This is zero when $A=0$, and otherwise it is at most $2P_4\leq1-P_2$ by \cref{eqn:shifted-pj-bound,eqn:line:p2-p4-bound}.
We may therefore assume that $a_1>0$.
If the selected small cluster does not contain $v^*$, then $A=a_1+1$ and
$(nk)Z_{c_X}=a_1P_{a_1}-(a_1+1)P_{a_1+1}\leq1-P_2$ by \cref{eqn:weighted-gap-bound}.
If it contains $v^*$, then $2\leq A\leq a_1$ and
\begin{equation*}
    (nk)Z_{c_X}
    =(a_1-A)P_{a_1}+(A-2)(P_A-P_{a_1})
    \leq4P_4
    \leq1-P_2
\end{equation*}
by \cref{eqn:shifted-pj-bound,eqn:line:p2-p4-bound}.

We next bound $Z_{c_Y}$ when $d_{c_Y}\geq1$.
Expanding \cref{def:degenerate-y-z-upper-bound} gives
\begin{equation*}
    (nk)Z_{c_Y}
    =\sum_i b_iP_{b_i}+(B-2b_{\max}-2)P_B.
\end{equation*}
If $B>0$, this is at most $d_{c_Y}+P_2$; if $B=0$, it is at most $d_{c_Y}$.
Consequently, $(nk)Z_{c_Y}\leq-1+(2-P_2)d_{c_Y}$ whenever $d_{c_Y}\geq2$.
For $d_{c_Y}=1$ and $B>0$, we have $B=b_1+1$, and direct substitution gives
$b_1P_{b_1}-(b_1+1)P_{b_1+1}\leq1-P_2$ by \cref{eqn:weighted-gap-bound}.
The same bound holds when $d_{c_Y}=1$, $B=0$, and either $b_1=0$ or $b_1\geq2$, since then
$b_1P_{b_1}\leq2P_2\leq1-P_2$ by \cref{eqn:line:geometric-decay,eqn:line:basic-flip-bounds}.

Finally, $d_{c_Y}=1$ and $b_1=1$ force $B=2$.
Indeed, $S_{Y_t}(u_1,c_X)=\{u_1\}$, so $u_1$ has no $c_X$-colored
neighbor away from $v^*$. Since $u_1$ is the only $c_Y$-colored neighbor
of $v^*$, it follows that $S_{X_t}(v^*,c_Y)=\{v^*,u_1\}$;
both vertices have both colors in their lists.
This cluster contains no $c_X$-colored neighbor of $v^*$, so its flip
has no representation in $\mathcal D_{c_X}$ and is not removed by
reassignment. Thus the case $d_{c_Y}=1$, $B=0$, $b_1=1$ cannot occur.

We have proved $(nk)Z_c\leq-1+(2-P_2)d_c$ for each $c\in\{c_X,c_Y\}$ with $d_c\geq1$.
This proves \cref{eqn:degenerate-color-sum-bound}.
\end{proof}

\section{Conclusions}

The major open question is to obtain a substantial improvement over \Cref{thrm:flip-mixing} by establishing rapid mixing of the flip dynamics (or any other dynamics) for general graphs when $k>\alpha_1\Delta$ for a constant $\alpha_1<1.8$.  If we restrict attention to triangle-free graphs the best known rapid mixing result holds for $k>\alpha_2\Delta$ where $\alpha_2\approx 1.763\ldots$ is the solution of $\alpha_2=\exp(1/\alpha_2)$ using spectral independence~\cite{FGYZ21,CGSV21} (or, for girth at least five, by first allowing the chain to approach locally typical behavior and then using near-uniform local color availability~\cite{DFHV}).  Can we utilize triangle-freeness to achieve similar bounds using a modified metric as in this paper?  It would be interesting to see if the threshold $\alpha_2$ is only an obstacle for certain proof techniques that use properties of the stationary distribution, or if it corresponds to the onset of worst-case mixing obstacles for local Markov chains on locally dense graphs.

We note that \cite{CDMPP19} also utilize an alternative approach using a variable-length coupling (e.g., see \cite{HV-variable}); see \cite{bordewich2005metric} for some connections between the metric-based approach and the variable-length coupling approach.  It would be interesting to obtain a similar $k\geq\lambbound\Delta$ result using a variable-length coupling approach.

\section*{AI Use Disclosure}
\phantomsection
\label{sec:ai-use-disclosure}

During the final revision of this work, the authors used OpenAI’s ChatGPT to assist with exposition, notation, and reference checks, as well as the review of algebraic calculations, proof arguments, and case coverage. The authors have read and understood all proofs and statements and take responsibility for the final manuscript.

\bibliographystyle{alpha}
\bibliography{main}

\end{document}